\RequirePackage{amsmath} 
\documentclass[runningheads]{llncs}
\usepackage[firstpage]{draftwatermark} 
\SetWatermarkAngle{0}
\SetWatermarkText{\raisebox{11.5cm}{
\hspace{-0.25cm}
\href{https://doi.org/10.5281/zenodo.15313204}{\includegraphics{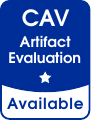}}
\hspace{8.25cm}
\includegraphics{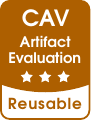}
}}

\usepackage[T1]{fontenc}
\usepackage{xcolor}

\usepackage[bookmarks,bookmarksopen,bookmarksdepth=2]{hyperref}
\usepackage{amssymb}	
\usepackage{graphicx} 
\usepackage{multicol} 
\usepackage{multirow}
\usepackage[capitalize,noabbrev]{cleveref}

\usepackage{bm}

\usepackage{subcaption}
\usepackage{wrapfig}
\usepackage[utf8]{inputenc}
\usepackage[english]{babel}

\usepackage{diagbox}
\usepackage{mathtools}
\usepackage{ytableau}
\usepackage{semantic}
\usepackage{pifont}

\newcommand{\inlinemat}[1]{\ensuremath{\left(\begin{smallmatrix}#1\end{smallmatrix}\right)}}

\DeclarePairedDelimiter\ket{\lvert}{\rangle}
\DeclarePairedDelimiter\bra{\langle}{\rvert}

\newcommand{\tr} {\operatorname{tr}}

\newcommand{\supp} {\operatorname{supp}}
\newcommand{\spanspace} {\operatorname{span}}

\newcommand{\ketbra}[2]{\ensuremath{\ket{#1}\!\bra{#2}}}

\renewcommand{\orcidID}[1]{\smash{\href{http://orcid.org/#1}{\protect\raisebox{-1.25pt}{\protect\includegraphics{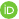}}}}}
\newcommand{\mvs}{\fontfamily{mvs}\fontencoding{U}\fontseries{m}\fontshape{n}\selectfont}
\newcommand\Letter{{\mvs\char66}}

\newcommand{\meas}{\ensuremath{\mathtt{measure}}}
\newcommand{\ifelse}[3]{\ensuremath{\mathtt{if}\ #1\ \mathtt{then}\ #2\ \mathtt{else}\ #3}}
\newcommand{\repeatuntil}[2]{\ensuremath{\mathtt{repeat}\  #1 \  \mathtt{until}\  #2}}
\newcommand{\cqconfig}[1]{\ensuremath{\langle#1\rangle}}
\newcommand{\cqprog}{\textit{cq-prog}}

\AtEndEnvironment{proof}{\phantom{}\qed\endproof}

\begin{document}
\title{Verifying Fault-Tolerance of Quantum Error Correction Codes}
%
%
\author{Kean Chen\inst{1}\textsuperscript{(\Letter)}\orcidID{0000-0002-0772-6635}
\and
Yuhao Liu\inst{1}\orcidID{0009-0005-2822-0448}
\and
Wang Fang\inst{2}\orcidID{0000-0001-7628-1185}
\and
Jennifer Paykin\inst{3} 
\and \\
Xin-Chuan Wu\inst{4}
\and
Albert Schmitz\inst{3} 
\and
Steve Zdancewic\inst{1}\textsuperscript{(\Letter)}\orcidID{0000-0002-3516-1512} 
\and
Gushu Li\inst{1}\textsuperscript{(\Letter)}\orcidID{0000-0002-6233-0334}}
%
\authorrunning{K. Chen et al.}
%
\institute{University of Pennsylvania, Philadelphia, USA \\ \email{\{keanchen,stevez,gushuli\}@seas.upenn.edu} \and University of Edinburgh, Edinburgh, UK \and
Intel Corporation, Hillsboro, USA \and
Intel Corporation, Santa Clara, USA}
%
\maketitle              
\setcounter{footnote}{0}

\begin{abstract}
Quantum computers have advanced rapidly in qubit count and gate fidelity. However, large-scale fault-tolerant quantum computing still relies on quantum error correction code (QECC) to suppress noise. Manually or experimentally verifying the fault-tolerance property of complex QECC implementation is impractical due to the vast error combinations. This paper formalizes the fault-tolerance of QECC implementations within the language of quantum programs. By incorporating the techniques of quantum symbolic execution, we provide an automatic verification tool for quantum fault-tolerance. We evaluate and demonstrate the effectiveness of our tool on a universal set of logical operations across different QECCs.

\end{abstract}

\section{Introduction}

Quantum computers have the potential to address many significant problems that are hard to tackle on classical computers~\cite{nielsen2010quantum,preskill2023quantumcomputing40years}. However, real-world quantum computers suffer from noise, and practical quantum computation relies on the quantum error correction codes (QECCs)~\cite{gottesman1997stabilizer,gottesman2009introqec} to correct the errors caused by the noises and protect the quantum information. 
Recently, prototyping QECCs have been experimentally demonstrated on various devices of different quantum computing technologies~\cite{lukin2023logical,acharya2024quantumerrorcorrectionsurface,acharya2022suppressingquantumerrorsscaling,reichardt2024logicalcomputationdemonstratedneutral}.

A critical property that a QECC must satisfy is \textit{fault-tolerance}. That is, a QECC should be able to correct a certain amount of errors even if the physical operations to implement the QECC are themselves noisy.
Proving fault-tolerance is required for every newly designed QECC. However, hand-written proofs~\cite{gottesman2009introqec,gottesman2024surviving} are only feasible for low-distance and few-qubit codes by propagating the state under combinations of errors.
Such proofs soon become overwhelming as more advanced and sophisticated QECCs~\cite{leverrier2022quantumtannercodes,Breuckmann2021qldpc,tamiya2024polylog} involving more physical qubits and physical operations are proposed.

Computer-aided techniques are well-suited for verifying the fault-tolerance of QECCs. However, current verification tools on quantum programs either 1) fail to capture the stabilizer formalism of QECCs~\cite{tao2021gleipnir,hung2019quantitative,chen2024automatic,gu2022giallar,shi2020certiq,guan2024measurement}, leading to significant scalability issues, or 2) do not consider execution faults~\cite{nan2023quantumsymbolicexecution,fang2024symbolic,bauer2023symqv,fang2024symphase}, making it hard for them to accommodate verification for fault-tolerance.

In particular, we identify several key challenges for mechanically proving the fault-tolerance of a QECC. \textbf{(1)} Quantum computing is intrinsically analog with \textit{continuous errors}, which are hard to reason about on digital classical computers. \textbf{(2)} The implementations of QECC often involve complex control flow, such as \textit{loops}, which introduce non-monotonicity in the transitions, making both semantics and error analysis less tractable. \textbf{(3)} QECCs contain non-Clifford logical operations for universality, which usually cannot benefit from the stabilizer formalism for efficient processing.

In this paper, we overcome these challenges and propose an automated verification tool for QECC fault-tolerance.
\textbf{First,} we extend the semantics of classical-quantum programming language~\cite{ying2010flowchart,ying2024practical} under the presence of quantum errors, where we can describe and formalize the criterion of fault-tolerance for QECC components.
\textbf{Second,} we show that the continuous errors in the input quantum states and during the execution can both be \textit{discretized} to certain input states and Pauli errors. This allows us to develop a symbolic execution engine to prove QECC fault-tolerance via stabilizer formalism.  
\textbf{Third,} we observe that the loops in typical QECCs have unique properties, \textit{memory-less/conservative}, which can enable new symbolic transition rules to overcome the complexity caused by loops.
\textbf{Fourth,} to handle the non-Clifford components, we propose a two-party framework that absorbs non-Clifford components in the inputs. Then, we show that the execution with only Clifford components is fault-tolerant for arbitrary inputs and thus prove the overall fault-tolerance.

We have implemented a symbolic execution framework to prove the fault-tolerance of QECCs and used it to check the fault-tolerance of the essential functional components (i.e., state preparation, gate, measurement, error correction) of various QECCs of different sizes. For example, we are able to prove the fault-tolerance of a two-logical-qubit gate in Toric code with up to 100 physical qubits in about 3 hours and the fault-tolerance of sophisticated state preparation protocol in surface code with 49 physical qubits in about 70 hours. 
When the checking fails, our framework can show the error propagation path that violates the fault-tolerance criterion to help debug the QECC implementation.

The major contributions of this paper can be summarized as follows:
\begin{enumerate}
 \vspace{-5pt}
    \item We extend and define the semantics of quantum programs with faulty executions to reason about the influence of faults and formalize the fault-tolerance properties of QECCs.
    \item We overcome the challenges of mechanically proving the QECC fault-tolerance by introducing a series of theorems to discretize the continuous errors, designing new transition rules by leveraging the unique properties of loops in QECCs, and a two-party framework to handle non-Clifford components.
    \item We implement a symbolic execution framework based on our theories and show that it can prove the fault-tolerance of various QECCs or indicate the error propagation path that violates the fault-tolerance criterion.
\end{enumerate}

\section{Background}

\subsection{Quantum Computing Basics}

The state space of a quantum system is described by a $d$-dimensional complex Hilbert space $\mathcal{H}^d$. A \textbf{pure state} of such a system is a unit vector $\ket{\psi}\in\mathcal{H}^d$. A \textbf{mixed state} is a statistical ensemble of pure states $\{\ket{\psi_i}\}$ with probability $p_i$, described by the \textbf{partial density operator} $\rho=\sum_i p_i\ketbra{\psi_i}{\psi}_i$. If \(\sum_{i} p_i=1\), then \(\rho\) is also called a \textbf{density operator}. A pure state \(\ket{\psi}\) can also be regarded as the the density operator $\rho=\ketbra{\psi}{\psi}$. Suppose \(\rho=\sum_i\lambda_i\ketbra{\psi_i}{\psi_i}\) is the spectral decomposition of a density operator \(\rho\) such that \(\lambda_i>0\), then the \textbf{support} of \(\rho\) is defined as  \(\supp(\rho)=\spanspace(\{\ket{\psi_i}\}_i)\) and is denoted as \textbf{supp}\(\bm{(\rho)}\). 

We use \textbf{quantum channel} to describe a general quantum evolution, which is mathematically a completely positive~\cite{choi1975completely,jamiolkowski1972linear,nielsen2010quantum} trace-preserving map \(\mathcal{E}\) on density operators. Specifically, \(\mathcal{E}\) has the non-unique form \(\mathcal{E}(\rho)=\sum_iE_i\rho E_i^\dagger\) such that $\sum_iE_i^\dagger E_i=I$. The operators $\{E_i\}$ are called Kraus operators.
In particular, a \textbf{quantum gate} corresponds to a unitary quantum evolution that can be described by a unitary channel:
\(\mathcal{E}(\rho)=U\rho U^\dag\), where \(U\) is a unitary matrix (i.e., \(U^\dagger U = I\)). \textbf{Qubit} (quantum bit) has a 2-dimensional Hilbert space $\mathcal{H}^2$ as its state space and $\{\ket{0},\ket{1}\}$ as the \textbf{computational basis}. The state space of $n$ qubits is the tensor product of all state spaces of each qubit.

A \textbf{quantum measurement} is described by a set of linear operators $\{M_i\}$ such that $\sum_i M_i^\dagger M_i=I$. When measuring a state $\rho$, the probability of outcome $i$ is $p_i=\tr(M_i\rho M_i^\dagger)$, after which the state collapses to $\rho'=M_i\rho M_i^\dagger/p_i$.

\subsection{Pauli Operator and Stabilizer Formalism}\label{sec-5191611}

The following unitaries are called the Pauli matrices:
\(X=\inlinemat{0 & 1 \\1 & 0}, Y=\inlinemat{0 & -i \\ i & 0}, Z=\inlinemat{1 & 0 \\ 0 & -1}\).
An \(\bm{n}\)\textbf{-qubit Pauli operator} $P$ is a tensor product of $n$ Pauli matrices $\sigma_i\in \{I=\inlinemat{1 & 0 \\ 0 & 1},X,Y,Z\}$ with a global phase $c\in\{1,-1,i,-i\}$: 
\(P=c\cdot \sigma_1\otimes\dots\otimes \sigma_n.\)
The \textbf{weight of the Pauli operator} \(P\) is the number of indices \(i\) such that \(\sigma_i\neq I\). Note that the set of all $n$-qubit Pauli operators forms a group \(\mathcal{P}_n\) under matrix multiplication. 
The \textbf{vector representation}~\cite{gottesman2004simulation,nielsen2010quantum} of the unsigned Pauli operator $\sigma_1\otimes\cdots \otimes \sigma_n$ is a $2n$-length vector $[x_1,\ldots,x_n, z_1,\ldots, z_n]$ where $x_i,z_i\in \{0,1\}$ and the pair $(x_i,z_i)$ corresponds to $(0,0),(0,1),(1,0),(1,1)$ representing $\sigma_i = I,Z,X,Y$, respectively.

A state \(\ket{\psi}\) is stabilized by a unitary $U$ if $U\ket{\psi}=\ket{\psi}$ and $U$ is called the stabilizer of $\ket{\psi}$.
Suppose $\{P_1,\dots,P_m\}$ is a set of $m$ commuting and independent\footnote{Each \(P_j\) cannot be written as a product of others.} $n$-qubit Pauli operators such that \(P_i\neq \pm I\) and $P_i^2 \neq-I$. These Pauli operators generate a group \(\mathcal{G}=\langle P_1,\ldots, P_m\rangle\) by matrix multiplication. Suppose $V$ is the subspace containing the states stabilized by elements in \(\mathcal{G}\), i.e., \(V:=\{\ket{\psi}\ |\ \forall P\in \mathcal{G}, P\ket{\psi}=\ket{\psi}\}\), then \(\mathcal{G}\) is called the \textbf{stabilizer group} of \(V\), and the group elements \(P\in\mathcal{G}\) are called the \textbf{stabilizers} of \(V\). Particularly, $P_1,\dots,P_m$ are called the \textbf{generators}.
Note that $V$ is a $2^{n-m}$-dimensional subspace. In particular, if $m=n$, then $V$ is a one-dimensional subspace.

An \(n\)-qubit unitary \(U\) is called an \textbf{Clifford} unitary if it maps Pauli operators to Pauli operators under conjugation, i.e., 
\(\forall P\in \mathcal{P}_n,\,\, UPU^\dagger\in \mathcal{P}_n\). 
All \(n\)-qubit Clifford unitaries can be generated by the $H$, $S$, and $\textup{CNOT}$ gates~\cite{gottesman2004simulation,nielsen2010quantum}.
There is a special type of quantum circuit called the \textbf{stabilizer circuits}, which contains only the following elements: 1) preparing qubits into the computational basis; 2) applying Clifford gates; 3) measuring qubits in the computational basis.
Stabilizer circuits are vastly utilized in quantum error correction and can be simulated efficiently by classical computers~\cite{gottesman1998heisenberg}.

\vspace{-10pt}\subsection{Quantum Error Correction}

Quantum computers in the real world suffer errors that will significantly impact the outcome.
Due to the no-cloning theorem~\cite{nielsen2010quantum}, it is impossible to duplicate quantum states to resist errors. Instead, we employ quantum error correction codes to encode quantum information, which can detect and correct errors.
A \textbf{quantum error correction code} (QECC) $\mathcal{C}$ uses \(n\) physical qubits ($\mathcal{H}^{2^n}$) to encode \(k\) logical qubits ($\mathcal{H}^{2^k}$). Specifically, it assigns a $2^k$-dimensional subspace \(\mathcal{H}'\subseteq \mathcal{H}^{2^n}\), called the \textbf{code space}, and an isomorphism between \(\mathcal{H}'\) and \(\mathcal{H}^{2^k}\). If it is clear from the context, we will use \(\mathcal{C}\) to refer to both the QECC and its code space.
The error correction process of a QECC consists of the following steps: 
\begin{itemize}
\vspace{-5pt}
    \item \textbf{Syndrome measurement}: applying a list of measurements on physical qubits to detect potential errors. 
    \item \textbf{Recovery}: performing recovery operations based on the outcome of the syndrome measurement.
    \vspace{-5pt}
\end{itemize}

The vast majority of QECC are \textbf{stabilizer codes}. For a stabilizer code \(\mathcal{C}\), its code space is the subspace stabilized by a stabilizer group $\mathcal{G}=\langle P_1,\dots,P_m\rangle$, where \(P_1,\ldots,P_m\) are Pauli operators. Note that the code space is of dimension \(2^{n-m}\). Therefore, it encodes \(k=n-m\) logical qubits. The error syndromes can be obtained by performing Pauli measurement of the generators \(P_1,\ldots,P_m\). 
A Pauli error $E$ that anticommutes with a generator can be detected by the measurement of this generator. 
The code distance $d$ is defined as the smallest integer $t$ such that there exists a Pauli operator of weight \(t\) acting as a non-trivial logical operation on the code space \(\mathcal{C}\).

Throughout this paper, the \textbf{code parameters} of stabilizer codes are denoted as \(\bm{[[n,k,d]]}\), where \(n,k\) and \(d\) stand for the number of physical qubits, number of encoded logical qubits and code distance, respectively.

\vspace{-10pt}\subsection{Error Propagation and Fault-Tolerance}

We first introduce the following definition.
\begin{definition}\label{def-1311407}
The \(\bm{r}\)\textbf{-error space surrounding} \(\ket{\psi}\) is defined as
\[\mathcal{S}_r(\ket{\psi}) \coloneqq \spanspace(\{P\ket{\psi} \,\, |\,\,  P \textup{ is a Pauli operator of weight at most } r\}).\]
\end{definition} 
We say that a quantum state \(\rho\) has at most \(\bm{r}\) \textbf{errors w.r.t.} \(\bm{\vert\psi\rangle}\) if \(\supp(\rho) \subseteq \mathcal{S}_r(\ket{\psi})\). We may ignore \(\ket{\psi}\) if it is clear from the context. 
We can further generalize the notion of error space for any pure state $\ket{\psi}\in \mathcal{H}_1\otimes \cdots\otimes \mathcal{H}_m$ in a composite system by defining:
\[\mathcal{S}_{r_1,\ldots,r_m}(\ket{\psi})\coloneqq \spanspace(\{P\ket{\psi}\,|\, P\,\,\textup{is a Pauli of weight at most}\,\, r_i\,\, \textup{on block}\,\, \mathcal{H}_i\}).\]
We say $\rho$ has at most $\bm{r_i}$ \textbf{errors on block} $\bm{\mathcal{H}_i}$ \textbf{w.r.t.} $\bm{\vert\psi\rangle}$ if $\supp(\rho)\subseteq \mathcal{S}_{r_1,\ldots,r_m}(\ket{\psi})$.

\begin{figure}[ht]
     \begin{subfigure}[b]{0.48\linewidth}
     \centering
     \includegraphics[width=1.0\linewidth]{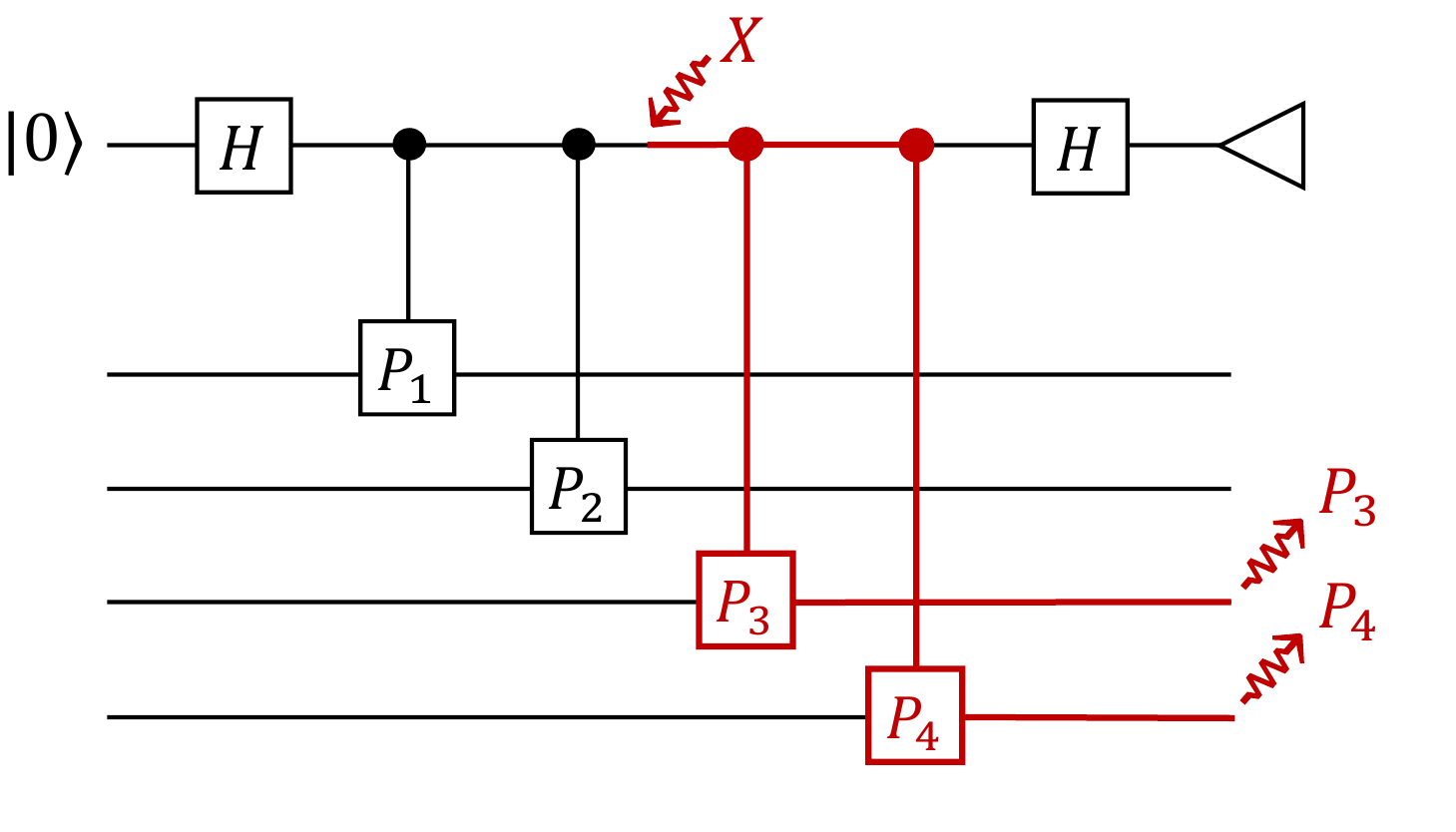}
     \vspace{-15pt}
     \caption{Multi-qubit Pauli measurement using a single ancilla, in which a Pauli-\(X\) error on the ancilla qubit will propagate to two Pauli errors \(P_3, P_4\) on the data qubits.}
     \label{fig-multi-Pauli-sing-anc-1181605}
     \end{subfigure}
     \hfill
     \begin{subfigure}[b]{0.4\linewidth}
     \centering
     \includegraphics[width=1.0\linewidth]{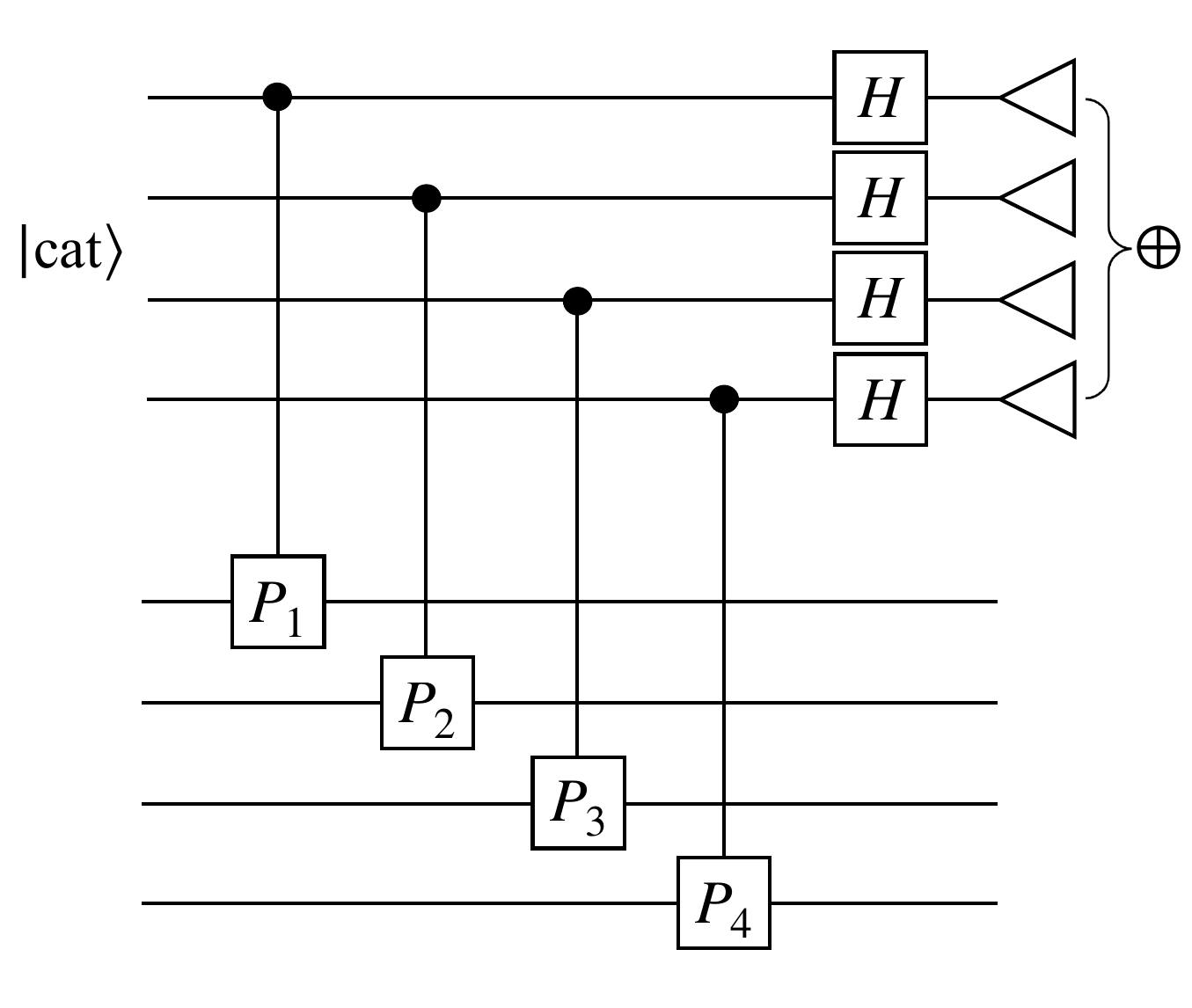}
          \vspace{-15pt}
     \caption{Multi-qubit Pauli measurement using the cat state.}
     \label{fig-multi-Pauli-cat-1181649}
     \end{subfigure}
      \vspace{0mm}
     \caption{Two implementations of multi-qubit Pauli measurement \(P_1\otimes P_2\otimes P_3\otimes P_4\).}
     \vspace{-3mm}
\end{figure}

While QECCs effectively detect existing errors on quantum states, the error correction process and logical operations can introduce new errors during executions. Even worse, errors may propagate among the qubits. For example, \cref{fig-multi-Pauli-sing-anc-1181605} shows an error propagation in an implementation of a multi-qubit Pauli measurement (which corresponds to a logical Pauli measurement). In contrast, \cref{fig-multi-Pauli-cat-1181649} shows a ``transversal'' implementation, where the ancilla qubits are initialized to the cat state \(\ket{\textup{cat}}=(\ket{0000}+\ket{1111})/\sqrt{2}\). Here, a single error will not propagate to two or more errors, which, roughly speaking, is called \textbf{fault-tolerance}~\footnote{The formal definition of fault-tolerance is given in \cref{sec:fault-tolerance}}. 
However, being entangled, the cat state itself cannot be prepared transversally, requiring a more elaborate design (see \cref{sec-cat-state-1212245}) for fault-tolerance. 
This leads to a fundamental question:

    \vspace{2pt}
\noindent\fbox{%
    \parbox{0.95\textwidth}{%
        \textit{Given a QECC and an implementation of a logical operation on it, how can we tell if it is fault-tolerant?}
    }
}   \vspace{2pt}

It turns out that manually checking the fault-tolerance of QECC implementation is highly non-trivial because the prover needs to iterate over all possible error combinations and propagation paths. 
In the rest of this paper, we will introduce our new theory and tools to formalize and verify the fault-tolerance of QEC code implementations.

\vspace{-3pt}\section{Related Work}

\textbf{Formal logic for quantum computing.} There are extensive works on building formal logic to reason quantum programs, including quantum assertions~\cite{li2020projection}, quantum Hoare logic~\cite{ying2012floyd,ying2019formalverfihoare,ying2021hoareclavar}, the logic for erroneous quantum programs~\cite{tao2021gleipnir,hung2019quantitative}, ZX-calculus~\cite{robert2022equivalence,robert2023equiv}, mechanized verification~\cite{shi2021symbolic,ying2023coqq,ying2019formalverfihoare,paykin2017qwire,hietala2020proving}, etc. However, most of these works require human interaction, and none focus on fault-tolerance. In contrast, our work formally characterizes QEC fault-tolerance and provides an automated verification tool.

\noindent\textbf{Automatic testing and verification of quantum programs.} Other works focus on proving properties of quantum programs based on certain automatic techniques, including quantum symbolic execution~\cite{nan2023quantumsymbolicexecution,fang2024symbolic,bauer2023symqv,fang2024symphase}, equivalence checking of quantum programs~\cite{Burgholzer2021equiv,Burgholzer2022nonunitary,ying2022equivseq,hong2021approximate}, program invariance generation~\cite{wu2017invariants,Hung2019OnTT}, automatic testing and verification~\cite{chen2024automatic,gu2022giallar,shi2020certiq,guan2024measurement}, etc. However, these tools are not designed to target the fault-tolerance of QECCs.

\vspace{-8pt}\section{Formalism of Quantum Fault-Tolerance}\label{sec:fault-tolerance}

Before verifying the fault-tolerance, we first formalize the quantum fault-tolerance by extending the language of classical-quantum programs~\cite{ying2010flowchart,ying2024practical} with faulty transitions and then formalizing the QEC fault-tolerance criteria.

\subsection{Classical-Quantum Program (\cqprog)}

Note that QEC protocols contain not only unitary circuits but also include mid-circuit measurements with sophisticated classical control flows (e.g., classical oracles, branches, and repeat-until-success). Therefore, it is necessary to adapt a language with classical components. In this paper, we select the classical-quantum program~\cite{ying2012floyd,fang2024symbolic,ying2024practical,ying2010flowchart}.

\begin{definition}[Syntax of Classical-Quantum Programs (\cqprog)]\label{def:syntax-qprog}

A classical-quantum program (\cqprog) on qubits $q_1,\dots,q_n$ with classical variables and \textit{repeat-until-success} structure is defined by the following syntax:\vspace{-5pt}
\begin{align*}
S ::=&\ S_1;S_2 & (\textit{sequence})\quad &|\ q:=\ket{0} & (\textit{initialization}) \\
|&\ U(\overline{q}) & (\textit{unitary})\quad &|\ x:=\meas\ q & (\textit{measurement}) \\
|&\ x:=e & (\textit{assign})\quad &|\ \ifelse{b}{S_1}{S_2} & (\textit{conditionals})\\
|&\ y:=f(x) & (\textit{classical oracle})\quad&|\ \repeatuntil{S}{b} &  (\textit{repeat-until})
\end{align*}
\vspace{-20pt}
\end{definition}
The notations used in \cref{def:syntax-qprog} are explained as follows:
$q$ is a qubit, $x$ and $y$ are classical variables, $e$ is an expression on classical variables, $f$ is a classical oracle, $\overline{q} = q_i,\dots,q_j$ is an arbitrary list of distinct qubits; 
$U(\overline{q})$ performs a unitary gate \(U\) on a list of qubits $\overline{q}$; 
$x:=\meas\ q$ performs a computational-basis measurement $\{\ketbra{0}{0},\ketbra{1}{1}\}$ on the qubit $q$ and stores the result ($0$ or $1$) in the classical variable $x$; 
$\repeatuntil{S}{b}$ performs the quantum program $S$ until the Boolean expression $b$ is evaluated as true.

Then, the execution of a classical-quantum program can be described in terms of the transition between the classical-quantum configurations.
\begin{definition}[Classical-Quantum Configuration]\label{def:config-qprog}
A classical-quantum configuration is a triple $\cqconfig{S,\sigma,\rho}$, where $S$ is a \cqprog, $\sigma:x\mapsto v$ is a map from classical variables to classical values, and $\rho$ is a partial density operator describing the (sub-normalized) quantum state of the qubits.

\end{definition}

\begin{figure}[ht]
    \centering
    \vspace{0pt}
    \begin{gather*}
(\textup{IN}) \quad \cqconfig{q:=\ket{0},\sigma,\rho}\rightarrow\cqconfig{\downarrow,\sigma,\rho_{q:=\ket{0}}}\quad\quad 
(\textup{UT}) \quad \cqconfig{U(\overline{q}),\sigma,\rho}\rightarrow\cqconfig{\downarrow,\sigma,U_{\overline{q}}\rho U_{\overline{q}}^\dagger}\\
(\textup{M0}) \quad \cqconfig{x:=\meas\ q,\sigma,\rho}\rightarrow\cqconfig{\downarrow,\sigma[0/x],\ket{0}_{q}\bra{0}\rho\ket{0}_{q}\bra{0}} \\
(\textup{M1}) \quad \cqconfig{x:=\meas\ q,\sigma,\rho}\rightarrow\cqconfig{\downarrow,\sigma[1/x],\ket{1}_{q}\bra{1}\rho\ket{1}_{q}\bra{1}}\\
(\textup{AS})\quad \cqconfig{x:=e,\sigma,\rho}\rightarrow\cqconfig{\downarrow,\sigma[\sigma(e)/x],\rho}\quad\quad (\textup{CO})\quad \cqconfig{y:=f(x),\sigma,\rho}\rightarrow\cqconfig{\downarrow,\sigma[f(x)/y],\rho}\\
(\textup{SC})\quad \inference{\langle S_1,\sigma,\rho\rangle\rightarrow \langle S_1',\sigma',\rho'\rangle}{\langle S_1;S_2,\sigma,\rho\rangle\rightarrow \langle S_1';S_2,\sigma',\rho'\rangle}\\
(\textup{CT})\quad \inference{\sigma\models b}{\langle \ifelse{b}{S_1}{S_2},\sigma,\rho\rangle\rightarrow \langle S_1,\sigma,\rho\rangle}\\
(\textup{CF})\quad \inference{\sigma\models \neg b}{\langle \ifelse{b}{S_1}{S_2},\sigma,\rho\rangle\rightarrow \langle S_2,\sigma,\rho\rangle}\\
\textup{(RU)} \quad \langle \repeatuntil{S}{b}, \sigma, \rho\rangle \rightarrow \langle S; \ifelse{b}{\downarrow}{\{\repeatuntil{S}{b}\}}, \sigma, \rho\rangle 
\end{gather*}
    \vspace{-15pt}
    \caption{Rules for the ideal transition relation. In rule \textup{(IN)}, \(\rho_{q:=\ket{0}}= \ket{0}_q\bra{0}\rho\ket{0}_q\bra{0}+\ket{0}_q\bra{1}\rho\ket{1}_q\bra{0}\) where we use $\ket{i}_q$ to denote the pure state $\ket{i}$ at qubit $q$ and $\ket{i}_q\bra{j}$ is the abbreviation of product of $\ket{i}_q$ and $\,_q\bra{j}$; in rule \textup{(UT)}, \(U_{\overline{q}}\) means a unitary that acts as \(U\) on \(\overline{q}=q_i\ldots q_j\), and acts trivially on other qubits. }
    \vspace{-10pt}
    \label{fig:idealtrans-qprog}
\end{figure}

We can now define the operational semantics of a \cqprog{} by specifying the transition relationship between the classical-quantum configurations.
The rules in Fig.~\ref{fig:idealtrans-qprog} define the ideal transition relation ``$\rightarrow$'' between classical-quantum configurations.

\subsection{Execution with Faults}
We now inject faults into ideal executions.
We assume that all the classical operations (i.e., (AS), (CO), (CT), (CF), and (RU)) are flawless, and only quantum operations (i.e., (IN), (UT), (M0), and (M1)) can be faulty. When a fault happens, the original quantum operation is replaced by an arbitrary quantum operation acting on the same qubits. The corresponding transition is called faulty transition.

\begin{definition}[Faulty Transition]\label{def:faulttrans-qprog}
The following rules define the faulty transition ``$\leadsto$'' between classical-quantum configurations when faults occur during the execution of a quantum program:\vspace{-5pt}
\begin{gather*}
(\textup{F-IN}) \quad \langle q:=\ket{0} ,\sigma,\rho\rangle \leadsto \langle \downarrow, \sigma,\mathcal{E}_q(\rho)\rangle \quad\quad\,\,
(\textup{F-UT})\quad \langle U(\overline{q}),\sigma,\rho\rangle \leadsto \langle \downarrow,\sigma, \mathcal{E}_{\overline{q}}(\rho) \rangle\\
(\textup{F-M0})\quad \langle x:=\meas\,\, q,\sigma,\rho\rangle\leadsto \langle \downarrow, \sigma[0/x] , \mathcal{E}_{0,q}(\rho)\rangle\\
(\textup{F-M1})\quad \langle x:=\meas\,\, q,\sigma,\rho\rangle\leadsto \langle \downarrow, \sigma[1/x] , \mathcal{E}_{1,q}(\rho)\rangle
\end{gather*}
where in rules \textup{(F-IN)} and \textup{(F-UT)}, \(\mathcal{E}_q\) and \(\mathcal{E}_{\overline{q}}\) can be arbitrary trace-preserving quantum operations that act non-trivially only on \(q\) and \(\overline{q}\), respectively; in rules \textup{(F-M0)} and \textup{(F-M1)}, \(\mathcal{E}_{0,q}\) and \(\mathcal{E}_{1,q}\) can be arbitrary (non-trace-preserving) quantum operations~\footnote{It is worth noting that \(\textup{(F-M0)}\) and \(\textup{(F-M1)}\) also cover the case of bit-flip error on the measurement outcome, since it is equivalent to letting \(\mathcal{E}_{0,q}(\rho)=\ket{1}_q\bra{1}\rho\ket{1}_q\bra{1}\) and \(\mathcal{E}_{1,q}=\ket{0}_{q}\bra{0}\rho\ket{0}_{q}\bra{0}\).} acting non-trivially only on \(q\). We omit the rule \(\textup{(F-SC)}\) here as it is simply analogous to the rule \(\textup{(SC)}\) in \cref{fig:idealtrans-qprog}.
\end{definition}

Now, we can describe the execution of a \cqprog\  with faults using the following faulty transition tree.
\begin{definition}[Transition Tree with Faults]\label{def-faultytree}
Let \(S_0\) be a \cqprog, \(\sigma_0\) be a classical state and \(r\geq 0\) be an integer. An \(r\)-fault transition tree \(\mathcal{T}\) starting with \((S_0,\sigma_0)\) is a (possibly infinite) tree with classical-quantum configurations as nodes and satisfies the following conditions:
\begin{enumerate}
\vspace{-3pt}
    \item The root is \(\langle S_0,\sigma_0,\widetilde{\rho}\rangle\), where \(\widetilde{\rho}\) is an indeterminate quantum state.
    \item A node is of the form \(\langle \downarrow,\sigma,\rho\rangle\) iff it is a leaf node.
    \item A node is of the form \(\langle x:=\meas \,\, q ;S,\sigma,\rho\rangle\) iff it has two children, either generated by \textup{(M0)} and \textup{(M1)}, or generated by \((\textup{F-M0})\) and \((\textup{F-M1})\). Furthermore, for the latter case, the quantum operations \(\mathcal{E}_{0,q}\) and \(\mathcal{E}_{1,q}\) in the transition rules satisfy that \(\mathcal{E}_{0,q}+\mathcal{E}_{1,q}\) is trace-preserving. 
    \item If a node $v$ is not of the form \(\langle \downarrow,\sigma,\rho\rangle\) or \(\langle x:=\meas \,\, q ;S,\sigma,\rho\rangle\), then it has only one child $v'$ such that $v\rightarrow v'$ or $v\leadsto v'$ (cf. Fig. \ref{fig:idealtrans-qprog} and Def. \ref{def:faulttrans-qprog}).
    \item Each path contains at most \(r\) faulty transitions ``\(\leadsto\)''.
    \vspace{-3pt}
\end{enumerate}
Let $\rho_0$ be a concrete quantum state. We use $\mathcal{T}(\rho_0)$ to denote the tree obtained by replacing the indeterminate quantum state $\widetilde{\rho}$ with $\rho_0$ in the root node.
\end{definition}

If we fix a \cqprog\ \(S\) and a classical state \(\sigma_0\), then any \(r\)-fault transition tree \(\mathcal{T}\) starting with \((S_0,\sigma_0)\) corresponds to a well-defined quantum channel.
\begin{proposition}\label{prop-1250405}
Let \(S_0\) be a quantum program, \(\sigma_0\) be a classical state and \(r\geq 0\). If \(\mathcal{T}\) is an \(r\)-fault transition tree starting with \((S_0,\sigma_0)\), then the map \(\mathcal{E}_{\mathcal{T}}\): \(\rho\mapsto \sum_{\langle\downarrow,\sigma,\rho'\rangle \in\mathcal{T}(\rho)}\rho'\)
is a quantum channel (not necessarily trace-preserving).
\end{proposition}
The proof can be found in \cref{sec:proof-prop1}. 

Note that the semantics of a \cqprog{}\ in the presence of faults is nondeterministic. This is because there can be infinite many \(r\)-fault transition trees starting with \((S_0,\sigma_0)\), where each transition tree corresponds to a valid semantics with \(r\) faults.
In this context, each tree $\mathcal{T}$ is also referred to as an $\bm{r}$\textbf{-fault instantiation} of \((S_0,\sigma_0)\), or simply an instantiation of \(S_0\) with \(r\) faults, if \(\sigma_0=\emptyset\).

\subsection{Formalization of Quantum Fault-Tolerance}

Given the definition of quantum gadgets and \cqprog, we can formally define quantum fault-tolerance developed in previous QEC theory~\cite{gottesman2009introqec,gottesman2024surviving}.
Different from their definitions, we separate the fault-tolerance property of a QEC program from its semantical correctness in the ideal case, and we further assume that a QEC program always produces the correct semantics in the ideal case.

A QECC usually consists of four types of \textbf{gadgets}:
(1). \textbf{Preparation}: prepare a logical state, e.g., the logical state $\ket{\overline{0}}$.
(2). \textbf{Measurement}: measure on a logical basis, e.g., the logical \(Z\)-basis measurement \(\{\ketbra{\overline{0}}{\overline{0}},\ketbra{\overline{1}}{\overline{1}}\}\).\footnote{A bar on a ket (pure state) denotes the corresponding logical state w.r.t. a QECC.}
(3). \textbf{Gate}: perform a logical quantum gate, e.g., the logical CNOT gate.
(4). \textbf{Error correction}: detect and correct physical errors in the logical state.
We provide some circuit implementation examples of the gadgets in \cref{sec:713-cc}. 

In the rest this section, we fix a QECC \(\mathcal{C}\) with distance \(d\) and set \(t\coloneqq \lfloor \frac{d-1}{2} \rfloor\) as the \textbf{number of correctable errors}.
The following symbols will be used:
\begin{itemize}
\vspace{-5pt}
    \item ``$\begin{picture}(17,10)\thicklines
\put(1,1.1){\line(1,0){15}}
\put(1,2.95){\line(1,0){15}}
\thicklines\end{picture}$'' stands for a logical block of $\mathcal{C}$ formed by physical qubits.
    \item $\begin{picture}(40,12)
\thicklines
\put(1,0){\(\ket{\psi}\)}
\put(15,1){\line(1,0){5}}
\put(15,3){\line(1,0){5}}
\put(27,-5){\line(1,1){7}}
\put(27,-5){\line(-1,1){7}}
\put(27,9){\line(1,-1){7}}
\put(27,9){\line(-1,-1){7}}
\put(27,1){\makebox(0,0){$r$}}
\put(33,1){\line(1,0){5}}
\put(33,3){\line(1,0){5}}
\thicklines\end{picture}$ stands for an arbitrary quantum mixed state having at most $r$ errors w.r.t. $\ket{\psi}$ (see Def. \ref{def-1311407}). When multiple $\begin{picture}(15,12)
\thicklines
\put(7,-5){\line(1,1){7}}
\put(7,-5){\line(-1,1){7}}
\put(7,9){\line(1,-1){7}}
\put(7,9){\line(-1,-1){7}}
\put(7,1){\makebox(0,0){$r_i$}}
\thicklines\end{picture}$ are applied on multiple blocks, they together denote a quantum mixed state having at most $r_i$ errors on the $i$-th block, respectively. 
    \item $\begin{picture}(17,18)
\thicklines
\put(11,2){\oval(20,20)[l]}
\put(11,-8){\line(0,1){20}}
\put(9.5,2){\oval(13,15)[l]}
\put(9.5,-5.5){\line(0,1){15}}
\put(11,1){\line(1,0){5}}
\put(11,3){\line(1,0){5}}
\end{picture}\vspace{2mm}$, $\begin{picture}(17,18)
\thicklines
\put(1,1){\line(1,0){5}}
\put(1,3){\line(1,0){5}}
\put(6,2){\oval(20,20)[r]}
\put(6,-8){\line(0,1){20}}
\put(7.5,2){\oval(13,15)[r]}
\put(7.5,-5.5){\line(0,1){15}}
\end{picture}$, $\begin{picture}(24,18)
\thicklines
\put(1,9){\line(1,0){5}}
\put(1,7){\line(1,0){5}}
\put(1,-3){\line(1,0){5}}
\put(1,-5){\line(1,0){5}}
\put(12,2){\oval(12,25)}
\put(12,2){\oval(8,21)}
\put(18,9){\line(1,0){5}}
\put(18,7){\line(1,0){5}}
\put(18,-3){\line(1,0){5}}
\put(18,-5){\line(1,0){5}}
\end{picture}$ and  $\begin{picture}(32,18)
\thicklines
\put(1,1){\line(1,0){5}}
\put(1,3){\line(1,0){5}}
\put(6,-8){\framebox(20,20){\small\text{\!EC}}}
\put(8,-6){\framebox(16,16){}}
\put(26,1){\line(1,0){5}}
\put(26,3){\line(1,0){5}}
\end{picture}$
stand for the preparation, measurement, (two-qubit) gate, and error correction gadgets, respectively.
    \vspace{1pt}\item ``$s$'' near a gadget indicates $s$ faults appear in the execution of the gadget. The gadget is ideal (fault-free) if no number is present near it.  
    \vspace{-2pt}
\end{itemize}

Then, we introduce the formalization of fault-tolerance. We only present the two-qubit case for gate gadgets, which can be easily extended to \(n\)-qubit cases.
\begin{definition}[Fault-Tolerance of QECC Gadgets]\label{def-ft-1310455}

1. A \textbf{state preparation gadget} is fault-tolerant if the following holds:
\begin{equation}\label{eq-FT-prepare-12312218}
\begin{picture}(250,10)
\thicklines

\put(100,2){\oval(20,20)[l]}
\put(100,-8){\line(0,1){20}}
\put(98.5,2){\oval(13,15)[l]}
\put(98.5,-5.5){\line(0,1){15}}

\put(100,1){\line(1,0){15}}
\put(100,3){\line(1,0){15}}
\put(103,7){\makebox(0,0)[bl]{$s$}}

\put(130,0){$=$}

\put(160,2){\oval(20,20)[l]}
\put(160,-8){\line(0,1){20}}
\put(158.5,2){\oval(13,15)[l]}
\put(158.5,-5.5){\line(0,1){15}}

\put(160,1){\line(1,0){11}}
\put(160,3){\line(1,0){11}}

\put(177,-5){\line(1,1){7}}
\put(177,-5){\line(-1,1){7}}
\put(177,9){\line(1,-1){7}}
\put(177,9){\line(-1,-1){7}}
\put(177,1){\makebox(0,0){$s$}}

\put(183,1){\line(1,0){10}}
\put(183,3){\line(1,0){10}}
\put(198,0){$.$}

\put(0,0){$\textup{whenever } s \leq t,$}
\end{picture}
\end{equation}
That is, a \cqprog{} \(S\) is a FT preparation gadget if for any $s$-fault instantiation of \(S\), s.t. \(s\leq t\), the prepared state has at most \(s\) errors.

2. A \textbf{two-qubit logical gate gadget} is fault-tolerant if the following holds:\vspace{-10pt}
\begin{equation}\label{eq-FT-gate-12312216}
\begin{picture}(290,30)
\thicklines

\put(100,-3){$\ket{\overline{\psi}}$}

\put(115,11){\line(1,0){11}}
\put(115,13){\line(1,0){11}}

\put(132,5){\line(1,1){7}}
\put(132,5){\line(-1,1){7}}
\put(132,19){\line(1,-1){7}}
\put(132,19){\line(-1,-1){7}}
\put(132,11){\makebox(0,0){\scriptsize $r_1$}}

\put(138,11){\line(1,0){12}}
\put(138,13){\line(1,0){12}}

\put(115,-11){\line(1,0){11}}
\put(115,-13){\line(1,0){11}}

\put(132,-19){\line(1,1){7}}
\put(132,-19){\line(-1,1){7}}
\put(132,-5){\line(1,-1){7}}
\put(132,-5){\line(-1,-1){7}}
\put(132,-13){\makebox(0,0){\scriptsize $r_2$}}

\put(138,-11){\line(1,0){12}}
\put(138,-13){\line(1,0){12}}

\put(160,0pt){\oval(20,50)}
\put(160,0pt){\oval(16,46)}
\put(160,0pt){\makebox(0,0){$U$}}

\put(170,11){\line(1,0){10}}
\put(170,13){\line(1,0){10}}

\put(170,-11){\line(1,0){10}}
\put(170,-13){\line(1,0){10}}

\put(171,20){\makebox(0,0)[bl]{$s$}}

\put(195,0){$=$}

\put(216,-3){$\ket{\overline{\psi}}$}

\put(230,11){\line(1,0){10}}
\put(230,13){\line(1,0){10}}

\put(230,-11){\line(1,0){10}}
\put(230,-13){\line(1,0){10}}

\put(250,0pt){\oval(20,50)}
\put(250,0pt){\oval(16,46)}
\put(250,0pt){\makebox(0,0){$U$}}

\put(260,11){\line(1,0){11}}
\put(260,13){\line(1,0){11}}

\put(277,5){\line(1,1){7}}
\put(277,5){\line(-1,1){7}}
\put(277,19){\line(1,-1){7}}
\put(277,19){\line(-1,-1){7}}
\put(277,11){\makebox(0,0){\scriptsize $e_1$}}

\put(283,11){\line(1,0){10}}
\put(283,13){\line(1,0){10}}

\put(260,-11){\line(1,0){11}}
\put(260,-13){\line(1,0){11}}

\put(277,-19){\line(1,1){7}}
\put(277,-19){\line(-1,1){7}}
\put(277,-5){\line(1,-1){7}}
\put(277,-5){\line(-1,-1){7}}
\put(277,-13){\makebox(0,0){\scriptsize $e_2$}}

\put(283,-11){\line(1,0){10}}
\put(283,-13){\line(1,0){10}}
\put(300,0){$,$}

\put(-27,0){\textup{whenever}\, $r_1+r_2+s \leq t,$}
\end{picture}
\vspace{5mm}
\end{equation}
where $e_1,e_2= r_1+r_2+s$.
That is, a \cqprog{} \(S\) is a FT gate gadget if for any input logical state \(\ket{\overline{\psi}}\) with \(r_1\) and \(r_2\) errors on each logical block and any $s$-fault instantiation of \(S\), s.t. \(r_1+r_2+s\leq t\), the output state has at most \(r_1+r_2+s\) errors on each block.

3. A \textbf{measurement gadget} is fault-tolerant if the following holds: 
\begin{equation}\label{eq-FT-meas-12312216}
\begin{picture}(250,10)
\thicklines

\put(90,0){$\ket{\overline{\psi}}$}

\put(105,1){\line(1,0){11}}
\put(105,3){\line(1,0){11}}
\put(122,-5){\line(1,1){7}}
\put(122,-5){\line(-1,1){7}}
\put(122,9){\line(1,-1){7}}
\put(122,9){\line(-1,-1){7}}
\put(122,1){\makebox(0,0){$r$}}

\put(128,1){\line(1,0){11}}
\put(128,3){\line(1,0){11}}

\put(139,2){\oval(20,20)[r]}
\put(139,-8){\line(0,1){20}}
\put(140.5,2){\oval(13,15)[r]}
\put(140.5,-5.5){\line(0,1){15}}

\put(150,5){\makebox(0,0)[bl]{$s$}}

\put(170,0){$=$}

\put(190,0){$\ket{\overline{\psi}}$}

\put(205,1){\line(1,0){10}}
\put(205,3){\line(1,0){10}}

\put(215,2){\oval(20,20)[r]}
\put(215,-8){\line(0,1){20}}
\put(216.5,2){\oval(13,15)[r]}
\put(216.5,-5.5){\line(0,1){15}}
\put(228,0){$.$}

\put(-15,0){$\textup{whenever } r+s\leq t,$}
\end{picture}
\end{equation}
That is, a \cqprog{} \(S\) is a FT measurement gadget if for any input logical state \(\ket{\overline{\psi}}\) with \(r\) errors and any $s$-fault instantiation of \(S\), s.t. \(r+s\leq t\), the measurement outcomes are the same (in distribution) as those in the ideal case.

4. An \textbf{error correction gadget} 
is fault-tolerant if the following holds:
\begin{equation}\label{eq-FT-EC-12312217}
\begin{picture}(250,10)
\thicklines

\put(80,0){$\ket{\overline{\psi}}$}

\put(95,1){\line(1,0){11}}
\put(95,3){\line(1,0){11}}
\put(112,-5){\line(1,1){7}}
\put(112,-5){\line(-1,1){7}}
\put(112,9){\line(1,-1){7}}
\put(112,9){\line(-1,-1){7}}
\put(112,1){\makebox(0,0){$r$}}

\put(118,1){\line(1,0){11}}
\put(118,3){\line(1,0){11}}

\put(130,-8){\framebox(20,20){\small\text{\!EC}}}
\put(132,-6){\framebox(16,16){}}
\put(150,1){\line(1,0){10}}
\put(150,3){\line(1,0){10}}
\put(153,7){\makebox(0,0)[bl]{$s$}}

\put(178,0){$=$}

\put(200,0){$\ket{\overline{\psi}}$}

\put(215,1){\line(1,0){11}}
\put(215,3){\line(1,0){11}}

\put(232,-5){\line(1,1){7}}
\put(232,-5){\line(-1,1){7}}
\put(232,9){\line(1,-1){7}}
\put(232,9){\line(-1,-1){7}}
\put(232,1){\makebox(0,0){$s$}}

\put(238,1){\line(1,0){11}}
\put(238,3){\line(1,0){11}}

\put(250,0){$.$}

\put(-25,0){$\textup{whenever } r+s\leq t,$}
\end{picture}
\end{equation}
That is, a \cqprog{} \(S\) is a fault-tolerant error correction gadget if for any input logical state \(\ket{\overline{\psi}}\) with \(r\) errors and any $s$-fault instantiation of \(S\), s.t. \(r+s\leq t\), the output state has at most \(s\) errors.
\end{definition}
\section{Discretization}
Because physical quantum errors are analog (e.g., small imprecisions in rotation angle), the fault-tolerance conditions formalized above are defined over continuous error channels and logical states that are difficult to reason about using computer-aided techniques.
It is essential to break these formulations into discrete yet equivalent forms.
In this section, we introduce two new discretization theorems that allow us to analyze and verify quantum fault-tolerance on \cqprog{} in a systematic way.

\subsection{Discretization of Input Space}
We show that, for fault-tolerance verification, the input space can be discretized from continuous code space into discrete logical basis states.
A similar result of input space discretization was given in \cite{fang2024symbolic} for verifying \textit{fault-free} QEC programs. Our result is a further refinement that uses fewer inputs and applies to fault-tolerance verification of various types of gadgets.

\begin{theorem}[Discretization of Input Space]\label{lemma-dis-input-1272245}
\begin{enumerate}
    \item  A gate or error correction gadget is fault-tolerant if and only if it satisfies the corresponding fault-tolerance properties (c.f. Equations \eqref{eq-FT-gate-12312216}, \eqref{eq-FT-EC-12312217}) on both the logical computational basis \(\{\ket{\overline{i}}\}_{i=0}^{2^k-1}\) and the logical state \(|\overline{+}\rangle\coloneqq\sum_{i=0}^{2^k-1} \ket{\overline{i}}/\sqrt{2^k}\), where \(k\) is the number of logical qubits.
    \item A \(Z\)-basis measurement gadget is fault-tolerant if and only if it satisfies the corresponding fault-tolerance properties (c.f. Equation \eqref{eq-FT-meas-12312216}) on the logical computational basis \(\{\ket{\overline{i}}\}_{i=0}^{2^k-1}\), where \(k\) is the number of logical qubits.
\end{enumerate}
\end{theorem}
The proof can be found in \cref{sec:proof-theorem1}. 

\subsection{Discretization of Faults}

Since a general \cqprog{}\ contains complicated control flows that can also be affected by quantum faults (e.g., through outcomes of mid-circuit measurements), the discretization theorem in conventional QEC theory~\cite{gottesman2009introqec,gottesman2024surviving} does not fit our \cqprog{} framework.
We provide our discretization result for \cqprog, showing that faults can be discretized from arbitrary quantum channels into Pauli channels.

\begin{definition}[Transition Tree with Pauli Faults]\label{def-tree-pauli-faults-110250}
Let \(S_0\) be a quantum program, \(\sigma_0\) be a classical state and \(r\geq 0\) be an integer. An \(\bm{r}\)\textbf{-Pauli-fault} transition tree \(\mathcal{T}\) starting with \((S_0,\sigma_0)\) is a transition tree (c.f. \cref{def-faultytree}), such that all faulty transitions in it satisfy the following rules instead:
\begin{gather*}
(\textup{PF-IN})\quad \langle q:=\ket{0} ,\sigma,\rho\rangle \leadsto \langle \downarrow, \sigma,P_q\rho_{q:=\ket{0}}P_q\rangle \\
(\textup{PF-UT})\quad \langle U(\overline{q}),\sigma,\rho\rangle \leadsto \langle \downarrow,\sigma, P_{\overline{q}} U_{\overline{q}}\rho U_{\overline{q}}^\dag P_{\overline{q}} \rangle\\
(\textup{PF-M0}) \quad \langle x:=\meas\,\, q,\sigma,\rho\rangle\leadsto \langle \downarrow, \sigma[0/x] , P_{0,q}\ket{0}_q\bra{0}Q_{q}\rho Q_{q}\ket{0}_q\bra{0}P_{0,q}\rangle\\
(\textup{PF-M1}) \quad \langle x:=\meas\,\, q,\sigma,\rho\rangle\leadsto \langle \downarrow, \sigma[1/x] , P_{1,q}\ket{1}_q\bra{1}Q_{q}\rho Q_{q}\ket{1}_q\bra{1} P_{1,q}\rangle
\end{gather*}
where in rules \((\textup{PF-IN})\) and \((\textup{PF-UT})\), \(P_q\) and \(P_{\overline{q}}\) are arbitrary Pauli operators acting on \(q\) and \(\overline{q}\), respectively. In rules \((\textup{PF-M0})\) and \((\textup{PF-M1})\), \(P_{0,q},P_{1,q},Q_{q}\) are arbitrary Pauli operators acting on \(q\).\footnote{Rules (PF-M0) and (PF-M1) introduce both pre- and post-measurement Pauli errors, aiming to account for all possible error patterns of a measurement operation.}
Furthermore, if a node is \(\langle x:=\meas q; S,\sigma,\rho\rangle\) and its two children are generated by \((\textup{PF-M0})\) and \((\textup{PF-M1})\), then both transitions share the same Pauli error \(Q_q\).
\end{definition}

We call such tree \(\mathcal{T}\) an $r$-Pauli-fault instantiation of \((S_0,\sigma_0)\) and if \(\sigma_0=\emptyset\), we call it an \(r\)-Pauli-fault instantiation of \(S_0\). Then, we have the following theorem.

\begin{theorem}[Discretization of Faults]\label{lemma-disc-fautls-191645}
A gadget is fault-tolerant if and only if it satisfies the corresponding fault-tolerance properties (c.f. Equations \eqref{eq-FT-prepare-12312218}, \eqref{eq-FT-gate-12312216}, \eqref{eq-FT-meas-12312216}, \eqref{eq-FT-EC-12312217}) with the Pauli fault instantiation (see \cref{def-tree-pauli-faults-110250}).
\end{theorem}
The proof can be found in \cref{sec:proof-theorem2}. 

\section{Symbolic Execution with Quantum Faults}\label{sec-QSE-1112059}

Having discretized the inputs and errors, we can now establish a symbolic execution framework based on the quantum stabilizer formalism to reason about fault-tolerance of QECC implementations.
Specifically, we maintain a symbolic configuration \(\langle S,\widetilde{\sigma},\widetilde{\rho},p,\varphi,\widetilde{F}\rangle\), where
\begin{itemize}
 \vspace{-5pt}
    \item \(S\) is the quantum program to be executed,
    \item \(\widetilde{\sigma}\) is a symbolic classical state, 
    \item \(\widetilde{\rho}\) is a symbolic quantum stabilizer state, parameterized by a set of symbols, and is equipped with a set of stabilizer operations (see \cref{sec-sym-stab-state-opt-120009}).
    \item \(p\) is the probability corresponding to the current execution path, which is, in fact, a concrete number instead of a symbolic expression.
    \item \(\varphi\) is the path condition, 
    expressing the assumptions introduced in an execution path of the \cqprog{}. 
    \item \(\widetilde{F}\) is a symbolic expression that records the number of faults occurring during the execution of faulty transitions. 
     \vspace{-0pt}
\end{itemize}

\subsection{Symbolic Stabilizer States and Error Injection}\label{sec-sym-stab-state-opt-120009}

\begin{definition}[Symbolic Stabilizer States]\label{def-1281938}
For any commuting and independent set \(\{P_1,\ldots, P_n|P_i\neq \pm I, P_i^2\neq -I\}\) of \(n\)-qubit Pauli operators and Boolean functions \(g_1,\ldots, g_n\) over \(m\) Boolean variables, the symbolic stabilizer state \(\widetilde{\rho}(s_1,\ldots,s_m)\) over symbols \(s_1,\ldots,s_m\) is a stabilizer state of
\[\langle (-1)^{g_1(s_1,\ldots,s_m)}P_1,\ldots,(-1)^{g_n(s_1,\ldots,s_m)}P_n \rangle,\]
where \(\langle A_1,\ldots, A_n \rangle\) denotes the group generated by \(A_1,\ldots, A_n\).
\end{definition}

When it is clear from the context, we will omit the symbols \(s_1,\ldots,s_m\) using \(\widetilde{\rho}\) to denote \(\widetilde{\rho}(s_1,\ldots,s_m)\), 
and directly use \(\langle A_1,\ldots, A_n\rangle\) to denote the stabilizer state of the group generated by \(A_1,\ldots, A_n\).

There are four types of symbolic stabilizer operations: 
\begin{enumerate}
 \vspace{-5pt}
    \item A symbolic function \(\textup{IN}(q_i,\widetilde{\rho})\) for the initialization statement \(q_i\coloneqq \ket{0}\).
    \item A symbolic function \(\textup{UT}(U,\overline{q},\widetilde{\rho})\) for the Clifford unitary transform statement \(U(\overline{q})\).
    \item A  symbolic function \(\textup{M}(q_i,\widetilde{\rho})\) for the measurement statement \(x\coloneqq \meas\  q_i\).
    \item A symbolic error injection function \(\textup{EI}(\overline{q},\widetilde{\rho})\) for the faulty transitions.
     \vspace{-5pt}
\end{enumerate}
The symbolic functions \(\textup{IN}(q_i,\widetilde{\rho})\), \(\textup{UT}(U,\overline{q},\widetilde{\rho})\) return the symbolic quantum states resulting from the corresponding quantum operations and \(\textup{M}(q_i,\widetilde{\rho})\) returns a triplet comprising the measurement outcome, probability and post-measurement state. They are similarly to those in \cite{fang2024symbolic} and are available in \cref{sec:symbolic-rule}. 

Here we introduce our new symbolic error injection function \(\textup{EI}(\overline{q},\widetilde{\rho})\). Suppose \(\widetilde{\rho}\) is a symbolic stabilizer state:
\begin{equation}\label{eq-150251}
\widetilde{\rho}=\langle (-1)^{g_1(s_1,\ldots,s_m)}P_1,\ldots,(-1)^{g_n(s_1,\ldots,s_m)}P_n \rangle.
\end{equation}
The symbolic error injection function acts on \(\widetilde{\rho}\) by injecting symbolic Pauli errors. It returns a symbol recording the activation of those Pauli errors and the symbolic state \(\widetilde{\rho}\) after error injection. Specifically, it is defined as follows.


\begin{definition}[Symbolic Error Injection]
The symbolic Pauli error injection function \(\textup{EI}(q_i,\widetilde{\rho})\) is defined as
\[\textup{EI}(q_i,\widetilde{\rho})=(e_X\lor e_Z, 
  \langle (-1)^{g_1\oplus  e_X\cdot c^X_{1}\oplus e_Z\cdot c^Z_1}P_1,\ldots,(-1)^{g_n\oplus e_X\cdot c^X_n\oplus e_Z\cdot c^Z_n} P_n \rangle),\]
where \(e_X\) and \(e_Z\) are newly introduced symbols recording the inserted Pauli \(X\) and \(Z\) errors; \(c_i^X=0\) if \(P_i\) commutes with \(X_i\) and $c_i^X=1$ otherwise (similarly for $c_i^Z$).

If the symbolic Pauli error injection function is applied on a sequence of qubit variables \(\overline{q}=q_{j_1},\ldots,q_{j_a}\), then \(\textup{EI}(\overline{q},\widetilde{\rho})\) is defined as
\[\textup{EI}(\overline{q},\widetilde{\rho})=(e_1\lor \cdots \lor e_a,\widetilde{\rho}^{(a)}),\]
where \((e_{i},\widetilde{\rho}^{(i)})=\textup{EI}(q_{j_i},\widetilde{\rho}^{(i-1)})\) for \(i=1,\ldots,a\) and \(\widetilde{\rho}^{(0)}\coloneqq \widetilde{\rho}\).
\end{definition}

Note that in the definition of error injection on a sequence of qubits, we set the fault counter as $e_1\lor \cdots\lor e_a$ rather than $e_1+\cdots+e_a$. This is because a single faulty multi-qubit operation may introduce errors on multiple qubits but we count it as only one fault.

\subsection{Symbolic Faulty Transitions}
With the new error injection rule, we can define the symbolic faulty transitions, which are a symbolization of the faulty transitions in \cref{def-tree-pauli-faults-110250}.

\begin{definition}[Symbolic Faulty Transitions]\label{def-symb-faut-tran-172111}
The following are the symbolic faulty transition rules on the symbolic configurations:\vspace{-3pt}
\begin{gather*}
(\textup{SF-IN}) \quad \cqconfig{q_i:=\ket{0},\widetilde{\sigma},\widetilde{\rho},p,\varphi,\widetilde{F}}\twoheadrightarrow\cqconfig{\downarrow,\widetilde{\sigma},\widetilde{\rho}', p, \varphi,\widetilde{F}+e}\\
\textup{where}\quad (e, \widetilde{\rho}')=\textup{EI}(q_i,\textup{IN}(q_i,\widetilde{\rho}))\\
(\textup{SF-UT}) \quad \cqconfig{U(\overline{q}),\widetilde{\sigma},\widetilde{\rho},p,\varphi,\widetilde{F}}\twoheadrightarrow\cqconfig{\downarrow,\widetilde{\sigma},\widetilde{\rho}',p,\varphi,\widetilde{F}+e}\\
\textup{where}\quad (e, \widetilde{\rho}')=\textup{EI}(\overline{q},\textup{UT}(U,\overline{q},\widetilde{\rho}))\\
(\textup{SF-M}) \quad \cqconfig{x:=\meas\ q_i,\widetilde{\sigma},\widetilde{\rho},p,\varphi,\widetilde{F}}\twoheadrightarrow\cqconfig{\downarrow,\widetilde{\sigma}[s/x],\widetilde{\rho}^{(3)},p p',\varphi,\widetilde{F}+e_1 \lor e_2} \\
\textup{where}\,\,\, (e_2, \widetilde{\rho}^{(3)})=\textup{EI}(q_i,\widetilde{\rho}^{(2)}),\,\,\, (s,p',\widetilde{\rho}^{(2)})=\textup{M}(q_i,\widetilde{\rho}^{(1)}),\,\,\, (e_1, \widetilde{\rho}^{(1)})=\textup{EI}(q_i,\widetilde{\rho})\\
(\textup{SF-AS})\quad \cqconfig{x:=e,\widetilde{\sigma},\widetilde{\rho},p,\varphi,\widetilde{F}}\twoheadrightarrow\cqconfig{\downarrow,\widetilde{\sigma}[\widetilde{\sigma}(e)/x],\widetilde{\rho},p,\varphi,\widetilde{F}}\\
(\textup{SF-CO})\quad \cqconfig{y:=f(x),\widetilde{\sigma},\widetilde{\rho},p,\varphi,\widetilde{F}}\twoheadrightarrow\cqconfig{\downarrow,\widetilde{\sigma}[s/y],\widetilde{\rho},p,\varphi\land C_f(\widetilde{\sigma}(x),s),\widetilde{F}}\\
(\textup{SF-CT})\quad \langle \ifelse{b}{S_1}{S_2},\widetilde{\sigma},\widetilde{\rho},p,\varphi,\widetilde{F}\rangle\twoheadrightarrow \langle S_1,\widetilde{\sigma},\widetilde{\rho},p,\varphi\land \widetilde{\sigma}(b),\widetilde{F}\rangle\\
(\textup{SF-CF})\quad \langle \ifelse{b}{S_1}{S_2},\widetilde{\sigma},\widetilde{\rho},p,\varphi,\widetilde{F}\rangle\twoheadrightarrow \langle S_2,\widetilde{\sigma},\widetilde{\rho},p,\varphi\land \neg \widetilde{\sigma}(b),\widetilde{F}\rangle\\
\textup{(SF-RU)} \quad\langle \repeatuntil{S}{b}, \widetilde{\sigma}, \widetilde{\rho}, p,\varphi, \widetilde{F}\rangle \twoheadrightarrow \langle S; S', \widetilde{\sigma}, \widetilde{\rho}, p, \varphi, \widetilde{F}\rangle\\
\textup{where}\,\, S'=\ifelse{b}{\downarrow}{\{\repeatuntil{S}{b}\}}\vspace{-8pt}
\end{gather*}
where in rules \((\textup{SF-IN}), (\textup{SF-UT})\) and \((\textup{SF-M})\), \(e, e_1\) and \(e_2\) are newly introduced bit symbols recording the inserted faults and in rule \((\textup{SF-CO})\), \(s\) is a newly introduced symbol and \(C_f\) is a logical formula asserting the behavior of \(f\). We omit the rule \(\textup{(SF-SC)}\) here as it is simply analogous to the rule \(\textup{(SC)}\) in \cref{fig:idealtrans-qprog}.
\end{definition}

Note that in rule \(\textup{(SF-CO)}\), we do not actually run the classical oracle \(f(x)\), but instead add an assertion \(C_f\) for the output of \(f(x)\) to the path condition. 
This is beneficial when the behavior of the classical oracle can be abstracted into a logical formula, such as the decoding algorithms for stabilizer codes.
More details about the assertion of the decoding algorithm can be found in \cref{sec-assert-dec-1312156}. 

\subsection{Repeat-Until-Success}\label{sec-rep-unt-suc-1232229}

Analyzing the repeat-until-success loop, also known as the do-while loop, presents significant challenges. The symbolic transition rule \(\textup{(SF-RU)}\) shown in \cref{def-symb-faut-tran-172111} does not perform well in practice since it can lead to numerous execution paths or even render the symbolic execution non-terminable.

To circumvent this issue, we observe that loops employed in QEC gadgets typically adhere to specific patterns. By leveraging these patterns, we propose the following symbolic transition rule in place of \(\textup{(SF-RU)}\):\vspace{-5pt}
\begin{gather*}
\textup{(SF-RU}'\textup{)}
\quad\inference{\cqconfig{S,\widetilde{\sigma},\widetilde{\rho},p,\varphi,\widetilde{F}}\twoheadrightarrow^{*} \cqconfig{\downarrow,\widetilde{\sigma}',\widetilde{\rho}',p',\varphi',\widetilde{F}'}}{ \langle \repeatuntil{S}{b}, \widetilde{\sigma}, \widetilde{\rho}, p,\varphi, \widetilde{F}\rangle \twoheadrightarrow \langle \downarrow, \widetilde{\sigma}', \widetilde{\rho}', p', \varphi' \land \widetilde{\sigma}'(b), \widetilde{F}'\rangle}\vspace{-5pt}
\end{gather*}
where \(\twoheadrightarrow^{*}\) stands for the transitive closure of \(\twoheadrightarrow\).
Notably, it executes the loop body only once and then post-selects on success by adding the success condition to the path condition.
We will explain its conditional validity in the following and then conclude the overall soundness and completeness in \cref{sec-sound-complete-1232230}.

The first pattern is from the fault-tolerant cat state preparation gadget. Here, a repeat-until-success structure is used to test the prepared state, where in each iteration, all quantum and classical variables are reset before being used again. This pattern is recognized as the \textit{memory-less repeat-until-success}.
\begin{definition}[Memory-Less Repeat-Until-Success]
A repeat-until-success loop ``\(\repeatuntil{S}{b}\)'' is memory-less if all classical and quantum variables used in the loop body \(S\) are reset before being used in each iteration.
\end{definition}

For the memory-less repeat-until-success structure, it suffices only to consider the last iteration. The intuition is straightforward: suppose $t$ errors are injected into the loop. Errors injected in iteration other than the last one have no effect beyond that iteration. Therefore, the worst case is when all errors are injected into the final iteration. If the fault-tolerance property holds in the worst-case scenario, it also holds in the general case. 

The second pattern witnesses the applicability of \((\textup{SF-RU\('\)})\) in a more general scenario, where qubits are not reset across iterations. Our observation is based on Shor's error correction. Here, syndrome measurements are repeated until consecutive all-agree syndrome results of length \(\lfloor \frac{d-1}{2}\rfloor+1\) are observed. Therefore, qubits are not reset, meaning errors in one iteration can affect subsequent iterations. 
Nevertheless, through closer examination, we observe that syndrome measurements are implemented ``transversally'' upon proper cat state preparation. This means the errors will not propagate. As a result, any errors injected in an early iteration can be commuted to the last iteration without increasing the total number of errors.
Additionally, in the fault-free case, the stabilizer syndrome measurements exhibit idempotent semantics, meaning a single iteration suffices to capture all possible quantum state outputs that would arise from multiple iterations. We refer to this type of loop as a \textit{conservative} repeat-until-success. 

\begin{definition}[Conservative Repeat-Until-Success]\label{def-con-rus-210106}
A repeat-until-success loop ``\(\repeatuntil{S}{b}\)'' is \textit{conservative} if 
\begin{enumerate}
\vspace{-5pt}
    \item the classical variables in \(S\) are reset before used in each iteration,
    \item in the fault-free case, the loop always terminates after a single iteration, and the fault-free semantics of
    the loop body $S$ is non-adaptive and idempotent (see \cref{sec:proof-theorem3} for more details). 
    \vspace{-5pt}
\end{enumerate}
\end{definition}

Consequently, for a conservative loop, the rule \((\textup{SF-RU\('\)})\) remains valid provided that an additional condition is met: errors do not propagate uncontrollably within the loop body \(S\). Fortunately, this can again be verified inside the loop body \(S\) recursively, using our symbolic execution.

\subsection{Verification of Fault-Tolerance}\label{sec-veri-ft-1200400}

Here, we present the overall pipeline of our verification tool for FT properties (i.e., Eq. \ref{eq-FT-prepare-12312218}, \ref{eq-FT-gate-12312216}, \ref{eq-FT-meas-12312216} and \ref{eq-FT-EC-12312217}). For simplicity, we only demonstrate the verification of a gate gadget for a logical gate \(\overline{U}\). Other types of gadgets can be handled similarly.

By \cref{lemma-dis-input-1272245}, it suffices to check whether \cref{eq-FT-gate-12312216} holds for the following two cases: \(\ket{\overline{\psi}}=\ket{\overline{i}}\) for all \(i\) and \(\ket{\overline{\psi}}=\sum_{i=0}^{2^k-1}\ket{\overline{i}}/\sqrt{2^{k}}\). 
Therefore, we construct two symbolic stabilizer states respectively:
\begin{equation}\label{eq-1101828}
\begin{split}
\widetilde{\rho}_Z&=\langle P_1,\ldots, P_{n-k},(-1)^{s_1}\overline{Z}_1,\ldots,(-1)^{s_k}\overline{Z}_k \rangle,\\ \widetilde{\rho}_X&=\langle P_1,\ldots, P_{n-k},\overline{X}_1,\ldots,\overline{X}_k\rangle,
\end{split}
\end{equation}
where \(P_1,\ldots,P_{n-k}\) are the Pauli stabilizers of the QECC \(\mathcal{C}\) and \(\overline{Z}_i, \overline{X}_i\) are the Pauli operators corresponding to the logical \(Z_i\) and logical \(X_i\) in \(\mathcal{C}\), respectively. 
We will only deal with the state \(\widetilde{\rho}_Z\), and the state \(\widetilde{\rho}_X\) can be handled analogously.

Let \(\widetilde{\rho}_0=\widetilde{\rho}_Z\). We apply the Pauli error injection function \(\textup{EI}\) on \(\widetilde{\rho}_0\) recursively for all physical qubits, i.e.,
\((e_i,\widetilde{\rho}_{i})=\textup{EI}(q_i,\widetilde{\rho}_{i-1})\) for \(i=1,\ldots,n\).
Then we define \(\widetilde{\rho}_{\textup{in}}\coloneqq \widetilde{\rho}_n\) and \(\widetilde{F}_{\textup{in}}\coloneqq \sum_{i=1}^n e_i\), which records the number of errors injected in the symbolic state \(\widetilde{\rho}_{\textup{in}}\).
The initial symbolic configuration is set to \(\langle S,\emptyset,\widetilde{\rho}_\textup{in},1,\textup{True}, \widetilde{F}_{\textup{in}} \rangle\), where \(S\) is the \cqprog{} implementing the gate gadget.
By applying our symbolic faulty transitions (see \cref{def-symb-faut-tran-172111}) recursively on the initial configuration, we obtain a set of output configurations \(\textup{cfgs}=\{ \textup{cfg}_i\}_{i}\), where the indices \(i\) refer to different execution paths and each \(\textup{cfg}_i\) is of the form \(\langle \downarrow,\widetilde{\sigma},\widetilde{\rho},p,\varphi,\widetilde{F}\rangle\), in which \(p>0\) is a concrete positive number.

Then, the gadget is fault-tolerant if and only if each symbolic configuration in \(\textup{cfgs}\) describes a set of noisy states that are not ``too far'' from the error-free output state. Specifically, for each \(\langle \downarrow, \widetilde{\sigma},\widetilde{\rho},p,\varphi,\widetilde{F}\rangle \in \textup{cfgs}\), let \(\{s_1,\ldots,s_m\}\) be the set of symbols involved in this configuration (i.e., all symbols occurring in \(\widetilde{\sigma},\widetilde{\rho},\varphi\) and \(\widetilde{F}\)). Then, we need to check whether the following holds
\begin{equation}\label{eq-1100520}
\forall s_1,\ldots,s_m,\,\, (\varphi \land \widetilde{F}\leq t) \rightarrow D(\widetilde{\rho},\widetilde{\rho}_{\textup{ideal}})\leq \widetilde{F},
\end{equation}
where \(\widetilde{\rho}_{\textup{ideal}}\) is the symbolic state obtained by applying the (concrete) ideal unitary transform \(\overline{U}\) on \(\widetilde{\rho}_0\), and \(D(\widetilde{\rho},\widetilde{\rho}_{\textup{ideal}})\) represents the Pauli distance between \(\widetilde{\rho}\) and \(\widetilde{\rho}_{\textup{ideal}}\), i.e., the minimal weight of a Pauli operator that can transform \(\widetilde{\rho}\) to \(\widetilde{\rho}_{\textup{ideal}}\). Note that the computation of \(D(\widetilde{\rho},\widetilde{\rho}_{\textup{ideal}})\) is not straightforward but can still be formulated in a first-order formula. Specifically, let \(P_1,\ldots,P_n\) be the unsigned stabilizers of \(\widetilde{\rho}\), which should also be the unsigned stabilizers of \(\widetilde{\rho}_{\textup{ideal}}\) (otherwise we can direct conclude that the gadget is not fault-tolerant). Suppose
\[\widetilde{\rho}=\langle (-1)^{g_1}P_1,\ldots,(-1)^{g_n}P_n\rangle,\quad \widetilde{\rho}_{\textup{ideal}}=\langle (-1)^{h_1}P_1,\ldots,(-1)^{h_n}P_n\rangle, \]
where  \(g_i\) and \(h_i\) are functions taking symbols \(s_1,\ldots,s_m\) as inputs.
Therefore, \(D(\widetilde{\rho},\widetilde{\rho}_{\textup{ideal}})\leq \widetilde{F} \) can be expressed as
\begin{equation}\label{eq-1100449}
\exists P,\, \textup{wt}(P) \leq \widetilde{F} \,\land\,\bigwedge_{i=1}^n h_i\oplus c(P,P_i)=g_i, 
\end{equation}
where \(P\) is a symbolic Pauli represented by a symbolic vector of length \(2n\) (c.f. the vector representation in \cref{sec-5191611}), \(\textup{wt}(P)\) is the weight of \(P\) and \(c(P,P_i)=0\) if \(P\) and \(P_i\) commutes and \(1\) otherwise. 

Note that the existence quantifier is, in fact, applied on \(2n\) variables. We can further reduce it to \(n\) variables. To see this, note that the computation of \(c(P,P_i)\) for \(i=1,\ldots,m\) involves multiplication (over GF(2)) of the matrix representing the stabilizers with the vector representing \(P\). Specifically, let \(M\) be the \(n\times 2n\) matrix, in which each row represents a stabilizer, and let \(v\) be the \(2n\)-length vector representing the symbolic Pauli \(P\). Then, \(c(P,P_i)= (M\Lambda v)_i\), where 
\(\Lambda=\inlinemat{0 &I \\ I &0}\)
is a \(2n\times 2n\) matrix with \(n\times n\) identity matrices as its off-diagonal blocks (see Sec. 10.5.1 in \cite{nielsen2010quantum}). By solving the linear system over $GF(2)$, \cref{eq-1100449} is equivalent to 
\begin{equation}\label{eq-1101822}
\exists w,\, \textup{wt}(N w \oplus p)\leq \widetilde{F},
\end{equation}
where \(w\) is an \(n\)-length vector, \(N\) is a \(2n\times n\) matrix with columns spanning the null space of \(M\Lambda\), \(p\) is a particular solution of the linear equation \(M\Lambda p=g\oplus h\), in which \(g\) and \(h\) are the \(n\)-length vectors with entries \(g_i\) and \(h_i\), respectively. As a result, we can reduce the number of quantified variables from \(2n\) to \(n\).

Interpreting \(D(\widetilde{\rho},\widetilde{\rho}_{\textup{ideal}})\leq \widetilde{F}\) by \cref{eq-1101822}, we can use the SMT solver to check the validity of \cref{eq-1100520}. The same method is also applied on \(\widetilde{\rho}_X\) (see \cref{eq-1101828}), and ``fault-tolerant'' is claimed if both checks on \(\widetilde{\rho}_Z\) and \(\widetilde{\rho}_X\) pass.

\subsection{Soundness and Completeness}\label{sec-sound-complete-1232230}

Here, we state our main theorem about the soundness and completeness of our verification tool.
\begin{definition}[Soundness and Completeness]\label{def-1310223}
Suppose \(\textup{Alg}\) is an algorithm taking a \cqprog{} \(S\) as input. Then, for the task of verifying quantum fault-tolerance, we define the soundness and completeness of \(\textup{Alg}\) as follows.
\begin{itemize}
\vspace{-5pt}
    \item \textbf{Soundness}: \(\textup{Alg}(S)\) returns ``fault-tolerant'' \(\Longrightarrow\) \(S\) is fault-tolerant.
    \item \textbf{Completeness}: \(S\) is fault-tolerant \(\Longrightarrow\) \(\textup{Alg}(S)\) returns ``fault-tolerant''.
    \vspace{-5pt}
\end{itemize}
\end{definition}

\begin{theorem}\label{thm-compl-sound-1251851}
Suppose \(S\) is a \cqprog{} implementing a quantum gadget with stabilizer quantum operations. 
\begin{itemize}
\vspace{-5pt}
\item If the repeat-until-success statements in \(S\) are memory-less, then our fault-tolerance verification is both sound and complete.
\item If the repeat-until-success statements in \(S\) are conservative, then our fault-tolerance verification is sound. 
\vspace{-5pt}
\end{itemize}
\end{theorem}
The proof can be found in \cref{sec:proof-theorem3}. 
Notably, when \(S\) is non-fault-tolerant, our quantum symbolic execution will always return a fault instantiation that witnesses the failure of fault-tolerance.

\section{Verifying Magic State Distillation}\label{sec-veri-magic-state-1210027}

We have applied our verification tools on QEC gadgets with Clifford circuits. We still need one more non-Clifford QEC gadget for universal quantum computing~\cite{nebe2001invariants,sawicki2017universality}. A widely used approach for universal fault-tolerant quantum computing is using the gate teleportation~\cite{nielsen2010quantum,gottesman2024surviving} with magic state distillation~\cite{knill2004fault,bravyi2005universal} to implement non-Clifford gates. Since the teleportation circuit is already a stabilizer circuit, the remaining problem is verifying the fault-tolerance of a magic state distillation protocol.

Typically, a magic state distillation protocol employs a distillation code \(\mathcal{D}\). It performs the error correction (or error detection) process of \(\mathcal{D}\) on multiple noisy magic states to yield a high-quality magic state encoded in \(\mathcal{D}\), followed by the decoding. The code \(\mathcal{D}\) must exhibit desirable properties for magic state distillation, such as the transversality of the logical \(T\)-gate~\cite{steane1999quantum}.

\begin{wrapfigure}{r}{0.48\linewidth}
    \vspace{-2mm}
    \includegraphics[width=\linewidth]{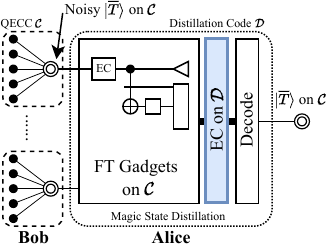}
    \caption{Two-party magic state distillation framework with a QECC $\mathcal{C}$ on Bob and a distillation code $\mathcal{D}$ on Alice.}
    \label{fig:magicstate}
    \vspace{-10pt}
\end{wrapfigure}
In the context of quantum fault-tolerance, we instead use a concatenated code with \(\mathcal{D}\) as an outer code and a QECC \(\mathcal{C}\) (used for fault-tolerance) as an inner code (i.e., each qubit of \(\mathcal{D}\) is encoded in \(\mathcal{C}\)).
To better illustrate the idea of fault-tolerant magic state distillation, we reformulate it in a \textbf{two-party framework} in \cref{fig:magicstate}.
\begin{itemize}
    \vspace{-5pt}
    \item Bob works on physical qubits and produces logical qubits encoded in \(\mathcal{C}\) equipped with fault-tolerant stabilizer operations and non-fault-tolerant magic state preparation.
    \item Alice works on the encoded qubits produced by Bob and performs a magic state distillation protocol using a distillation code \(\mathcal{D}\).
    \vspace{-5pt}
\end{itemize}
Specifically, Alice makes use of the logical operations of \(\mathcal{C}\) provided by Bob and performs magic state distillation as follows:
\begin{enumerate}
 \vspace{-5pt}
\item Call the non-FT magic preparation to produce noisy magic states on \(\mathcal{C}\).
\item Call error correction process of \(\mathcal{C}\) on each noisy magic state.
\item Prepare logical state \(\ket{\overline{+}}\) encoded in \(\mathcal{D}\).
\item Implement logical-\(T\) gate of \(\mathcal{D}\) by gate teleportation with noisy magic states.
\item Perform error correction of \(\mathcal{D}\).
 \vspace{-5pt}
\end{enumerate}
All the steps except step \(1\) contain only stabilizer operations of \(\mathcal{C}\), which can be verified for fault-tolerance. Therefore, we only need to investigate how the errors from the noisy magic states propagate through the distillation process. Intuitively, step \(2\) ensures that the physical errors (in Bob's perspective) are suppressed, and only physical errors (in Alice's perspective), which are logical errors in Bob's view, can further propagate. Fortunately, in Alice's perspective, the gate teleportation is transversal, so errors do not propagate between qubits. Therefore, it suffices to check whether Alice's error correction (i.e., step 5) can correct the errors on the magic state encoded in \(\mathcal{D}\).

Our quantum symbolic execution is not directly applicable for this case since the magic state encoded in \(\mathcal{D}\) is not a stabilizer state. Nevertheless, we can circumvent it by generalizing the input. That is, we instead verify a \textbf{stronger} statement: Alice's error correction can correct errors on an arbitrary state, not only on the magic state. Then, we can use a discretization lemma like \cref{lemma-dis-input-1272245} to reduce it back to a verification task on stabilizer states. We remark that Alice's error correction itself is not required to be fault-tolerant but only needs to be able to correct errors on the input assuming no faults occur during its execution, which we call \textbf{ideal-case correct}. Formally, 
the magic state distillation gadget is fault-tolerant if
\begin{enumerate}
\vspace{-3pt}
\item  Bob's stabilizer operations are fault-tolerant,
\item  Alice's error correction process is ideal-case correct. 
\vspace{-3pt}
\end{enumerate}

Our tool verifies the fault-tolerance of Bob's stabilizer operations. For the ideal-case correctness of Alice's EC process, we can use similar techniques where error injections are only applied on inputs and are disabled in symbolic transitions. Consequently, 
we have:
\begin{theorem}\label{thm-magic-d-210221}
\vspace{0pt}
If the magic state preparation is implemented within the two-party distillation framework and the repeat-until-success statements are either memory-less or conservative, then our fault-tolerance verification is sound.
\vspace{0pt}
\end{theorem}
The proof can be found in \cref{sec-210242}. 
\vspace{-5pt}\section{Case Study}\label{sec-case-study-1221814}

We implemented our verification tool~\footnote{The code is available at: \url{https://github.com/vftqc/vftqecc}.} based on the Julia~\cite{bezanson2017julia} package QuantumSE.jl~\cite{fang2024symbolic} and use Bitwuzla~\cite{DBLP:conf/cav/NiemetzP23} 0.7.0 as the SMT solver. Our experiments are executed on a desktop with AMD Ryzen 9 7950X and 64GB of RAM.

\vspace{-0pt}\subsection{Bug Finding Example: Cat State Preparation}\label{sec-cat-state-1212245}

\begin{figure}[htb]
    \centering
    \begin{subfigure}[b]{0.46\linewidth}
    \centering
    \includegraphics[width=1.0\linewidth]{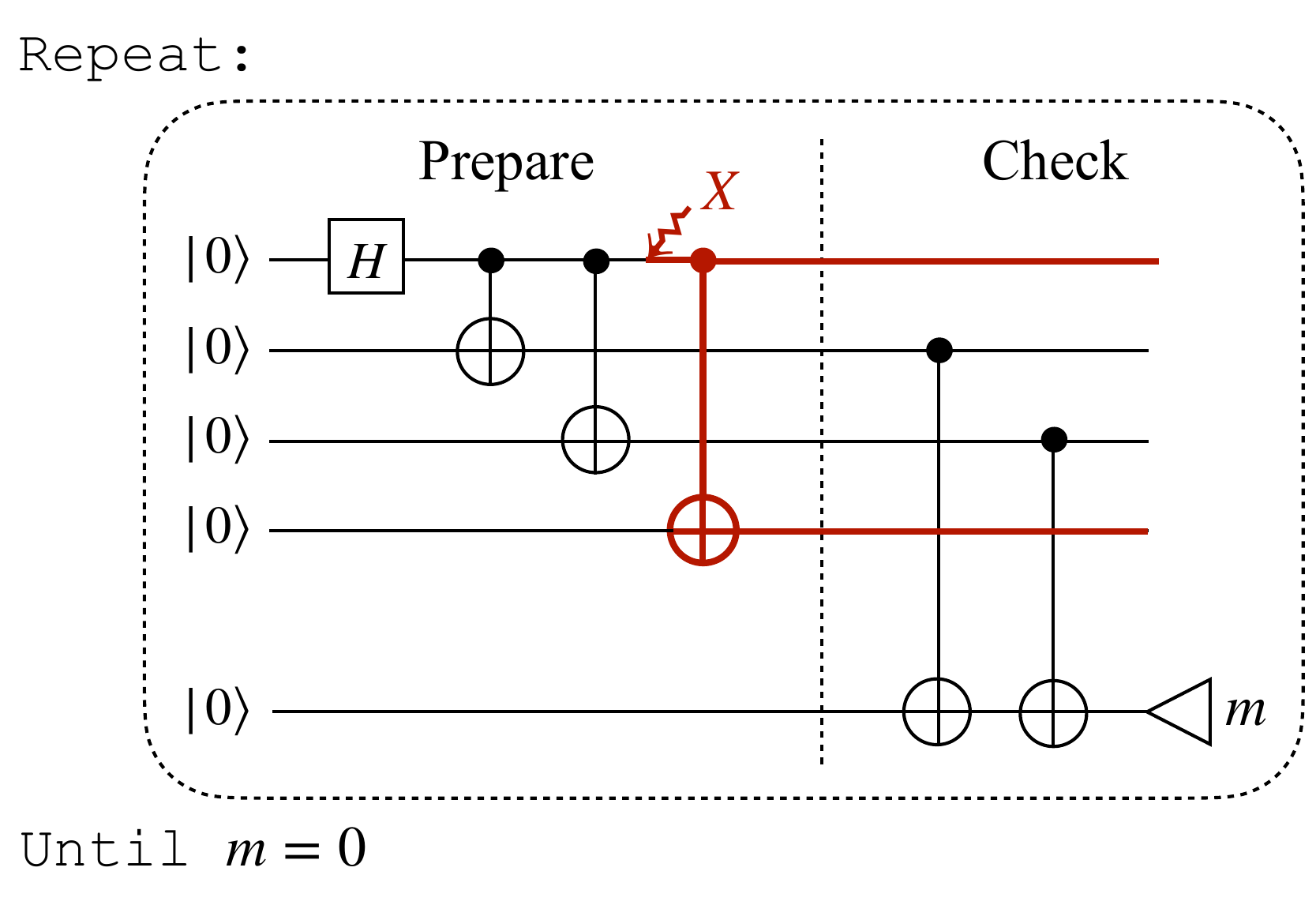}
    \vspace{-15pt}
    \caption{Non-FT cat state preparation.}
     \vspace{-0pt}
    \label{fig-non-ft-cat-prep-1182239}
    \end{subfigure}
    \hfill
    \begin{subfigure}[b]{0.46\linewidth}
    \centering
    \includegraphics[width=1.0\linewidth]{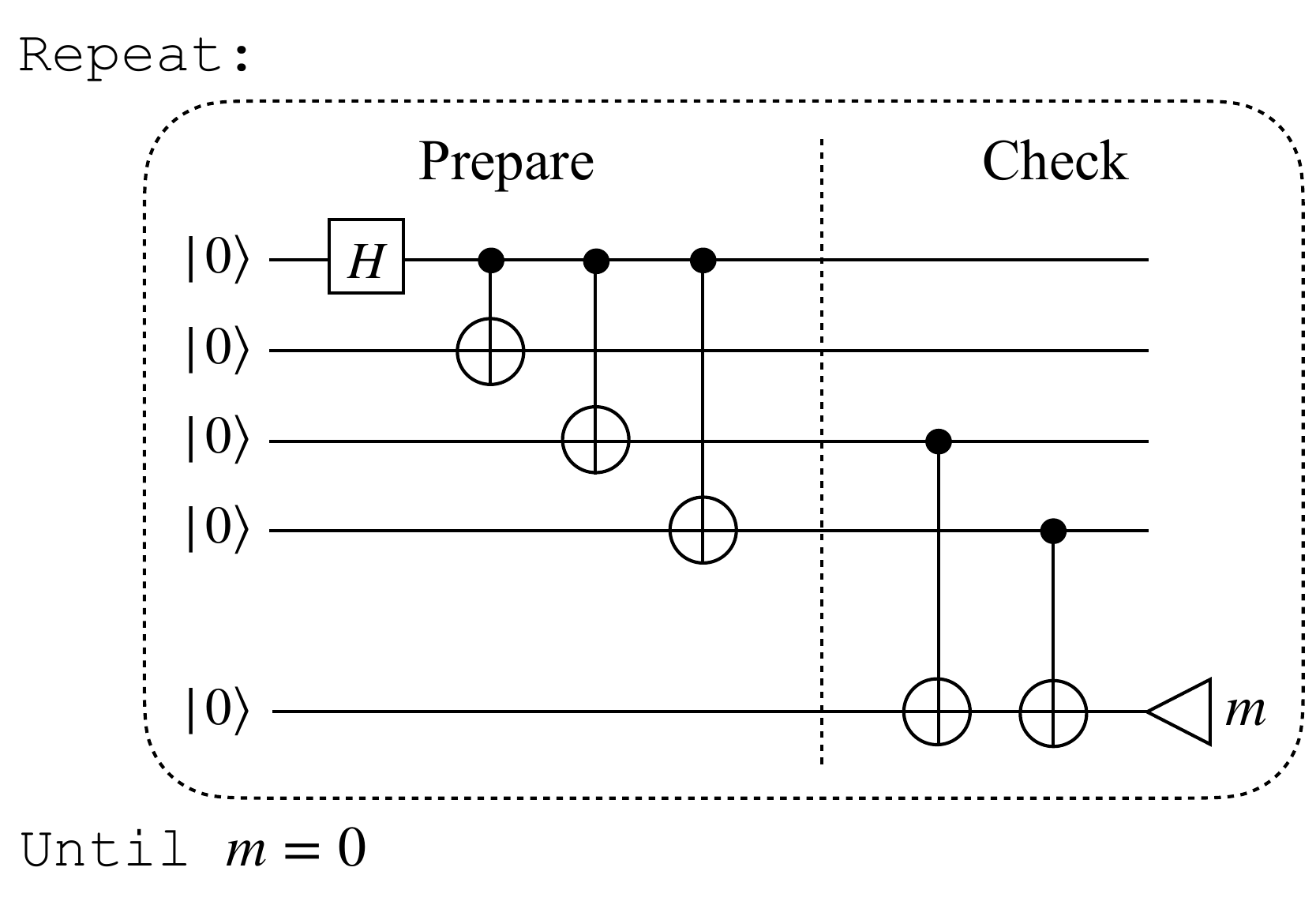}
    \vspace{-15pt}
    \caption{FT cat state preparation.}
     \vspace{-0pt}
    \label{fig-ft-cat-prepr-1182240}
    \end{subfigure}
    \vspace{-1mm}
    \caption{Comparison of \(4\)-qubit cat state preparations with different checks. By our verification tool, we found that the first is non-FT and the second is FT.}
    \vspace{-5mm}
\end{figure}

Cat state preparation is a crucial module in fault-tolerant (FT) quantum computing, extensively used in various FT gadgets such as FT Pauli measurement and Shor's error correction~\cite{shor1996fault}. 
A fault-tolerant~\footnote{Although the cat state itself is not encoded in any QECC with positive code distance, we can still define its fault-tolerance up to \(t\) faults as follows: for any \(s\leq t\), if \(s\) faults occur during the cat state preparation, the output contains at most \(s\) errors.} cat state preparation can be implemented through the following steps: \textbf{1)} Prepare the cat state non-fault-tolerantly; \textbf{2)} Perform a check on the cat state; \textbf{3)} If the check fails, discard the state and start over.
For the second step, a parity check on pairs of qubits of the cat state is performed. However, the pairs must be selected carefully.

For example, we found with our tool that the \(4\)-qubit cat state preparation with the check on the second and third qubits, as shown in \cref{fig-non-ft-cat-prep-1182239}, turns out to be non-FT, and the corresponding error propagation path found by our tool is marked in~\cref{fig-non-ft-cat-prep-1182239}. In contrast, our tool proves that when checking the third and fourth qubits (\cref{fig-ft-cat-prepr-1182240}), the cat state preparation is FT.

\begin{wrapfigure}{r}{0.5\linewidth}
    \vspace{-7mm}
    \includegraphics[width=\linewidth]{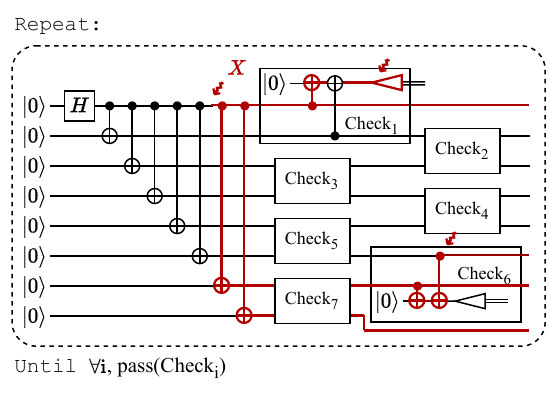}
    \vspace{-20pt}
    \caption{An implementation of \(8\)-qubit cat state preparation. By our verification tool, it is proved to be FT up to \(2\) faults but disproved to be FT for \(3\) faults.}
    \label{fig-8-qubit-cat-prep-1190101}
    \vspace{-5mm}
\end{wrapfigure}
The implementation of cat state preparation also depends on the number of faults it needs to tolerate. For example, consider the cat state preparation in \cref{fig-8-qubit-cat-prep-1190101} with checks on pairs of consecutive qubits. This implementation is only fault-tolerant up to \(2\) faults but fails to be fault-tolerant for \(3\) faults. The error pattern violating the fault-tolerance of \cref{fig-8-qubit-cat-prep-1190101} reported by our tool is marked in red (Pauli-\(X\) error, failed measurement, failed CNOT). Note that all the checks passed since the first and sixth checks failed to detect the errors. Consequently, these \(3\) faults cause \(4\) errors in the output state, which violates the fault-tolerance condition. Two additional examples of fault-tolerance bugs are provided in \cref{sec:bugfinding}. 

\vspace{-5pt}\subsection{Verification of Fault-Tolerance}

\begin{table}[ht]
\vspace{-10pt}
\centering
\caption{Verification time of fault-tolerant gadgets over different QECCs. The quantum Reed-Muller code is used as distillation code for magic state distillation and only the error correction (EC) gadget is reported.}
\vspace{1mm}
\begin{tabular}{|c||c||cccc|}
\hline
\multirow{2}{*}{\textbf{QECC}} &
  \multirow{2}{*}{\(\bm{[[n,k,d]]}\)} &
  \multicolumn{4}{c|}{\textbf{Time (s)}} \\ \cline{3-6} 
 &
   &
  \multicolumn{1}{c|}{\textbf{Prep.}} &
  \multicolumn{1}{c|}{\textbf{CNOT}} &
  \multicolumn{1}{c|}{\textbf{Meas.}} &
  \textbf{EC} \\ \hline
\multirow{2}{*}{Color Code~\cite{bombin2006topological,rodriguez2024experimental}} &
  [[7,1,3]] &
  \multicolumn{1}{c|}{2.81} &
  \multicolumn{1}{c|}{1.36} &
  \multicolumn{1}{c|}{3.65} &
  3.15 \\ \cline{2-6} 
 &
  [[17,1,5]] &
  \multicolumn{1}{c|}{13.91} &
  \multicolumn{1}{c|}{30.29} &
  \multicolumn{1}{c|}{27.92} &
  20.98 \\ \hline
\multirow{3}{*}{Rotated Surface Code~\cite{bombin2007optimal,horsman2012surface}} &
  [[9,1,3]] &
  \multicolumn{1}{c|}{2.96} &
  \multicolumn{1}{c|}{1.27} &
  \multicolumn{1}{c|}{3.91} &
  3.10 \\ \cline{2-6} 
 &
  [[25,1,5]] &
  \multicolumn{1}{c|}{22.38} &
  \multicolumn{1}{c|}{181.72} &
  \multicolumn{1}{c|}{42.79} &
  52.15 \\ \cline{2-6}
  & [[49,1,7]] &
  \multicolumn{1}{c|}{250818} &
  \multicolumn{1}{c|}{Out of Time} &
  \multicolumn{1}{c|}{82319} &
  435011 \\ \hline
\multirow{2}{*}{Toric Code~\cite{Kitaev2003anyons}} &
  [[18,2,3]] &
  \multicolumn{1}{c|}{4.42} &
  \multicolumn{1}{c|}{2.37} &
  \multicolumn{1}{c|}{5.53} &
  4.51 \\ \cline{2-6} 
 &
  [[50,2,5]] &
  \multicolumn{1}{c|}{818.34} &
  \multicolumn{1}{c|}{12168.51} &
  \multicolumn{1}{c|}{916.51} &
  1918.95 \\ \hline
Quantum Reed-Muller Code~\cite{steane1999quantum} &
  [[15,1,3]] &
  \multicolumn{1}{c|}{N.A.} &
  \multicolumn{1}{c|}{N.A.} &
  \multicolumn{1}{c|}{N.A.} &
  4.89 \\ \hline
\end{tabular}

   \vspace{-10pt}
   
    \label{tab:veri-time-1132028}
\end{table}

We perform our verification tool on the fault-tolerant quantum gadgets over the color code~\cite{bombin2006topological,rodriguez2024experimental}, rotated surface code~\cite{bombin2007optimal,horsman2012surface}, and the toric code~\cite{Kitaev2003anyons}. The verification time is shown in \cref{tab:veri-time-1132028}. The preparation~(Prep.) gadgets prepare the logical state \(\ket{\overline{0}}\), (\(\ket{\overline{00}}\) for toric code), the measurement~(Meas.) gadgets perform the logical-\(Z\) measurement (logical-\(Z_1\) for toric code) and the error correction~(EC) gadgets implements the Shor's error correction~\cite{shor1996fault}. 
We also perform our verification on the magic state distillation protocol in \cref{sec-veri-magic-state-1210027} with quantum Reed-Muller code~\cite{steane1999quantum} as the distillation code. 
Other single-qubit Clifford gates, such as the phase gate and Hadamard gate, can be either implemented transversally or via fault-tolerant state preparation of specific ``magic'' states (see Theorem 13.2 in \cite{gottesman2024surviving}. These preparations are analogous to the preparation gadgets verified in \cref{tab:veri-time-1132028}, and the results are provided in \cref{sec-5211304}. 
We remark that the time cost mostly comes from the SMT solver, which is a typical challenge faced by most symbolic execution frameworks. For example, for the $121$-qubit rotated surface code, our symbolic execution (excluding the SMT solver stage) uses only 54.72s to generate the FT constraints of the EC gadget.
\vspace{-5pt}\section{Conclusion and Future Work}
\vspace{-5pt}
In this paper, we presented a novel framework for the verification for fault-tolerance of QECCs. By extending the semantics of classical-quantum programs to model faulty executions and developing discretization theorems for both quantum states and faults, we enabled the use of quantum symbolic execution to verify fault-tolerance of quantum programs. Our approach is both sound and complete under certain structural assumptions on loops and has been implemented into a practical tool capable of verifying real-world QECC implementations across a range of codes and gadgets. 
For future work, we believe the scalability can be significantly enhanced by incorporating interactive strategies~\cite{ying2012floyd,ying2024practical,ying2023coqq} and human-assisted annotations and assertions~\cite{li2020projection}. In addition, we believe our framework could be extended to accommodate leakage errors~\cite{ghosh2013understanding,suchara2015leakage} by introducing new language primitives, such as qubit loss detection and qubit refilling.

\begin{credits}
 \subsubsection{Acknowledgement} We thank the anonymous reviewers for their careful and constructive feedback. The work was supported in part by the U.S. National Science Foundation CAREER Award No. 2338773 and the U.S. Department of Energy, Office of Science, Office of Advanced Scientific Computing Research through the Accelerated Research in Quantum Computing Program MACH-Q project. 
 WF was supported by the Engineering and Physical Sciences Research Council grant EP/X025551/1.
GL was also supported in part by the Intel Rising Star Faculty Award.
\end{credits}
\bibliographystyle{splncs04}
\bibliography{bibliography}
%





\newpage\appendix

\section{Symbolic Stabilizer Operations}\label{sec:symbolic-rule}
\begin{definition}[Symbolic Clifford Gates]
For a Clifford unitary \(U\), the symbolic unitary transform function  \(\textup{UT}(U,\overline{q},\widetilde{\rho})\) is defined as
\[\textup{UT}(U,\overline{q},\widetilde{\rho})=\langle (-1)^{g_1(s_1,\ldots,s_m)}U_{\overline{q}} P_1 U_{\overline{q}}^\dag,\ldots,(-1)^{g_n(s_1,\ldots,s_m)}U_{\overline{q}} P_n U_{\overline{q}}^\dag \rangle.\]
\end{definition}

\begin{definition}[Symbolic Measurement]
The symbolic measurement function \(\textup{M}(q_i,\widetilde{\rho})\) is defined separately in the following two cases:
\begin{enumerate}
    \item \(Z_{q_i}\) commutes with all Pauli operators \(P_j\) in \cref{eq-150251}. This means \(Z_{q_i}\) is expressible by the set of Pauli operators \(\{P_i\}_{i=1}^n\). More specifically, there exists a bit value \(b\) and a sequence of indices \(1\leq j_1<j_2<\cdots < j_k\) such that \(Z_{q_i}=(-1)^b P_{j_1}P_{j_2}\cdots P_{j_k}\). In this case, the measurement outcome is deterministic, and
    \[\textup{M}(q_i,\widetilde{\rho})=(b\oplus g_{j_1}(s_1,\ldots,s_m)\oplus\cdots\oplus g_{j_k}(s_1,\ldots,s_m),1,\widetilde{\rho}).\]

    \item \(Z_{q_i}\) anti-commutes with some of the Pauli operators \(P_j\) in \cref{eq-150251}. By re-choosing a set of stabilizer generators, we can assume without loss of generality that \(Z_{q_i}\) only anti-commutes with \(P_1\), and commutes with \(P_2,\ldots,P_n\). In this case, the measurement outcome is \(0\) or \(1\) with equal probability, and we introduce a new symbol \(s\) to record this measurement outcome as
    \[\textup{M}(q_i,\widetilde{\rho}) = \left(s, \frac{1}{2}, \langle (-1)^s Z_{q_i}, (-1)^{g_2(s_1,\ldots,s_m)} P_2,\ldots, (-1)^{g_n(s_1,\ldots,s_m)} P_n \rangle\right)\]
\end{enumerate}
\end{definition}

The initialization statement on \(q_i\) can be viewed as a measurement on \(q_i\) followed by a Pauli-\(X\) gate controlled by the measurement outcome.

\begin{definition}[Symbolic State Initialization]
The symbolic initialization function \(\textup{IN}(q_i,\widetilde{\rho})\) is defined separately in the following cases:
\begin{enumerate}
\item \(Z_{q_i}\) commutes with all Pauli operators \(P_j\) in \cref{eq-150251}. There exists a bit value \(b\) and a sequence of indices \(j_1,\ldots,j_k\) such that \(Z_{q_i}=(-1)^b P_{j_1}P_{j_2}\cdots P_{j_k}\). Then,
\[\textup{IN}(q_i,\widetilde{\rho})=\langle (-1)^{g_1(s_1,\ldots,s_m)\oplus r\cdot c_1}P_1,\ldots,(-1)^{g_n(s_1,\ldots,s_m)\oplus r\cdot c_n}P_n\rangle,\]
where \(c_j\in\{0,1\}\) indicates whether \(X_{q_i}\) commutes or anti-commutes with \(P_j\) and \(r=b\oplus g_{j_1}(s_1,\ldots,s_m)\oplus \cdots \oplus g_{j_k}(s_1,\ldots,s_m)\) can be viewed as the measurement outcome of \(\textup{M}(q_i,\widetilde{\rho})\).
\item \(Z_{q_i}\) anti-commutes with some of the Pauli operators \(P_j\) in \cref{eq-150251}. Without loss of generality, we assume \(Z_{q_i}\) only anti-commutes with \(P_1\) and commutes with \(P_2,\ldots,P_n\), then
\[\textup{IN}(q_i,\widetilde{\rho})=\langle Z_{q_i}, (-1)^{g_2(s_1,\ldots,s_m)}P_2,\ldots,(-1)^{g_n(s_1,\ldots,s_m)}P_n\rangle.\]
\end{enumerate}
\end{definition}

\section{Assertion for the Decoders}\label{sec-assert-dec-1312156}
In this section, we show how to implement the assertion that is used in the transition \((\textup{SF-CO})\) for the decoding algorithm. 

Suppose our QECC \(\mathcal{C}\) has parameters \([[n,k,d]]\). Let \(t=\lfloor \frac{d-1}{2} \rfloor\) be the number of correctable errors of \(\mathcal{C}\). In our task of building a fault-tolerant error correction gadget, it suffices to let the decoding algorithm behave correctly when the detected syndromes can be reproduced by a Pauli operator of weight \(\leq t\). We will then explain how to construct the assertion.

Let \(\mathcal{G}=\langle P_1,\ldots,P_{n-k}\rangle\) be a set of generators of the stabilizer group \(\mathcal{G}\) of the stabilizer code \(\mathcal{C}\). Then, the syndrome measurements can be defined to be a sequence of multi-qubit Pauli measurements \(P_1,\ldots,P_{n-k}\). Suppose \(m_1,\ldots,m_{n-k}\in\{0,1\}\) are the corresponding measurement outcomes and are provided to the decoding algorithm.

Suppose now a Pauli error \(Q\) occurred on the data qubits. Note that if \(Q\) commutes with \(P_i\), then for any logical state \(\ket{\overline{\psi}}\in\mathcal{C}\), we have
\[Q\ket{\overline{\psi}}=QP_i\ket{\overline{\psi}}=P_iQ\ket{\overline{\psi}},\]
which means \(Q\ket{\overline{\psi}}\) is also in the \(+1\) eigenvalue space of \(P_i\). Thus the measurement outcome of \(P_i\) on \(Q\ket{\overline{\psi}}\) is \(0\) (by convention). In contrast, if \(Q\) anti-commutes with \(P_i\), then
\[Q\ket{\overline{\psi}}=QP_i\ket{\overline{\psi}}=-P_iQ\ket{\overline{\psi}},\]
which means \(Q\ket{\overline{\psi}}\) is in the \(-1\) eigenvalue space of \(P_i\). Thus the measurement outcome of \(P_i\) on \(Q\ket{\overline{\psi}}\) is \(1\). 

The commutativity of Pauli operators can be computed with linear algebra over the field \(\textup{GF}(2)\). Specifically, let \([x_1,\ldots,x_n,z_1,\ldots,z_n]\) be a vector of length \(2n\), where \(x_i,z_i\in\{0,1\}\), \(x_i\) (or \(z_i\)) represent whether a Pauli \(X\) (or \(Z\)) exists at qubit \(i\), and if both \(x_i=z_i=1\), then a Pauli \(Y\) exists at qubit \(i\). For example, the Pauli operator \(X\otimes Z \otimes Y\) corresponds to the vector \([1,0,1,0,1,1]\). Let vector \(v_1,v_2\) represent two \(n\)-qubit Pauli operators \(P_1,P_2\). Then, the value
\[c_{P_1,P_2}\coloneqq v_1^T \Lambda v_2\in\{0,1\}\]
indicate whether \(P_1\) commutes with \(P_2\), where \(\Lambda\) is
\[\Lambda=\begin{bmatrix}0&I\\I&0\end{bmatrix},\]
in which \(I\) is an identity matrix of size \(n\times n\). The weight of a Pauli operator \(P\) can also be computed by its vector representation \(v\):
\[\textup{wt}(v)=\sum_{i=1}^{n} v_i\lor v_{n+i},\]
where \(\lor\) is the logical-or and the summation is for integers but not \(\textup{GF}(2)\).

Let \(G\) be the \((n-k)\times 2n\) matrix with each row being the vector representation of the stabilizer \(P_i\), and \(q\) be the vector representation of the Pauli error \(Q\).
Then, the \(i\)-th entry in the vector
\[G\Lambda q\]
represents whether \(Q\) commutes with \(P_i\) (\(0\) for commute and \(1\) for anti-commute), which is exactly the syndromes when the Pauli error \(Q\) occurred on the data qubits. Recall that \(m=[m_1,m_2,\ldots,m_{n-k}]\) be the detected syndromes. Suppose the output of the decoding algorithm (i.e., the predicted error) is \(r=[r_1,\ldots,r_n,r_{n+1},\ldots,r_{2n}]\). Our requirement for the decoding algorithm is
\[(\exists v\in \textup{GF}(2)^{2n}, \textup{wt}(v)\leq t\land G\Lambda v=m)\rightarrow \textup{wt}(r)\leq t\land G\Lambda r=m,\]
That is, when there exists a Pauli error with weight \(\leq t\) that can reproduce the syndrome, then the output by the decoding algorithm must output a Pauli error with weight \(\leq t\) that can reproduce the syndrome.

\section{Proof of Proposition~\ref{prop-1250405}}\label{sec:proof-prop1}
\begin{proof}
Suppose \(\mathcal{T}\) is a transition tree.
Note that each transition corresponds to a quantum channel, e.g., the \textup{(UT)} transition corresponds to a quantum channel \(\rho\mapsto U_{\overline{q}}\rho U^\dag_{\overline{q}}\).
Since the composition of quantum channels is also a quantum channel, each finite transition path in the tree corresponds to a quantum channel. For any node \(v\in \mathcal{T}\), we use \(\mathcal{E}_v\) to denote the quantum channel corresponding to the path from the root to the node \(v\).

Let \(\mathcal{T}_t\) be the tree truncated from \(\mathcal{T}\) to height \(t\).
Let \(A_t\) be the multi-set consisting of leaf nodes of \(\mathcal{T}_t\) that have terminated (i.e., the nodes of the form \(\langle \downarrow, \sigma,\rho\rangle\)), and let \(B_t\) be the multi-set consisting of leaf nodes of \(\mathcal{T}_t\) that have not terminated.
Then, it is easy to prove by induction on \(t\) that 
\[\sum_{v\in A_t} \mathcal{E}_{v}+\sum_{v\in B_t}\mathcal{E}_v\] 
is a trace preserving quantum channel. Let
\[\mathcal{E}_{t}=\sum_{v\in A_t} \mathcal{E}_v.\]
Then \(\mathcal{E}_t\) must also be a quantum channel (not necessarily trace preserving). Our goal is to prove that \(\lim_{t\rightarrow \infty}\mathcal{E}_t\) exists and is a quantum channel.

Note that \(\mathcal{E}_{t}\preceq \mathcal{E}_{t+1}\) since \(A_{t}\subseteq A_{t+1}\). Here, \(\mathcal{E}_1\preceq \mathcal{E}_2\) means \(\mathcal{E}_2-\mathcal{E}_1\) is completely positive. Let \(C_\mathcal{E}\) be the Choi-Jamio{\l}kowski state~\cite{choi1975completely,jamiolkowski1972linear} of the quantum channel \(\mathcal{E}\). It is well known that the linear map \(\mathcal{E}:\mathcal{L}(\mathcal{H}_0)\rightarrow \mathcal{L}(\mathcal{H}_1)\) is a quantum channel if and only if \(C_\mathcal{E}\) is positive semi-definite and \(\tr_{\mathcal{H}_1}(C_\mathcal{E})\sqsubseteq I_{\mathcal{H}_0}\), where \(\mathcal{L}(\mathcal{H}_0)\) and \(\mathcal{L}(\mathcal{H}_1)\) are the space of linear operators on \(\mathcal{H}_0\) and \(\mathcal{H}_1\), respectively, \(\tr_{\mathcal{H}_0}\) is the partial trace on \(\mathcal{H}_0\) and \(\sqsubseteq\) refers to the Loewner order. Then, it suffices to prove that the limit \(\lim_{t\rightarrow \infty} C_{\mathcal{E}_t}\) exists and is the Choi-Jamio{\l}kowski state of a quantum channel. Since \(\mathcal{E}_{t}\preceq \mathcal{E}_{t+1}\), we have 
\[C_{\mathcal{E}_{t}}\sqsubseteq C_{\mathcal{E}_{t+1}},\quad\quad \textup{and} \quad \quad \tr(C_{\mathcal{E}_t})\leq d_{\mathcal{H}_0},\]
which is a constant upper bound. Therefore, by standard analysis on semi-definite positive matrices, we can prove that the limit of \(\{C_{\mathcal{E}_t}\}\) exists and is also a semi-definite positive matrix. On the other hand, since the partial trace is continuous, we have 
\[\tr_{\mathcal{H}_1}(\lim_{t\rightarrow \infty}C_{\mathcal{E}_t})\sqsubseteq I_{\mathcal{H}_0},\]
which means the limit is the Choi-Jamio{\l}kowski state of a quantum channel.
\end{proof}

\section{Some Notations and Preliminary Results}
In this section, we introduce some notations and preliminary results that will be used in our proof later.

\subsection{Support of Density Operators}
The following are useful characterizations about the support of positive semi-definite operators.
\begin{proposition}\label{prop-1280005}
Suppose \(A\) and \(B\) are positive semi-definite operators, then
\[\supp(A)\subseteq \supp(B)\Longleftrightarrow \exists c>0,\,\, A\sqsubseteq cB.\]
\end{proposition}

By doing this, we can easily get the following results:

\begin{proposition}\label{prop-1271706}
Suppose \(A_i\) are positive semi-definite operators, then
\[\sum_{i}\supp(A_i)=\supp\left(\sum_i A_i\right),\]
where the first summation is on subspaces (the span of the union of subspaces).
\end{proposition}

\begin{proposition}\label{prop-1271718}
Suppose \(\mathcal{E}\) is quantum channel and \(\rho_1,\rho_2\) are positive semi-definite operators. If \(\supp(\rho_1)\supseteq \supp(\rho_2)\), then \(\supp(\mathcal{E}(\rho_1))\supseteq \supp(\mathcal{E}(\rho_2))\).
\end{proposition}

\begin{proposition}\label{prop-1271727}
Suppose \(\mathcal{E}:\rho\mapsto \sum_{i} E_i\rho E_i^\dag\) is a quantum channel and \(\rho\) is a quantum state. Then, for any operator \(E\) that is a linear combination of \(E_i\) (i.e., \(E=\sum_{i} a_i E_i\)), we have
\[\supp(\mathcal{E}(\rho))\supseteq \supp(E\rho E^\dag).\]
\end{proposition}

\subsection{\texorpdfstring{\(r\)-Error Corrector}{r-Error Corrector}}
For a positive integer \(r\leq t\), we define the \(\bm{r}\)\textbf{-error corrector} as the quantum recovery channel \(\mathcal{F}_r\) such that it can correct and can only correct up to \(r\) errors. In fact, we can explicitly construct the \(r\)-error corrector \(\mathcal{F}_r\) by modifying the full corrector (i.e., the \(t\)-error corrector):
\[\mathcal{F}_r:\rho\mapsto \sum_i R_i\mathcal{C}_i\rho\mathcal{C}_i^\dag R_i^\dag,\]
where \(\mathcal{C}_i\) are orthogonal projectors of pairwise orthogonal syndrome subspaces containing at most \(r\) errors, and \(R_i\) is the recovery operator corresponding to the syndrome subspace \(\mathcal{C}_i\) (for more details about the recovery channel, see, e.g., \cite{nielsen2010quantum}). Note that if there are more than \(r\) errors, then \(\mathcal{F}_r\) is strictly trace-decreasing. Specifically, we have the following result:
\begin{proposition}\label{prop-1272320}
For and quantum state \(\rho\) and any logical state \(\ket{\overline{\psi}}\in\mathcal{C}\), we have
\[\mathcal{F}_r(\rho)=\ketbra{\overline{\psi}}{\overline{\psi}}\Longleftrightarrow \supp(\rho)\subseteq \mathcal{S}_r(\ket{\overline{\psi}}),\]
where \(\mathcal{S}_r(\ket{\overline{\psi}})\) is the \(r\)-error space surrounding \(\ket{\overline{\psi}}\) as defined in \cref{def-1311407}.
\end{proposition}

\subsection{Weight of Quantum Channels}

The \textbf{weight of a quantum channel} \(\mathcal{E}:\rho\mapsto \sum_i E_i\rho E_i^\dag\) is the smallest integer \(r\) such that there exists a set of Kraus operators $\{E_i\}$ for \(\mathcal{E}\), where each \(E_i\) is a linear combination of Pauli operators with weight \(\leq r\) (note that the Pauli operators form a linear basis of the matrix space). Then we have the following result
\begin{proposition}\label{prop-1312051}
Suppose \(\mathcal{E}\) is a quantum channel of weight \(r\), then we have
\[\supp(\mathcal{E}(\ketbra{\psi}{\psi}))\subseteq \mathcal{S}_r(\ket{\psi}).\]
\end{proposition}

There is also a similar result in another direction.
\begin{proposition}\label{prop-1312110}
There is a quantum channel \(\mathcal{E}\) of weight \(r\) such that
\[\supp(\mathcal{E}(\ketbra{\psi}{\psi}))\supseteq \mathcal{S}_r(\ket{\psi}).\]
\end{proposition}
\begin{proof}
Define 
\[\mathcal{E}:\rho\mapsto \frac{1}{N} \sum_{i: \textup{wt}(P_i)\leq r} P_i \rho P_i^\dag,\]
where the summation is over all Pauli operators \(P_i\) with weight \(\leq r\), and \(N\) is the number of such Pauli operators. Then it is easy to see \(\mathcal{E}\) satisfies the required condition.
\end{proof}

We will also use the following result.
\begin{proposition}
Suppose \(\mathcal{E}:\rho\mapsto \sum_i E_i\rho E_i^\dag\) is a quantum channel. If \(\mathcal{E}(\ketbra{\psi}{\psi})=\ketbra{\phi}{\phi}\), then for each \(i\)
\[E_i \ket{\psi} = c_i \ket{\phi},\]
for some complex number \(c_i\).
\end{proposition}
\begin{proof}
Using the fact that if the sum of positive semi-definite operators is \(\ketbra{\phi}{\phi}\), then each positive semi-definite operator must be (positively) proportional to \(\ketbra{\phi}{\phi}\).
\end{proof}

\section{Discretization of Input Space (Proof of Theorem~\ref{lemma-dis-input-1272245})}\label{sec:proof-theorem1}
In this section, we provide the proof of our result on the discretization of input space (see \cref{lemma-dis-input-1272245}).

Suppose \(\mathcal{C}\) is a QECC with parameters \([[n,k,d]]\). Let \(t=\lfloor\frac{d-1}{2} \rfloor\) be the number of correctable errors.

\subsection{Proof of Theorem~\ref{lemma-dis-input-1272245}, First Part}\label{sec-210338}
\begin{proof}
The ``only if'' direction is trivial. We only focus on the ``if'' direction. 

First, suppose \(S\) is a \cqprog{}\ implementing an error correction gadget.
Let \(\ket{\overline{\psi}}\) be an arbitrary logical quantum state, and \(\mathcal{E}_r\) denote the quantum channel defined in \cref{prop-1312110} with weight \(r\), and \(\mathcal{T}\) be an arbitrary transition tree of \(S\) with \(s\) faults, such that \(s+r\leq t\). Let \(\mathcal{E}_\mathcal{T}\) be the quantum channel corresponding to the tree \(\mathcal{T}\). We will first prove
\begin{equation*}
\mathcal{F}_{r+s}\circ \mathcal{E}_\mathcal{T}\circ \mathcal{E}_r (\ketbra{\overline{\psi}}{\overline{\psi}})=\ketbra{\overline{\psi}}{\overline{\psi}},
\end{equation*}
where \(\mathcal{F}_{r+s}\) is the \((r+s)\)-error corrector defined in \cref{prop-1272320}.

By \cref{prop-1272320}, \cref{prop-1312051} and the assumption, we know that
\begin{equation}\label{eq-1280042}
\mathcal{F}_{r+s}\circ \mathcal{E}_\mathcal{T}\circ \mathcal{E}_r(\ketbra{\overline{i}}{\overline{i}})=\ketbra{\overline{i}}{\overline{i}},
\end{equation}
for any logical computational state \(\ket{\overline{i}}\), and
\begin{equation}\label{eq-1280043}
\mathcal{F}_{r+s}\circ \mathcal{E}_\mathcal{T}\circ \mathcal{E}_r(\ketbra{\overline{+}}{\overline{+}})=\ketbra{\overline{+}}{\overline{+}},
\end{equation}
for \(\ket{\overline{+}}=\sum_{i=0}^{2^k-1} \ket{\overline{i}}/\sqrt{2^k}\).

Now, let \(\{E_i\}\) be a set of Kraus operators representing the quantum channel \(\mathcal{F}_{r+s}\circ\mathcal{E}_{\mathcal{T}}\circ \mathcal{E}_r\). Then, by \cref{eq-1280042}, it must hold that, for all \(i,j\),
\begin{equation}\label{eq-1280229}
E_i\ket{\overline{j}}= \alpha_{i,j}\ket{\overline{j}},
\end{equation}
for some complex number \(\alpha_{i,j}\). Similarly, by \cref{eq-1280043}, we have
\begin{equation}\label{eq-1280230}
E_i\ket{\overline{+}}=\beta_i \ket{\overline{+}},
\end{equation}
and \(\sum_i |\beta_i|^2 =1\).
Combining \cref{eq-1280229} and \cref{eq-1280230}, we have
\[\sum_{j=0}^{2^k-1} \alpha_{i,j}\ket{\overline{j}}=\beta_i \sum_{j=0}^{2^k-1} \ket{\overline{j}}.\]
Therefore, \(\alpha_{i,j}=\beta_i\) for any \(j\).

Suppose \(\ket{\overline{\psi}}=\sum_{j} x_j \ket{\overline{j}}\). Then
\[E_i\ket{\overline{\psi}}=\beta_i\sum_j  x_j\ket{\overline{j}}=\beta_i \ket{\overline{\psi}}.\]
Since \(\sum_i |\beta_i|^2 =1\), we can conclude that
\[\mathcal{F}_{r+s}\circ \mathcal{E}_\mathcal{T}\circ \mathcal{E}_r(\ketbra{\overline{\psi}}{\overline{\psi}})=\ketbra{\overline{\psi}}{\overline{\psi}}.\]

Now, let \(\rho\) be an arbitrary quantum state containing \(r\) errors w.r.t. \(\ket{\overline{\psi}}\) (i.e., \(\supp(\rho)\subseteq \mathcal{S}_r(\overline{\psi})\)). Then, by \cref{prop-1312110}, we have
\[\supp(\rho)\subseteq \supp(\mathcal{E}_r(\ketbra{\overline{\psi}}{\overline{\psi}})).\]
This means 
\[\supp(\mathcal{F}_{r+s}\circ\mathcal{E}_\mathcal{T}(\rho))\subseteq \supp(\mathcal{F}_{r+s}\circ\mathcal{E}_{\mathcal{T}}\circ \mathcal{E}_{r}(\ketbra{\overline{\psi}}{\overline{\psi}}))=\supp(\ketbra{\overline{\psi}}{\overline{\psi}}),\]
and \(\mathcal{F}_{r+s}\circ\mathcal{E}_{\mathcal{T}}\) is trace-preserving on \(\rho\) (since it is trace-preserving on \(\mathcal{E}_r(\ketbra{\psi}{\psi})\)). 
Then we can conclude that
\[\mathcal{F}_{r+s}\circ\mathcal{E}_{\mathcal{T}}(\rho)=\ketbra{\overline{\psi}}{\overline{\psi}}.\]

The case that \cqprog{}\ \(S\) is a gate gadget can be proved similarly. 
\end{proof}

\subsection{Proof of Theorem~\ref{lemma-dis-input-1272245}, Second Part}
\begin{proof}
The ``only if'' direction is trivial. We only focus on the ``if'' direction.

Let \(S\) be a \cqprog{}\ implementing a logical-\(Z\) basis measurement gadget.
Let \(\ket{\overline{\psi}}\) be a logical quantum state, \(\mathcal{E}_r\) be the quantum channel of weight \(r\) given in \cref{prop-1312110}, and \(\mathcal{T}\) be a transition tree of \(S\) with \(s\) faults, such that \(s+r\leq t\). Let \(L_i\) be the set of leaf nodes of \(\mathcal{T}\) such that the measurement outcome is \(i\). We define \(\mathcal{E}_{\mathcal{T},i}\) to be
\[\mathcal{E}_{\mathcal{T},i} = \sum_{w\in L_i} \mathcal{E}_{w},\]
where \(\mathcal{E}_w\) is the quantum channel corresponding to the transition path from root node to \(w\). 

By assumption, we know that
\begin{equation}\label{eq-1280311}
\tr(\mathcal{E}_{\mathcal{T},i}\circ \mathcal{E}_r\left(\ketbra{\overline{j}}{\overline{j}})\right)=\delta_{i,j}.
\end{equation}
Let \(\{E_{i,j}\}_j\) be a set of Kraus operators of the quantum channel \(\mathcal{E}_{\mathcal{T},i}\circ\mathcal{E}_r\).
This means for any \(l\neq i\), we have
\begin{equation}\label{eq-1280326}
E_{i,j}\ket{\overline{l}}=0.
\end{equation}

Assume \(\ket{\overline{\psi}}=\sum_{i} a_i \ket{\overline{i}}\). Then by \cref{eq-1280326}, we have
\[E_{i,j}\ket{\overline{\psi}}=a_i E_{i,j}\ket{\overline{i}},\]
which means
\begin{align}
\tr(\mathcal{E}_{\mathcal{T},i}\circ\mathcal{E}_r(\ketbra{\overline{\psi}}{\overline{\psi}}))&=|a_i|^2 \tr\left(\sum_{j} E_{i,j}\ketbra{\overline{i}}{\overline{i}}E_{i,j}^\dag\right) \nonumber \\
&=|a_i|^2 \tr(\mathcal{E}_{\mathcal{T},i}\circ\mathcal{E}_r(\ketbra{\overline{i}}{\overline{i}})) \nonumber \\
&=|a_i|^2, \label{eq-1280348}
\end{align}
where \cref{eq-1280348} is by \cref{eq-1280311}. Therefore, we can see that the distribution of the measurement outcome is correct for an arbitrary input \(\ket{\overline{\psi}}\).
\end{proof}

\section{Discretization of Faults (Proof of Theorem~\ref{lemma-disc-fautls-191645})}\label{sec:proof-theorem2}
In this section, we provide the proof of our result on the discretization of faults (see \cref{lemma-disc-fautls-191645}).
First, we introduce some notations and auxiliary lemmas that will be used in our proof.

\subsection{Tree Schema}
Suppose \(S\) is a fixed \cqprog{}. Note that any faulty transition tree starting from \(S\) has a fixed structure. That is, the classical control flow is fixed, though the faulty transitions are not specified and can be arbitrary. We call this tree structure as the \textbf{tree schema} of \(S\). Then, any faulty transition tree of \(S\) can be viewed as an instance of the tree schema of \(S\). Suppose \(v\) is a node in the tree schema, and \(\mathcal{T}\) is a faulty transition tree, then we will use \(\bm{v_\mathcal{T}}\) to denote the node in tree \(\mathcal{T}\) corresponding to \(v\). We will use \(\bm{q(v_\mathcal{T})}\) to denote the quantum state in the classical-quantum configuration at node \(v_\mathcal{T}\) in tree \(\mathcal{T}\). 

To further explain the tree schema, suppose \(v\) is a node in the tree schema of \(S\), \(\mathcal{T}_1\) and \(\mathcal{T}_2\) are faulty transition trees of \(S\). Let the classical-quantum configuration at \(v_{\mathcal{T}_1}\) be \(\langle S_1,\sigma_1,\rho_1\rangle\) and that at \(v_{\mathcal{T}_2}\) be \(\langle S_2,\sigma_2,\rho_2\rangle\). It is easy to see that \(S_1=S_2\) and \(\sigma_1=\sigma_2\). That is, the \cqprog{}\  and classical state in a configuration do not depend on a specific choice of the faulty transition tree and are fixed once the initial \cqprog\ \(S\) (and thus the tree schema) is fixed.

\subsection{Non-Terminating Probability}

Suppose \(\mathcal{T}\) is a transition tree, and \(\rho\) is a quantum state. The \textbf{non-terminating probability} of a node \(v\) in the tree \(\mathcal{T}(\rho)\) is defined as
\[\tr(q(v))-\sum_{w} \tr(q(w)),\]
where the summation is over all leaf nodes of the sub-tree starting at the node \(v\).
A node \(v\) in the tree \(\mathcal{T}(\rho)\) is called \textbf{terminating} if the non-terminating probability is zero.
Intuitively, a node is terminating if, starting from this node, we can always get an output (i.e., with probability \(1\)). A node is called \textbf{non-terminating} if it is not terminating. 

Then, we can easily see the following results.
\begin{proposition}\label{lemma-1262122}
For a transition tree, if the root node is non-terminating, then a sub-tree containing all non-terminating nodes exists with (faulty-) transitions between them. This sub-tree is an infinite tree and does not have leaf nodes.
\end{proposition}
We will call this sub-tree the \textbf{non-terminating sub-tree}. The proof of \cref{lemma-1262122} is obvious if we notice the following property.
\begin{proposition}\label{lemma-1270338}
For each node, its non-terminating probability equals the sum of the non-terminating probabilities of its children.
\end{proposition}
\begin{proof}
This can be proved by checking all the (faulty)-transition rules.
\end{proof}

\subsection{An Auxiliary Lemma}
\begin{lemma}\label{lemma-1261556}
For any node \(v\) in the tree schema of \(S\), any transition tree \(\mathcal{T}\) with \(s\) faults, and any quantum state \(\rho\), we have
\begin{equation}\label{eq-1270407}
\supp(q(v_{\mathcal{T}(\rho)})) \subseteq \sum_{\mathcal{T}'} \supp(q(v_{\mathcal{T}'(\rho)})),
\end{equation}
where the summation is over all \(s\)-Pauli-fault transition trees \(\mathcal{T}'\) such that the locations of Pauli-fault transitions in \(\mathcal{T}'\) are the same as the locations of fault transitions in \(\mathcal{T}\).
\end{lemma}
\begin{proof}
We fix the quantum state \(\rho\) and the transition tree \(\mathcal{T}\) and prove that \cref{eq-1270407} holds for any node \(v\) using the induction on the level of \(v\).

We define the level of a node as the length of the path from it to the root node. \cref{eq-1270407} holds trivially for the level-\(0\) node (the root note). Now assume that \cref{eq-1270407} holds for all nodes of level less than or equal to \(n\). Now suppose \(v\) is a node of level \(n+1\) and the parent node of \(v\) is \(u\). Then, the induction hypothesis guarantees that 
\begin{equation}\label{eq-1271712}
\supp(q(u_{\mathcal{T}(\rho)}))\subseteq \sum_{\mathcal{T}'}\supp(q(u_{\mathcal{T}'(\rho)})).
\end{equation}

If \(u\rightarrow v\) (ideal transition), then it is quite straightforward that \cref{eq-1270407} holds for \(v\). 
Let \(\mathcal{E}\) be the quantum channel corresponding to the transition of \(u\rightarrow v\). We have
\begin{align}
\supp(q(v_{\mathcal{T}(\rho)}))&=\supp(\mathcal{E}(q(u_{\mathcal{T}(\rho)}))) \nonumber\\
&\subseteq \sum_{\mathcal{T}'}\supp(\mathcal{E}(q(u_{\mathcal{T}'(\rho)}))) \label{eq-1270523}\\
&=\sum_{\mathcal{T}'}\supp(q(v_{\mathcal{T}'(\rho)})), \label{eq-1270531}
\end{align}
where \cref{eq-1270523} is by combining \cref{eq-1271712} with \cref{prop-1271706} and \cref{prop-1271718}, \cref{eq-1270531} is because \(\mathcal{E}\) is also the quantum channel corresponding to the transition \(u\rightarrow v\) in Pauli-fault transition trees \(\mathcal{T}'\).

If \(u\leadsto v\), let 
\[\mathcal{E}:\rho\mapsto \sum_{i} E_i\rho E_i^\dag\] 
be the quantum channel corresponding to the transition \(u\leadsto v\) in the tree \(\mathcal{T}\). First, assume that the program statement of this transition is the unitary transform statement \(U(\overline{q})\). Then we can see that the quantum channel of \(u\leadsto v\) in a Pauli fault transition tree has the form 
\[\rho\mapsto P_{\overline{q}} U_{\overline{q}}\rho U_{\overline{q}}^\dag P_{\overline{q}},\] 
and \(P_{\overline{q}}\) can be arbitrary Pauli operator acting non-trivially on \(\overline{q}\). Therefore, we have
\begin{align}
\supp(q(v_{\mathcal{T}(\rho)}))&= \supp(\mathcal{E}(q(u_{\mathcal{T}(\rho)}))) \nonumber\\
&= \sum_i \supp(E_iq(u_{\mathcal{T}(\rho)})E_i^\dag) \label{eq-1271705}\\
&\subseteq \supp\left(\sum_{P_{\overline{q}}}P_{\overline{q}}U_{\overline{q}}q(u_{\mathcal{T}(\rho)})U_{\overline{q}}^\dag P_{\overline{q}}\right) \label{eq-1271726}\\
&=\sum_{\mathcal{T}'}\supp(q(v_{\mathcal{T}'(\rho)})), \label{eq-1271729}
\end{align}
where \cref{eq-1271705} is by \cref{prop-1271706}, \cref{eq-1271726} is by combining \cref{prop-1271727} with the fact that each \(E_i\) acts non-trivially only on \(\overline{q}\) and thus can be written as a linear combination of the operators \(P_{\overline{q}}U_{\overline{q}}\), \cref{eq-1271729} is again by \cref{prop-1271706}. Other types of program statements can be handled similarly. For example, if the program statement is the measurement statement \(\meas\  q\), then we only need to note the fact that the operators \(P_q\ket{0}_q\bra{0}Q_q\) (or \(P_q\ket{0}_q\bra{0}Q_q\)), where \(P_q\) and \(Q_q\) can be arbitrary Pauli operators acting non-trivially on \(q\), span all the operators that act non-trivially only on \(q\).

Therefore, the induction hypothesis also holds for all nodes of level \(n+1\). We can conclude that \cref{eq-1270407} holds in general.

\end{proof}

\subsection{Proof of Theorem~\ref{lemma-disc-fautls-191645}}
Now, we are ready to prove \cref{lemma-disc-fautls-191645}.
\begin{proof}
The ``only if'' direction is trivial. We only focus on the ``if'' direction. Without loss of generality, we assume the gadget is an error correction gadget. Other gadget types can be handled similarly.

Suppose \(\ket{\overline{\psi}}\) is a logical quantum state.
Let \(\rho\) be a density operator such that \(\supp(\rho)\subseteq \mathcal{S}_r(\ket{\overline{\psi}})\), which can be viewed as a noisy version of \(\ket{\overline{\psi}}\) with \(r\) errors.
Further, suppose \(\mathcal{T}\) is a transition tree of \(S\) with \(s\) faults such that \(r+s\leq t\), and \(v\) is a leaf node of the tree schema of \(S\).

We will separate the proof into two parts: the first part proves that the quantum state at any leaf node of \(\mathcal{T}(\rho)\) is close to \(\ket{\overline{\psi}}\); the second part proves that the tree \(\mathcal{T}(\rho)\) is almost surely terminating (the root node is terminating).

\textbf{First part.} We will first prove that
\[\supp(q(v_{\mathcal{T}(\rho)}))\subseteq \mathcal{S}_s(\ket{\overline{\psi}}),\]
where \(\mathcal{S}_s(\ket{\overline{\psi}})\) is the \(s\)-error space surrounding \(\ket{\overline{\psi}}\).
This can be proved by \cref{lemma-1261556}. Specifically, by the assumption we have that the \cqprog{}\  \(S\) is fault-tolerant under Pauli faults, then for any \(s\)-Pauli-fault transition tree \(\mathcal{T}'\), its leaf node \(v_{\mathcal{T}'(\rho)}\) satisfies
\[\supp(q(v_{\mathcal{T}'(\rho)})) \subseteq \mathcal{S}_s(\ket{\overline{\psi}}).\]
Then, by \cref{lemma-1261556}, we have
\[\supp(q(v_{\mathcal{T}(\rho)}))\subseteq \sum_{\mathcal{T}'}\supp(q(v_{\mathcal{T}'(\rho)}))\subseteq \mathcal{S}_s(\ket{\overline{\psi}}),\]
where the summation is over all \(s\)-Pauli-fault transition tree \(\mathcal{T}'\). Therefore, we can see that each output quantum state (the quantum state at a leaf node) of the tree \(\mathcal{T}(\rho)\) contains at most \(s\) errors. 

\textbf{Second part.}
Then, it remains to show that the probability of getting an output state is \(1\), i.e., the root node of \(\mathcal{T}(\rho)\) is terminating. Now, suppose the root node is not terminating. Then by \cref{lemma-1262122}, there exists an infinite sub-tree of \(\mathcal{T}(\rho)\), called the non-terminating tree, such that all nodes in the sub-tree are not terminating. Since the number of faults on each path is finite (i.e., \(\leq s\)), there exists a node \(v_{\mathcal{T}(\rho)}\) in the non-terminating sub-tree such that the sub-tree of the non-terminating sub-tree below \(v_{\mathcal{T}(\rho)}\) does not contain any faulty transitions. Let \(p=\tr(q(v_{\mathcal{T}(\rho)}))\) and \(p^{\textup{nt}}\) be the non-terminating probability of \(v_{\mathcal{T}(\rho)}\). Note that \(p\geq p^{\textup{nt}}>0\). More generally, for any integer \(j\geq 1\), we define
\[\mathcal{E}_j=\sum_{w\in N_j} \mathcal{E}_{v:w},\]
where \(N_j\) is the set of non-terminating nodes that are the \(j\)th-level descendants of \(v_{\mathcal{T}(\rho)}\) (e.g. the first-level descendants are the children of \(v_{\mathcal{T}(\rho)}\)), and \(\mathcal{E}_{v:w}\) is the quantum channel corresponding to the transition path from \(v\) to \(w\). Then, we have
\begin{equation}\label{eq-1270307}
\begin{split}
\tr\left(\mathcal{E}_j(q(v_{\mathcal{T}(\rho)}))\right)&=\tr\left(\sum_{w\in N_j} q(w)\right) \\
&\geq \sum_{w\in N_j} p_w^{\textup{nt}}\\
&=p^{\textup{nt}},
\end{split}
\end{equation}
where \(p_w^{nt}\) is the non-terminating probability of \(w\) and the last equality is by \cref{lemma-1270338}.
On the other hand, by \cref{lemma-1261556}, we can find a finite set of \(s\)-Pauli-fault transition trees \(\{\mathcal{T}'_1,\mathcal{T}'_2,\ldots,\mathcal{T}'_l\}\) such that
\[\supp(q(v_{\mathcal{T}(\rho)}))\subseteq \sum_{i=1}^l \supp(q(v_{\mathcal{T}'_i(\rho)})).\]
Then, by \cref{prop-1280005}, there exists a positive number \(\alpha\) such that
\[q(v_{\mathcal{T}(\rho)}) \sqsubseteq \alpha \sum_{i=1}^l q(v_{\mathcal{T}'_i(\rho)}).\]
Then, we have
\[\tr(\mathcal{E}_j(q(v_{\mathcal{T}(\rho)}))) \leq \alpha \sum_{i=1}^l \tr(\mathcal{E}_j(q(v_{\mathcal{T}'_i(\rho)}))),\]
for any \(j>0\).
Note that the non-terminating sub-tree below \(v_{\mathcal{T}(\rho)}\) is fault-free. Thus, its transitions are exactly the same as those in the Pauli-fault transition trees \(\mathcal{T}'_i\). 
Furthermore, by assumption (i.e., \(S\) is fault-tolerant with \(s\) Pauli faults and thus is terminating), any node in any \(s\)-Pauli-fault transition tree of \(S\) has zero non-terminating probability, which means any infinite sub-tree without leaf nodes must have vanishing probability. Therefore, together with \cref{eq-1270307}, we have
\[0<p^{\textup{nt}}\leq \lim_{j\rightarrow \infty} \tr(\mathcal{E}_j(q(v_{\mathcal{T}(\rho)}))) \leq \lim_{j\rightarrow\infty} \alpha \sum_{i=1}^l \tr(\mathcal{E}_j(q(v_{\mathcal{T}'_i(\rho)})))=0,\]
which is a contradiction. Therefore, the root node of \(\mathcal{T}(\rho)\) is terminating.
\end{proof}

\section{Soundness and Completeness (Proof of Theorem~\ref{thm-compl-sound-1251851})}\label{sec:proof-theorem3}
In this section, we provide the proof of our result on the soundness and completeness of our fault-tolerance verification (see \cref{thm-compl-sound-1251851}). We will make the following adjustment.
\begin{remark}\label{remark-1311516}
We replace each of the initialization statements in a \cqprog{}\ \(S\) with a \(Z\)-basis measurement followed by a classically controlled \(X\) gate. This adjustment does not affect the overall semantics and fault-tolerance property and can simplify our proof (since the ideal transitions without initialization preserve the purity of quantum states). Furthermore, note that \(S\) remains a \cqprog{} (though without initialization statements), our discretization theorems still apply.
\end{remark}


Then, we introduce some notations and auxiliary lemmas that will be used in our proof.

\subsection{Explanation of Conservative Repeat-Until-Success}\label{sec-cons-loop-210103}
Here, we explain the definition of conservative repeat-until-success (see \cref{def-con-rus-210106}). 

Let \(\{\repeatuntil{S}{b}\}\) be a conservative repeat-until-success, with sub-program \(S\) as its loop body. First note that, the classical variables of \(S\) are reset at the beginning of each iteration. Therefore, in each iteration, the execution of \(S\) is on the same transition tree (i.e., with the same classical control flow), though the initial quantum state could be different.

Let \(\mathcal{T}\) be the unique ideal (fault-free) transition tree of \(S\). For a node \(v\) in the transition tree, there is a unique path from the root node to \(v\). We simply call it the path to \(v\). We use \(\mathcal{E}_v\) to denote the quantum channel corresponding to this transition path (the composition of quantum channels for every transitions on this path). Since we do not include the initialization statement (see \cref{remark-1311516}), $\mathcal{E}_v$ has the following form: $\mathcal{E}_v=E_v (\cdot) E_v^\dag$, where $E_v$ is a product of projectors $\ketbra{0}{0}$ or $\ketbra{1}{1}$ (from (\textup{M0}) or (\textup{M1})) and unitaries (from (\textup{UT})).



Now, we explain our definition of ``\textbf{the fault-free semantics of $S$ is non-adaptive and idempotent}'' in details. 
Since the quantum program only includes Clifford unitaries and Pauli projectors (i.e., $\frac{I+Z}{2}$ or $\frac{I-Z}{2}$ from measurements), $E_v$ has the following form:
\[E_v = U_{m+1} \cdot \frac{I+P_m}{2} \cdot U_l \cdots \frac{I+P_2}{2} \cdot U_2 \cdot \frac{I+P_1}{2} \cdot U_1,\]
where $P_i$ are Pauli operators and $U_i$ are Clifford unitaries. 

The \textbf{non-adaptivity} means that the operation to be performed does not depend on the previous measurement outcomes. It thus ensures that all leaf nodes share the same unitaries $U_1,\ldots,U_m$ and only the phases preceding $P_i$ are different. More specifically, there exists a set of Pauli operators $P_1,\ldots,P_m$ and a set of Clifford unitaries $U_1,\ldots,U_{m+1}$ such that for any leaf node $v$, we have
\begin{equation}\label{eq-5232121}
E_v = U_{m+1} \cdot \frac{I+\beta_m^v P_m}{2} \cdot U_m \cdots \frac{I+\beta_2^v P_2}{2} \cdot U_2 \cdot \frac{I+\beta_1^v P_1}{2} \cdot U_1,
\end{equation}
where $\beta_i^v\in\{-1,1\}$. Moreover, since $(\beta_1^v,\ldots,\beta_m^v)$ encode all measurement outcomes in the path to $v$, the leaf nodes $v$ one-to-one corresponds to the sequences $(\beta_1^v,\ldots,\beta_m^v)$.

The \textbf{idempotency} of the semantics of $S$ means that each $E_v$ is a projector, i.e., $E_vE_v= E_v$. By this, we can further show that $E_v$ must be a Hermitian projector.
\begin{proposition}
$E_v$ is a Hermitian projector.
\end{proposition}
\begin{proof}
We can commute all Clifford unitaries $U_i$ in \cref{eq-5232121} to the leftmost position:
\[E_v=(U_{m+1}\cdots U_1)\cdot \frac{I+P'_m}{2}\cdots \frac{I+P'_1}{2},\]
for some Pauli operators $P'_1,\ldots, P'_m$.
Suppose $E_v\neq 0$ (otherwise it is already Hermitian). Since $E_v$ is a projector, for any $\ket{\psi}$ in the image subspace $\textup{im}(E_v)$ of $E_v$, we have $E_v\ket{\psi}=\ket{\psi}$. Note that $U_{m+1}\cdots U_1$ is a unitary and $\frac{I+ P'_i}{2}$ are projectors, which means it must be the case that $\frac{I+ P'_i}{2}\ket{\psi}=\ket{\psi}$ and $U_{m+1}\cdots U_1\ket{\psi}=\ket{\psi}$.
This further means $\ket{\psi}$ is in the $+1$ eigenspaces of every $P_i'$. If there exists $i\neq j$ such that $P_i', P_j'$ anti-commute, we have 
\[\ket{\psi}=P'_iP'_j\ket{\psi}=-P'_jP'_i\ket{\psi}=-\ket{\psi},\]
which is impossible. Therefore, $P_1',\ldots,P'_m$ pairwise commute.
This means $\Pi\coloneqq \frac{I+P'_1}{2}\cdots \frac{I+P'_m}{2}$ is a Hermitian projector and $U_{m+1}\cdots U_1$ acts as identity on the image subspace of $\textup{im}(\Pi)$. Therefore we have
\[E_v = U_{m+1}\cdots U_1 \cdot \Pi= \Pi,\]
which is a Hermitian projector.
\end{proof}

Moreover, note that $S$ always terminates, it is trace-preserving. We can see that
\begin{align}
I&=\sum_v E^\dag_v E_v=\sum_v E_v=\sum_{(\beta_1^v,\ldots,\beta_m^v)\in\{-1,1\}^m} E_v\nonumber\\
&=U_{m+1}\cdot U_m\cdots U_2\cdot U_1.\label{eq-5232201}
\end{align}
Then, in \cref{eq-5232121}, commuting all Clifford unitaries $U_i$ to the rightmost position, we can see the following fact:
\begin{proposition}\label{prop-5230556}
For any conservative loop body $S$, there exists a set of commuting Pauli operators $\hat{P}_1,\ldots, \hat{P}_m$, such that for any leaf node $v$, $E_v$ is a stabilizer projector of the following form 
$$E_v=\frac{I+\beta^v_1\hat{P}_1}{2} \cdots \frac{I+\beta^v_m \hat{P}_m}{2},$$
where $\beta_{i}^v\in\{-1,1\}$. 
\end{proposition}
\begin{proof}
In \cref{eq-5232121}, we commute all Clifford unitary to the rightmost position, and we can obtain the desired form. It remains to show that $\hat{P}_1\cdots \hat{P}_m$ pairwise commute.

If $E_v\neq 0$, we can prove $\hat{P}_1,\ldots,\hat{P}_m$ pairwise commute similar to the previous proof. 

For the general case, suppose without loss of generality that $\hat{P}_1$ anti-commutes with $\hat{P}_2$. Then consider the sum:
\begin{align}
\sum_{v:\beta^v_1=\beta^v_2=1}E_v&=\sum_{\beta_3,\ldots,\beta_m \in\{-1,1\}^{m-2} } \frac{I+\hat{P}_1}{2}\frac{I+\hat{P}_2}{2}\frac{I+\beta_3 \hat{P}_3}{2}\cdots \frac{I+\beta_m \hat{P}_m}{2}\\
&=\frac{I+\hat{P}_1}{2}\frac{I+\hat{P}_2}{2}.
\end{align}
This means $\frac{I+\hat{P}_1}{2}\frac{I+\hat{P}_2}{2}$ is also Hermitian. However, this is not possible since
\[((I+\hat{P}_1)(I+\hat{P}_2))^\dag-(I+\hat{P}_1)(I+\hat{P}_2)=\hat{P}_2\hat{P}_1-\hat{P}_1\hat{P}_2=2\hat{P}_2\hat{P}_1\neq 0.\]
\end{proof}

\begin{remark}
Note that the error correction program is not strictly conservative, this is because the cat state preparation used in syndrome measurement employs an adaptive strategy (parity check and repetition). Nevertheless, we can independently verify the cat state preparation as shown in \cref{sec-cat-state-1212245} and treat the cat state measurement as a black box, where one fault causes at most one error on data qubit (possibly with bit-flip on measurement outcome). With this abstraction, the overall error correction program becomes effectively non-adaptive.
\end{remark}

\subsection{Explanation of Error Propagation}
The requirement that the loop body $S$ \textbf{does not propagate errors} means that: 
for any state $\ket{\psi}$, any execution of $S$ with $s\leq t$ faults on $\ket{\psi}$ will produce a state in the subspace 
$$\mathcal{S}_s(E_v\ket{\psi})=\spanspace(\{Q E_v \ket{\psi}\,\,|\,\, Q \textup{ is a Pauli operator of weight}  \leq s\}),$$
for some leaf node $v$.

\subsection{Technical Lemma}
We will also use the following result:
\begin{lemma}\label{lemma-5232211}
For any conservative loop body $S$, let $\hat{P}_1,\ldots,\hat{P}_m$ be those Pauli operators in \cref{prop-5230556}. Then, for an arbitrary leaf node $v$, let $E'_v$ be the Kraus operator of an arbitrary faulty transition from the root node to node $v$. We have that
\[E'_v = Q \cdot \left(\frac{I+(-1)^{\gamma_1}\hat{P}_1}{2}\right)\cdots \left(\frac{I+(-1)^{\gamma_m}\hat{P}_m}{2}\right),\]
for some Pauli operator $Q$ and $\gamma_i\in\{0,1\}$.
\end{lemma}
\begin{proof}
Since the execution only contains Pauli faults, we can write $E'_v$ based on \cref{eq-5232121}, as the following interleaved product:
\[E'_v = Q_{2m+1}U_{m+1}Q_{2m} \left(\frac{I+\beta_m^v P_m}{2}\right) Q_{2m-1} U_m  \cdots Q_3 U_2 Q_2 \left(\frac{I+\beta_1^v P_1}{2}\right) Q_1 U_1,\]
where $Q_i$ are Pauli operators. Then, we commute all Clifford unitaries $U_i$ to the rightmost position:
\begin{align}
E'_v &=Q_{2m+1}Q'_{2m} \left(\frac{I+\beta_m^v \hat{P}_m}{2}\right) Q'_{2m-1} \cdots Q_2'\left(\frac{I+\beta_1^v \hat{P}_1}{2}\right) Q'_1 \cdot (U_{m+1}\cdots U_1)\nonumber \\
&=Q_{2m+1}Q'_{2m} \left(\frac{I+\beta_m^v \hat{P}_m}{2}\right) Q'_{2m-1} \cdots Q_2'\left(\frac{I+\beta_1^v \hat{P}_1}{2}\right) Q'_1,
\end{align}
where the second equality is because \cref{eq-5232201}.
Then, we commute all Pauli operators $Q'_i$ to the leftmost position:
\[E'_v=(Q_{l+1}Q'_l\cdots Q'_1)\cdot \left(\frac{I+\beta_m^{'v} \hat{P}_m}{2}\right)\cdots \left(\frac{I+\beta_1^{'v} \hat{P}_1}{2}\right).\]
This completes the proof.
\end{proof}


\subsection{Valuation of Symbols}
Let \(\textup{Symb}\) be the alphabet set containing all possible symbols that can be used in the symbolic execution.
We define the \textbf{valuation} \(V\) as a function from the set \textup{Symb} to concrete values in \(\{0,1\}\). 

Suppose \(\langle S, \widetilde{\sigma},\widetilde{\rho},p,\varphi,\widetilde{F}\rangle\) is a symbolic configuration, \(\{s_1,\ldots,s_m\}\) is the set containing all symbols occurring in \(\widetilde{\sigma},\widetilde{\rho},\varphi,\widetilde{F}\). Let \(V\) be a valuation on symbols \(\{s_1,\ldots,s_m\}\). Then, we define \(V(\widetilde{\rho})\) as the concrete quantum state by replacing each symbols \(s_i\) in the phase of stabilizers of \(\widetilde{\rho}\) (see \cref{def-1281938}). Similarly, we define the \(V(\widetilde{\sigma})\), \(V(\varphi)\) and \(V(\widetilde{F})\) as the concrete classical state, truth value and integer by replacing all symbols by their valuations. We say concrete quantum states \(\rho\) is \textbf{representable} by a symbolic states \(\widetilde{\rho}\) if there exists a valuation \(V\) such that \(V(\widetilde{\rho})\sqsupseteq \rho\).

\subsection{Monotonicity}
We show some \textbf{monotonicity} properties of our symbolic faulty transitions, which can be easily checked. Suppose 
\[\langle S,\widetilde{\sigma},\widetilde{\rho},p,\varphi,\widetilde{F}\rangle \twoheadrightarrow \langle S',\widetilde{\sigma}',\widetilde{\rho}',p',\varphi',\widetilde{F}'\rangle,\]
and \(V\) is any valuation, then we have the following results.
\begin{enumerate}
    \item If \(V(\varphi')=\textup{True}\), then \(V(\varphi)=\textup{True}\).
    \item \(V(\widetilde{F}')\geq V(\widetilde{F})\).
    \item \(S'\) is a sub-program of \(S\).
    \item \(p\geq p'>0\)
\end{enumerate}

\subsection{Soundness and Completeness of Symbolic Faulty Transitions}
In this subsection, we prove an auxiliary result: the soundness and completeness of our symbolic execution with symbolic faulty transitions (note that this is NOT equivalent to the soundness and completeness of our fault-tolerance verification), provided that the loops are either 
\textbf{1)} memory-less, \textbf{2)} conservative and the loop body does not propagate errors. 

If \(\langle S,\sigma,\rho\rangle\) and \(\langle S',\sigma',\rho'\rangle\) are two concrete classical-quantum configurations, we introduce the following notation:
\[\langle S,\sigma,\rho\rangle \Rightarrow^s \langle S', \sigma',\rho' \rangle.\]
where \(\Rightarrow^s\) refers to the transition path consisting of arbitrary number of \(\rightarrow\) and at most \(s\) \(\leadsto\). We call \(\Rightarrow\) the \textbf{mixed transition}.

Then, we introduce the concepts of soundness and completeness of quantum symbolic faulty execution.
\begin{definition}[Soundness and Completeness of Symbolic Faulty Transitions~\footnote{Note that this is NOT equivalent to the soundness and completeness of our fault-tolerance verification given in \cref{def-1310223}.}]\label{def-1310351}
\begin{itemize}
    \item \textbf{Soundness}: If \(\langle S,\widetilde{\sigma},\widetilde{\rho},p,\varphi,\widetilde{F}\rangle \twoheadrightarrow^{*} \langle \downarrow,\widetilde{\sigma}',\widetilde{\rho}',p',\varphi',\widetilde{F}'\rangle\), then for any valuation \(V\) such that \(V(\varphi')=\textup{True}\), we have
    \[\langle S,V(\widetilde{\sigma}),pV(\widetilde{\rho})\rangle \Rightarrow^s \langle \downarrow, V(\widetilde{\sigma}'),p' V(\widetilde{\rho}') \rangle,\]
    where \(s\leq V(\widetilde{F}')-V(\widetilde{F})\).
    \item \textbf{Completeness}: Suppose \(s\) is a positive integer and
    \[\langle S,\sigma,\rho \rangle \Rightarrow^{s} \langle \downarrow, \sigma',\rho' \rangle.\]
    If \(\rho\) is representable by \(\widetilde{\rho}\), then, there exists a set of symbolic faulty transition paths
    \[\{\langle S,\widetilde{\sigma},\widetilde{\rho},1,\textup{True},0\rangle \twoheadrightarrow^{*} \langle \downarrow,\widetilde{\sigma}'_i,\widetilde{\rho}'_i,p'_i,\varphi'_i,\widetilde{F}'_i\rangle\}_i,\]
    where \(\widetilde{\sigma}=\sigma\) (i.e., there are no symbols in \(\widetilde{\sigma}\)), \(p_i'>0\) and there exists a set of valuations \(\{V_i\}_i\) such that \(V_i(\widetilde{\rho})=\rho\), \(V_i(\widetilde{\sigma}_i')=\sigma'\), \(V_i(\widetilde{F}_i')\leq s\), \(V_i(\varphi'_i)=\textup{True}\) and 
    \[\sum_i\supp(V_i(\widetilde{\rho}_i'))\supseteq \supp(\rho').\]
\end{itemize}
\end{definition}

\subsubsection{Soundness}
\begin{lemma}\label{lemma-1310323}
Our symbolic faulty transitions with rule \((\textup{SF-RU}')\) are sound (see \cref{def-1310351}).
\end{lemma}
\begin{proof}
We use induction on the structure of program \(S\). If \(S=\downarrow\), then the soundness holds trivially. Now assume the soundness holds for all sub-programs in \(S\).

Suppose
\begin{equation}\label{eq-1290218}
\langle S,\widetilde{\sigma},\widetilde{\rho},p,\varphi,\widetilde{F}\rangle \twoheadrightarrow
\langle S',\widetilde{\sigma}',\widetilde{\rho}',p',\varphi',\widetilde{F}'\rangle\twoheadrightarrow^{*}
\langle \downarrow,\widetilde{\sigma}'',\widetilde{\rho}'',p'',\varphi'',\widetilde{F}''\rangle.
\end{equation}
Due to the monotonicity, \(S'\) is a sub-program of \(S\). Therefore, by induction hypothesis, for any valuations \(V\) such that \(V(\varphi'')=\textup{True}\), we have
\[\langle S',V(\widetilde{\sigma}'),p'V(\widetilde{\rho}')\rangle \Rightarrow^{s'} \langle \downarrow,V(\widetilde{\sigma}''),p''V(\widetilde{\rho}'')\rangle,\]
where \(s'=V(\widetilde{F}'')-V(\widetilde{F}')\). Then it suffices to prove the following holds true
\begin{equation}\label{eq-1290259}
\langle S,V(\widetilde{\sigma}),pV(\widetilde{\rho})\rangle \Rightarrow^{s} \langle S',V(\widetilde{\sigma}'),p'V(\widetilde{\rho}')\rangle,
\end{equation}
where \(s=V(\widetilde{F}')-V(\widetilde{F})\).

Then, we check all the possible types of the first symbolic transition in \cref{eq-1290218}. For simplicity, here we check the unitary transform statement \(U(\overline{q})\) with transition \((\textup{SF-UT})\) and the repeat-until-success statement \(\repeatuntil{S}{b}\) with transition \((\textup{SF-RU}')\), where the others can be treated similarly.

\textbf{1)} Assume the last symbolic transition is \((\textup{SF-UT})\) corresponding to the statement \(U(\overline{q})\). Then we have \(S=U(\overline{q});S'\), \(\widetilde{\sigma}'=\widetilde{\sigma}\), \(p=p'\), \(\varphi=\varphi'\) and \(\widetilde{F}+e=\widetilde{F}'\), where
\[(e,\widetilde{\rho}')=\textup{EI}(\overline{q},\textup{UT}(U,\overline{q},\widetilde{\rho})).\]

Now, if \(V(e)=0\), we know that \(s=0\). Thus we have
\[V(\widetilde{\rho}')=V(\textup{UT}(U,\overline{q},\widetilde{\rho}))=U_{\overline{q}}V(\widetilde{\rho})U^\dag_{\overline{q}},\]
where the first equality is by the definition of \(\textup{EI}\) and the second equality is by the definition of \(\textup{UT}\). Obviously, we have
\[\langle U(\overline{q});S', V(\widetilde{\sigma}),pV(\widetilde{\rho})\rangle\rightarrow \langle  S',V(\widetilde{\sigma}), pU_{\overline{q}}V(\widetilde{\rho})U_{\overline{q}}^\dag\rangle =\langle  S',V(\widetilde{\sigma}'), p'V(\widetilde{\rho}') \rangle ,\]
therefore \cref{eq-1290259} follows immediately.

If \(V(e)=1\), we know that \(s=1\). Thus we have
\[V(\widetilde{\rho}')=P_{\overline{q}}V(\textup{UT}(U,\overline{q},\widetilde{\rho}))P_{\overline{q}}=P_{\overline{q}}U_{\overline{q}}V(\widetilde{\rho})U^\dag_{\overline{q}}P_{\overline{q}},\]
where \(P_{\overline{q}}\) is some Pauli operators acting on \(\overline{q}\), the first equality is by the definition of \(\textup{EI}\) and the second equality is by the definition of \(\textup{UT}\). Obviously, we have
\[\langle U(\overline{q});S', V(\widetilde{\sigma}),pV(\widetilde{\rho})\rangle\leadsto \langle  S',V(\widetilde{\sigma}), pP_{\overline{q}}U_{\overline{q}}V(\widetilde{\rho})U_{\overline{q}}^\dag P_{\overline{q}}\rangle=\langle  S',V(\widetilde{\sigma}'), p'V(\widetilde{\rho}') \rangle ,\]
therefore \cref{eq-1290259} follows immediately.

\textbf{2)} Assume the first symbolic transition is \((\textup{SF-RU}')\) corresponding to the statement \(\repeatuntil{S_1}{b}\) where \(S_1\) is a sub-program of \(S\). Then, 
\[S=\repeatuntil{S_1}{b};S'.\]
Suppose we have
\begin{equation}\label{eq-1290429}
\langle S_1, \widetilde{\sigma},\widetilde{\rho},p,\varphi,\widetilde{F}\rangle\twoheadrightarrow^{*} \langle \downarrow, \widetilde{\sigma}_1,\widetilde{\rho}_1,p_1,\varphi_1,\widetilde{F}_1\rangle.
\end{equation}
Then, by the rule \((\textup{SF-RU}')\), we have
\begin{equation*}
\langle S,\widetilde{\sigma},\widetilde{\rho},p,\varphi,\widetilde{F}\rangle\twoheadrightarrow \langle S',\widetilde{\sigma}_1,\widetilde{\rho}_1,p_1,\varphi_1\land  \widetilde{\sigma}_1(b),\widetilde{F}_1\rangle.
\end{equation*}
Therefore, we know that \(\widetilde{\sigma}'=\widetilde{\sigma}_1\), \(\widetilde{\rho}'=\widetilde{\rho}_1\), \(p'=p_1\), \(\varphi'=\varphi_1\land \widetilde{\sigma}_1(b)\) and \(\widetilde{F}'=\widetilde{F}_1\). Note that \(V(\varphi')=\textup{True}\). 
By the induction hypothesis applied on \cref{eq-1290429} (using the structural induction hypothesis on the sub-program \(S_1\)), and the fact that \(V(\varphi_1)=\textup{True}\), we have
\begin{equation}\label{eq-1290512}
\langle S_1,V(\widetilde{\sigma}),pV(\widetilde{\rho})\rangle\Rightarrow^{s}\langle \downarrow , V(\widetilde{\sigma}_1), p_1V(\widetilde{\rho}_1)\rangle=\langle \downarrow, V(\widetilde{\sigma}'), p'V(\widetilde{\rho}')\rangle
\end{equation}
where \(s=V(\widetilde{F}_1)-V(\widetilde{F})=V(\widetilde{F}')-V(\widetilde{F})\). Therefore, by the rule \(\textup{(RU)}\), we have
\begin{align}
&\quad \,\,\, \langle \repeatuntil{S_1}{b};S',V(\widetilde{\sigma}),pV(\widetilde{\rho})\rangle \nonumber\\ 
&\rightarrow \langle S_1;\ifelse{b}{\downarrow}{\{\repeatuntil{S_1}{b}\}};S',V(\widetilde{\sigma}),pV(\widetilde{\rho}) \rangle \label{eq-1290509}\\
&\Rightarrow^s  \langle \ifelse{b}{\downarrow}{\{\repeatuntil{S_1}{b}\}}; S',V(\widetilde{\sigma}'),p'V(\widetilde{\rho}') \rangle\label{eq-1290510}\\
& \rightarrow \langle S',V(\widetilde{\sigma}'),p'V(\widetilde{\rho}')\rangle \label{eq-1290511}
\end{align}
where \cref{eq-1290509} is by the rule \((\textup{RU})\), \cref{eq-1290510} is by \cref{eq-1290512}, \cref{eq-1290511} is by the rule \((\textup{CT})\) and the fact that \(V(\widetilde{\sigma}')\models b\) since \(V(\varphi')=V(\varphi_1\land \widetilde{\sigma}_1(b))=\textup{True}\). This is exactly what we want (i.e., \cref{eq-1290259}).

Therefore, by the structural induction, we conclude that the soundness holds for any program \(S\).
\end{proof}

\subsubsection{Completeness.}
Now, we prove the completeness.

A mixed transition path is called \textbf{non-repeated} (and denoted by \(\Rightarrow^{*}_{\textup{nr}}\)) if any \((\textup{RU})\) rule is followed by an execution of loop body and then a \((\textup{CT})\) rule, i.e.,
\begin{align}
\langle \repeatuntil{S}{b},\sigma,\rho\rangle &\rightarrow \langle S;\ifelse{b}{\downarrow}{\{\repeatuntil{S}{b}\}},\sigma,\rho\rangle \nonumber \\
&\Rightarrow^{*}_{\textup{nr}} \langle \ifelse{b}{\downarrow}{\{\repeatuntil{S}{b}\}},\sigma',\rho'\rangle \label{eq-1290600} \\
&\rightarrow \langle \downarrow, \sigma',\rho'\rangle,\nonumber
\end{align}
where the transition path in \cref{eq-1290600} is also non-repeated.
Formally, this can be defined inductively on the structure of the starting program of the path.

We will use the following lemma.
\begin{lemma}\label{eq-1302251}
Suppose the loops are either \textbf{1)} memory-less or \textbf{2)} conservative, and the loop body does not propagate errors. If
\begin{equation}\label{eq-1291351}
\langle S,\sigma,\rho\rangle \Rightarrow^s \langle \downarrow,\sigma',\rho'\rangle,
\end{equation}
then there exists a collection of non-repeated transition paths
\[\{\langle S,\sigma,\rho\rangle \Rightarrow^s_{\textup{nr}} \langle \downarrow,\sigma',\rho'_i\rangle\}_i,\]
such that \(\sum_i \supp(\rho'_i)\supseteq \supp(\rho')\).
\end{lemma}
\begin{proof}
We use induction on the structure of \(S\). For the case that \(S=\downarrow\), it holds trivially. Now suppose it holds for all sub-programs of \(S\).

If the first transition of \cref{eq-1291351} is not the \((\textup{RU})\) rule, we get a sub-program after the first transition, and the existence of non-repeated path follows immediately by the induction hypothesis on this sub-program. Now, suppose \cref{eq-1291351} is of the form
\begin{align}\label{eq-1291413}
\langle \repeatuntil{S_1}{b};S_2,\sigma,\rho\rangle \Rightarrow^{s}&\langle \downarrow,\sigma',\rho'\rangle.
\end{align}

First suppose \(S_2\) is not \(\downarrow\).
Then let 
\[\langle \repeatuntil{S_1}{b};S_2,\sigma,\rho\rangle \Rightarrow^{s_1} \langle S_2,\sigma'',\rho'' \rangle \Rightarrow^{s-s_1} \langle \downarrow,\sigma',\rho'\rangle.\]
Note that \(S=\repeatuntil{S_1}{b};S_2\), which is a concatenation of \(\repeatuntil{S_1}{b}\) and \(S_2\). Thus \(\repeatuntil{S_1}{b}\) is also a sub-program, which, by the induction hypothesis, 
has a collection of non-repeated paths
\begin{equation}\label{eq-1301249}
\{\langle \repeatuntil{S_1}{b},\sigma,\rho\rangle \Rightarrow^{s_1}_{\textup{nr}} \langle \downarrow,\sigma'',\rho_i''\rangle\}_i,
\end{equation}
such that \(\sum_i \supp(\rho_i'')\supseteq \supp(\rho'')\)
and similarly \(S_2\) also has a collection of non-repeated paths
\begin{equation}\label{eq-1301250}
\{\langle S_2,\sigma'',\rho''\rangle\Rightarrow^{s-s_1}_{\textup{nr}}\langle \downarrow,\sigma',\rho_i'\rangle\},
\end{equation}
such that \(\sum_i\supp(\rho_i')\supseteq\supp(\rho')\).
Then, obviously, the collection of transition path
\begin{equation}\label{eq-1301326}
\{\langle \repeatuntil{S_1}{b};S_2,\sigma,\rho\rangle \Rightarrow^{s_1}_{\textup{nr}} \langle S_2,\sigma'',\rho''_i\rangle\Rightarrow^{s-s_1}_{\textup{nr}} \langle \downarrow,\sigma',\rho_{i,j}'\rangle\}_{i,j}
\end{equation}
satisfies \(\sum_{i,j}\supp(\rho_{i,j})\supseteq \supp(\rho')\), where \(\rho'_{i,j}\) refers to the quantum state corresponding to the output of the transition path by concatenating the \(i\)-th path in \cref{eq-1301249} and \(j\)-th path in \cref{eq-1301250}.

Now suppose \(S_2=\downarrow\), and thus we consider the following case
\begin{align}\label{eq-1291432}
\langle \repeatuntil{S_1}{b},\sigma,\rho\rangle \Rightarrow^{s} \langle \downarrow,\sigma',\rho'\rangle.
\end{align}
Assume that in \cref{eq-1291432}, the loop body \(S_1\) is repeated for \(l\) times. This means
\begin{equation}\label{eq-1301315}
\langle \underbrace{S_1;S_1;\cdots;S_1}_{l}, \sigma_0,\rho_0\rangle\Rightarrow^{s_1}\langle \underbrace{S_1;\cdots;S_1}_{l-1}, \sigma_1,\rho_1\rangle \Rightarrow^{s_2} \cdots\Rightarrow^{s_{l}} \langle \downarrow, \sigma_l,\rho_l\rangle,
\end{equation}
where \(\sigma_0=\sigma\), \(\rho_0=\rho\), \(\sigma_l=\sigma'\), \(\rho_l=\rho'\) and \(s_1+\cdots+s_l=s\).
Note that \(S_1\) is a sub-program of \(S\). By the induction hypothesis, we know that each transition path in \cref{eq-1301315} can be covered by a collection of non-repeated paths. Then, we can concatenate these paths (similar to that in \cref{eq-1301326}) to obtain
\begin{equation}\label{eq-1301433}
\{\langle \underbrace{S_1;S_1;\cdots;S_1}_{l}, \sigma_0,\rho_0\rangle\Rightarrow^{s_1}_{\textup{nr}}\langle \underbrace{S_1;\cdots;S_1}_{l-1}, \sigma_1,\rho'_1\rangle \Rightarrow^{s_2}_{\textup{nr}} \cdots\Rightarrow^{s_{l}}_{\textup{nr}} \langle \downarrow, \sigma_l,\rho'_l\rangle\}.
\end{equation}
such that \(\sum \supp(\rho'_l)\supseteq \supp(\rho_l)\).

\textbf{1)} If the loop body \(S_1\) is memory-less, the execution of \(S_1\) does not depend on the input states. Combining with \cref{eq-1301433}, we immediately have
\[\{\langle S_1,\sigma_0,\rho_0\rangle \Rightarrow_{\textup{nr}}^{s_l}\langle \downarrow,\sigma_l,\rho'_l\rangle\},\]
is a set of non-repeated paths such that \(\sum \supp(\rho'_l)\supseteq \supp(\rho_l)\).
Note that \(s_l\leq s\), we have
\[\{\langle \repeatuntil{S_1}{b},\sigma,\rho\rangle \Rightarrow^{s}_{\textup{nr}} \langle \downarrow,\sigma_l,\rho_l'\rangle\},\]
such that \(\sum \supp(\rho_l')\supseteq \supp(\rho_l)\), where \(S_1\) is executed only once (note that the Boolean expression \(b\) can be satisfied by \(\sigma_l\) due to \cref{eq-1291432}) and the fact that \(\sigma_l=\sigma'\).

\textbf{2)} Suppose the loop body \(S_1\) is conservative and does not propagate errors. Let 
\begin{equation}\label{eq-1301552}
\langle \underbrace{S_1;S_1;\cdots;S_1}_{l}, \sigma_0,\rho_0\rangle\Rightarrow^{s_1}_{\textup{nr}}\langle \underbrace{S_1;\cdots;S_1}_{l-1}, \sigma_1,\rho'_1\rangle \Rightarrow^{s_2}_{\textup{nr}} \cdots\Rightarrow^{s_{l}}_{\textup{nr}} \langle \downarrow, \sigma_l,\rho'_l\rangle
\end{equation}
be any transition path in \cref{eq-1301433}. It suffices to prove that the support of \(\rho'_l\) in \cref{eq-1301552} can be covered by a collection of non-repeated paths such that each of these paths
\textbf{1)} ends with classical state \(\sigma_l\);
\textbf{2)} executes \(S_1\) only once;
\textbf{3)} contains at most \(s=s_1+\cdots+s_l\) faults. 

Since the classical variables in \(S_1\) are reset in each iteration, the transition paths starting with \(S_1\) do not depend on the classical state. Suppose \(N_L\) is the set of leaf nodes of the non-repeated transition paths starting with \(S_1\). Let 
\[\{\mathcal{E}_v\,|\, v\in N_L\}\]
be the set of quantum channels such that each \(\mathcal{E}_v\) corresponds to the ideal transition path to leaf node \(v\). Therefore, the quantum channel corresponding to the transition path in \cref{eq-1301552} is of the form
\[\mathcal{E}^{s_l}_{v_l}\circ\cdots\circ\mathcal{E}^{s_2}_{v_2}\circ\mathcal{E}^{s_1}_{v_1},\]
where \(\mathcal{E}_{v_i}^{s_i}\) is a quantum channel corresponding to a faulty transition path with at most \(s_i\) Pauli faults starting with \(S_1\) and ending at the leaf node \(v_i\).

Let $\rho_0=\ketbra{\psi_0}{\psi_0}$, $\rho'_i=\ketbra{\psi'_i}{\psi'_i}$, $\mathcal{E}_{v_i}^{s_i}=E_{v_i}^{s_i}(\cdot) E_{v_i}^{s_i\dag}$. We can see that 
\[E_{v_1}^{s_1}\ket{\psi_0}=\ket{\psi'_1},\quad  E_{v_2}^{s_2}\ket{\psi'_1}=\ket{\psi'_2},\quad \ldots ,\quad E_{v_l}^{s_l}\ket{\psi'_{l-1}}=\ket{\psi'_l}. \]
Then, since errors do not propagate within the loop body, we have
\[\ket{\psi_{i}'}=E_{v_i}^{s_i}\ket{\psi_{i-1}'}=Q_i E_{v_i'}\ket{\psi_{i-1}'},\]
where $v'_i$ is some leaf node and $Q_i$ is some Pauli operator of weight at most $s_i$. This means
\begin{align}
\ket{\psi_l'}&=E_{v_l}^{s_l} \ket{\psi_{l-1}'}= E_{v_l}^{s_l} \cdot \left[(Q_{l-1}E_{v'_{l-1}})\cdots (Q_1 E_{v'_{1}})\right]\ket{\psi_0}\nonumber \\
&=E_{v_l}^{s_l}\cdot (F_{l-1}\cdots F_1)\cdot  (Q_{l-1}\cdots Q_1)\ket{\psi_0},\label{eq-5230558}
\end{align}
where $F_i$ is a stabilizer projector obtained from $E_{v_{i}'}$ by Pauli conjugation. By \cref{prop-5230556}, we can see that operators $\{F_i\}_i$ is also of the form
\begin{equation}\label{eq-5232213}
F_i=\left(\frac{I+(-1)^{\gamma^i_1}\hat{P}_1}{2}\right)\cdots \left(\frac{I+(-1)^{\gamma^i_m}\hat{P}_m}{2}\right).
\end{equation}
Then, by \cref{lemma-5232211}, we can see that
\[E_{v_l}^{s_l}= Q_l F_l,\]
for some Pauli operator $Q_l$ and some stabilizer projector $F_l$ that is also of the form \cref{eq-5232213}, sharing the same Pauli stabilizers $\hat{P}_1,\ldots,\hat{P}_m$ while with (possibly) different phases.
This means $F_l F_i$ is either $0$ or $F_l$. Therefore, \[E_{v_l}^{s_l}F_i=Q_l F_l F_i= c_i Q_l F_l=c_i E_{v_l}^{s_l},\]
for some $c_i\in\{0,1\}$.
Therefore, \cref{eq-5230558} becomes
\[\ket{\psi_l'}=c_1\cdots c_{l-1}\cdot  E_{v_l}^{s_l}\cdot  (Q_{l-1}\cdots Q_1)\ket{\psi_0}.\]
Since the weight of $Q_{l-1}\cdots Q_{1}$ is at most $s_1+\cdots +s_{l-1}$, we can see that the faulty transition that first inserts $Q_{l-1}\cdots Q_1$ to $\ket{\psi_0}$ and then applies the faulty transition $\mathcal{E}_{v_{l}}^{s_l}$ suffices to produce a non-zero state that is proportional to $\ket{\psi'_l}$ (provided that $\ket{\psi'_l}$ is non-zero). Moreover, \textbf{1)} since the transition ends in the leaf node $v_l$, it ends with the classical state $\sigma_l$, \textbf{2)} the loop body \(S_1\) is executed only once, \textbf{3)} the number of faults is at most $s_1+\cdots+s_l$. Recalling \cref{eq-1301552} and the texts there, we can conclude that for the case that the loop body \(S_1\) is conservative and does not propagate errors, the induction hypothesis also holds.
\end{proof}

Now, we prove the completeness of our symbolic execution.

\begin{lemma}\label{lemma-1310325}
If loops in a \cqprog{}\  are either \textbf{1)} memory-less or \textbf{2)} conservative and the loop body does not propagate errors, our symbolic faulty transitions with rule \((\textup{SF-RU}')\) are complete (see \cref{def-1310351}).
\end{lemma}
\begin{proof}
Suppose \(s\) is a positive integer and
    \[\langle S,\sigma,\rho \rangle \Rightarrow^{s} \langle \downarrow, \sigma',\rho' \rangle.\]
    where \(\rho\) is representable by \(\widetilde{\rho}\).

Using \cref{eq-1302251}, we can always find a set of collection of non-repeated transition paths
\begin{equation*}
    \{\langle S,\sigma,\rho\rangle \Rightarrow_{\textup{nr}}^{s}\langle \downarrow,\sigma',\rho_i' \rangle\}_i,
\end{equation*}
such that \(\sum_i \supp(\rho_i')\supseteq \supp(\rho')\). 

Therefore, it suffices to show that for each non-repeated faulty transition path
\begin{equation}\label{eq-1310225}
    \langle S,\sigma,\rho\rangle \Rightarrow_{\textup{nr}}^{s}\langle \downarrow,\sigma',\rho' \rangle,
\end{equation}
if \(\rho\) is representable by \(\widetilde{\rho}\), then there exists a symbolic faulty transition path
\[\langle S,\widetilde{\sigma},\widetilde{\rho},1,\textup{True},0\rangle \twoheadrightarrow^{*} \langle \downarrow,\widetilde{\sigma}',\widetilde{\rho}',p',\varphi',\widetilde{F}'\rangle,\]
where \(\widetilde{\sigma}=\sigma\) (i.e., there are no symbols in \(\widetilde{\sigma}\)), such that \(p'>0\) and there exists a valuation \(V\) such that \(V(\widetilde{\rho})=\rho\), \(V(\widetilde{\sigma}')=\sigma'\), \(\supp(V(\widetilde{\rho}'))\supseteq \supp(\rho')\), \(V(\varphi')=\textup{True}\) and \(V(\widetilde{F}')\leq s\).

We use induction on the structure of the \cqprog{}\ \(S\). When \(S=\downarrow\), it holds trivially. Now assume it holds for all the sub-programs of \(S\).

Then, we investigate the first transition in \cref{eq-1310225}. Here we only demonstrate the case where the first transition is the rule \((\textup{RU})\) for the repeat-until-success statement. This is the most non-trivial case and will make use of the non-repeatedness of the transition path in \cref{eq-1310225}. Other cases are more straightforward and can be handled similarly.

Suppose \(S=\repeatuntil{S_1}{b};S_2\). If \(S_2\neq \downarrow\), then both \(\repeatuntil{S_1}{b}\) and \(S_2\) are sub-programs. We can directly use the induction hypothesis on them to obtain two symbolic transition paths. Then by concatenating these two paths (with minor modifications), we obtain the desired symbolic transition path.

Therefore, suppose \(S=\repeatuntil{S_1}{b}\). We have
\[\langle \repeatuntil{S_1}{b},\sigma,\rho\rangle \rightarrow \langle S_1;\ifelse{b}{\downarrow}{\{\repeatuntil{S_1}{b}\}},\sigma,\rho\rangle.\]
Since the transition path in \cref{eq-1310225} is non-repeated, the \(S_1\) is only executed once. Therefore, we can know that
\[\langle S_1,\sigma,\rho\rangle \Rightarrow^{s}_{\textup{nr}}\langle \downarrow,\sigma',\rho'\rangle,\]
and \(\sigma'\models b\). Since \(S_1\) is a sub-program of \(S\), by induction hypothesis, there exists a symbolic transition path
\[\langle S_1,\widetilde{\sigma},\widetilde{\rho},1,\textup{True},0\rangle\twoheadrightarrow^{*}\langle \downarrow, \widetilde{\sigma}',\widetilde{\rho}',p',\varphi',\widetilde{F}' \rangle,\]
where \(\widetilde{\sigma}=\sigma\), such that \(p'>0\) and there exists a valuation \(V\) such that \(V(\widetilde{\rho})=\rho\), \(V(\widetilde{\sigma}')=\sigma'\), \(\supp(V(\widetilde{\rho}'))\supseteq \supp(\rho')\), \(V(\varphi')=\textup{True}\) and \(V(\widetilde{F}')\leq s\). Using this, we can obtain the following symbolic transition path using the rule \((\textup{SF-RU}')\):
\[\langle \repeatuntil{S_1}{b},\widetilde{\sigma},\widetilde{\rho},1,\textup{True},0\rangle\twoheadrightarrow\langle \downarrow, \widetilde{\sigma}',\widetilde{\rho}',p',\varphi'\land \widetilde{\sigma}'(b),\widetilde{F}' \rangle.\]
Note that \(V(\widetilde{\sigma}')=\sigma'\) and \(\sigma'\models b\), which means \(V(\varphi'\land\widetilde{\sigma}'(b))=\textup{True}\). Therefore, the induction hypothesis also holds for \(S\).
\end{proof}

\subsection{Soundness and Completeness of Fault-Tolerance Verification (Proof of Theorem~\ref{thm-compl-sound-1251851})}\label{sec-210412}
Now, we are able to prove the soundness and completeness of our fault-tolerance verification. 
We will first use our discretization theorems (i.e. \cref{lemma-dis-input-1272245} and \cref{lemma-disc-fautls-191645}) to reduce the verification task on continuous spaces to the verification task on discrete sets.
Then, roughly speaking, the soundness of fault-tolerance verification will use the completeness of symbolic faulty transitions (i.e., \cref{lemma-1310325}) and the completeness of fault-tolerance verification will use the soundness of symbolic faulty transitions (i.e., \cref{lemma-1310323}).

\begin{proof}
By our discretization theorems \cref{lemma-dis-input-1272245} and \cref{lemma-disc-fautls-191645}, the soundness and completeness of fault-tolerance verification is equivalent to that with stabilizer state inputs and Pauli faults. As demonstrated in \cref{sec-veri-ft-1200400}, our verification process uses symbolic stabilizer states to simulate all possible stabilizer state inputs that are needed for proving fault-tolerance. Then we use symbolic faulty transitions to (partially) simulate Pauli faulty transitions on the symbolic inputs.
We claim the fault-tolerance based on whether the output  symbolic states satisfy the fault-tolerance condition (i.e., \cref{eq-1100520}). More specifically, we assume the loops are either memory-less or conservative, where for conservative loops, we additionally perform a sub-verification for the error propagation property on its loop body using our symbolic execution (this is performed recursively if loops are nested) and claim ``fault-tolerance'' if both \cref{eq-1100520} holds and the sub-verification on the loop body passes.

\textbf{Soundness.}
First, we use induction on the structure of the \cqprog{}\  to prove that our sub-verification for the error propagation properties of the conservative loops is sound. This can be easily check by calling \cref{lemma-1310325}. Since if the errors propagate and increase in the loop body \(S\), then there must exists a faulty transition path witnessing the increasing of errors:
\begin{equation}\label{eq-1311520}
\langle S,\sigma,\rho\rangle \Rightarrow^{s} \langle \downarrow,\sigma',\rho'_{e}\rangle,
\end{equation}
where the corresponding ideal transition path ending in the same leaf node is
\[\langle S,\sigma,\rho\rangle \rightarrow^{*} \langle \downarrow ,\sigma',\rho'\rangle.\]
Here, \(\rho=\ketbra{\psi}{\psi}\) is a stabilizer state, \(\rho'=\ketbra{\psi'}{\psi'}\) (all the ideal transitions preserve purity, see \cref{remark-1311516}). Since \cref{eq-1311520} witnesses the increasing of errors, we have 
\begin{equation}\label{eq-1311535}
\supp(\rho'_e)\not\subseteq \mathcal{S}_{s}(\ket{\psi'}),
\end{equation}
where \(\mathcal{S}_{s}\) is the \(s\)-error space surrounding \(\ket{\psi'}\) (see \cref{def-1311407}). 
Then, by \cref{lemma-1310325}, there exists a set of symbolic faulty transition paths
\[\{\langle S,\widetilde{\sigma},\widetilde{\rho},1,\textup{True},0\rangle \twoheadrightarrow^{*} \langle \downarrow,\widetilde{\sigma}'_i,\widetilde{\rho}'_i,p'_i,\varphi'_i,\widetilde{F}'_i\rangle\}_i,\]
where \(\widetilde{\sigma}=\sigma\) (i.e., there are no symbols in \(\widetilde{\sigma}\)), \(p_i'>0\) and there exists a set of valuations \(\{V_i\}_i\) such that \(V_i(\widetilde{\rho})=\rho\), \(V_i(\widetilde{\sigma}_i')=\sigma'\), \(V_i(\widetilde{F}_i')\leq s\), \(V_i(\varphi'_i)=\textup{True}\) and 
\begin{equation}\label{eq-1311536}
\sum_i\supp(V_i(\widetilde{\rho}_i'))\supseteq \supp(\rho'_e).
\end{equation}
Combining \cref{eq-1311535} with \cref{eq-1311536}, we know that there must exists an index \(i\) such that 
\[\supp(V_i(\widetilde{\rho}'_i))\not\subseteq \mathcal{S}_s(\ket{\psi'}).\]
Then, since the SMT solver logically traverses all possible valuations, the ``bad'' faulty transition path will be found by comparing the faulty outputs (i.e., \(V_i(\widetilde{\rho}'_i)\)) with the ideal (fault-free) output state (i.e., \(\ket{\psi'}\)) using a similar formula as that in \cref{eq-1100520}.

Then, assume conversely that the gadget is not fault-tolerant but our verification claims ``fault-tolerance''. This means all sub-verifications of the loop bodies of conservative loops pass and the final fault-tolerance condition holds. Since we have already proved that our sub-verification is sound, the only situation is that there exists a faulty transition path that violates the final fault-tolerance property (see \cref{def-ft-1310455}). However, again by \cref{lemma-1310325}, there exists a set of symbolic faulty transition paths with a set of valuations that covers the support of the output quantum state violating the fault-tolerance condition. With a similar proof as above, we can conclude that this will be detected by comparing the support of the faulty symbolic output state with the ideal output state. Therefore, our verification tool will not claim ``fault-tolerance'', which is a contradiction.

\textbf{Completeness.}
For the completeness, we only consider the cases that the loops are memory-less~\footnote{The reason that we cannot prove completeness for the cases for conservative loops is: we currently cannot deny the possibility that the loop body propagates and increases errors but the overall program is still fault-tolerant.}.

Assume that the \cqprog{}\ \(S\) implements a gate gadget corresponding to the logical gate \(\overline{U}\) (other types of gadget can be handled similarly). 
Then, assume conversely that \(S\) is fault-tolerant but our verification claims ``non-fault-tolerant'' (note that our verification tool always terminates and returns either ``fault-tolerant'' or ``non-fault-tolerant'').
This means there exists a symbolic faulty transition path 
\[\langle S,\emptyset,\widetilde{\rho},1,\textup{True},\widetilde{F}\rangle \twoheadrightarrow^{*} \langle \downarrow,\widetilde{\sigma}',\widetilde{\rho}',p',\varphi',\widetilde{F}'\rangle,\]
where \(p'>0\) and a valuation \(V\) such that \(V(\widetilde{F})=r\), \(V(\widetilde{\rho})=\ketbra{\psi_r}{\psi_r}\) (where \(\ket{\psi_r}\) is a noisy version a logical stabilizer state \(\ket{\overline{\psi}}\) with \(r\) Pauli errors), \(V(\widetilde{F}')=r+s\) (where \(r+s\leq t\)), \(V(\varphi')=\textup{True}\), and
\begin{equation}\label{eq-1311636}
\supp(V(\widetilde{\rho}')) \not\subseteq \mathcal{S}_{r+s}(\overline{U}\ketbra{\overline{\psi}}{\overline{\psi}} \overline{U}^\dag),
\end{equation}
where \(\mathcal{S}_{r+s}(\overline{U}\ketbra{\overline{\psi}}{\overline{\psi}} \overline{U}^\dag)\) is the \((r+s)\)-error space surrounding \(\overline{U}\ketbra{\overline{\psi}}{\overline{\psi}} \overline{U}^\dag\) (see \cref{def-1311407}). Then, by \cref{lemma-1310323}, there must exists a concrete faulty transition
\[\langle S,\emptyset,V(\widetilde{\rho})\rangle \Rightarrow^s \langle \downarrow, V(\widetilde{\sigma}'),p' V(\widetilde{\rho}') \rangle.\]
However, since \(S\) is fault-tolerant (by assumption), we have
\[\supp(V(\widetilde{\rho}'))\subseteq \mathcal{S}_{r+s}(\overline{U}\ketbra{\overline{\psi}}{\overline{\psi}} \overline{U}^\dag),\]
which contradicts \cref{eq-1311636}.
\end{proof}

\section{Verifying Magic State Distillation (Proof of Theorem~\ref{thm-magic-d-210221})}\label{sec-210242}

\begin{proof}
First, it is already known (see, e.g., \cite{gottesman2024surviving}) that the magic distillation framework introduced in \cref{sec-veri-magic-state-1210027}, which uses the error correction process of distillation code \(\mathcal{D}\) with fault-tolerant stabilizer operations on a QECC \(\mathcal{C}\), is fault-tolerant. Therefore, provided that a magic state distillation is implemented in the two-party framework, it suffices to verify that
\begin{enumerate}
    \item any stabilizer operations on \(\mathcal{C}\) are fault-tolerant,
    \item the error correction process on \(D\) is ideal-case correct.
\end{enumerate}
The soundness for the verification of the first part is directly guaranteed by \cref{thm-compl-sound-1251851}. Then, we consider the soundness of the verification of the second part. 

First, we prove a similar result as that in \cref{lemma-dis-input-1272245} to discretize the input space to logical-\(Z\) basis with an additional state \(\ket{\overline{+}}=\sum_{i=0}^{2^k-1}\ket{\overline{i}}/\sqrt{2^{k}}\). Note that the proof is essentially covered by that in \cref{sec-210338} by setting \(s=0\), and therefore is omitted here.

After discretization, we use the completeness result (i.e., \cref{lemma-1310325}) of the symbolic faulty transitions (note that this is not equivalent to the completeness or soundness of verification). In particular, we only need to use the completeness (see \cref{def-1310351}) where the number of faults \(s\) during the execution is set to \(0\). Then, by a similar arguments as that in \cref{sec-210412}, we can conclude that the verification is sound (in fact, the proof in \cref{sec-210412} strictly covers the proof of ideal-case correctness, as it is just a sub-case of fault-tolerance).
\end{proof}

\section{Additional Examples of Bug Finding}\label{sec:bugfinding}
\subsection{Bad Ordering of Syndrome Measurements.}
\begin{figure}[ht]
    \centering
    \vspace{-5mm}
    \includegraphics[width=0.88\linewidth]{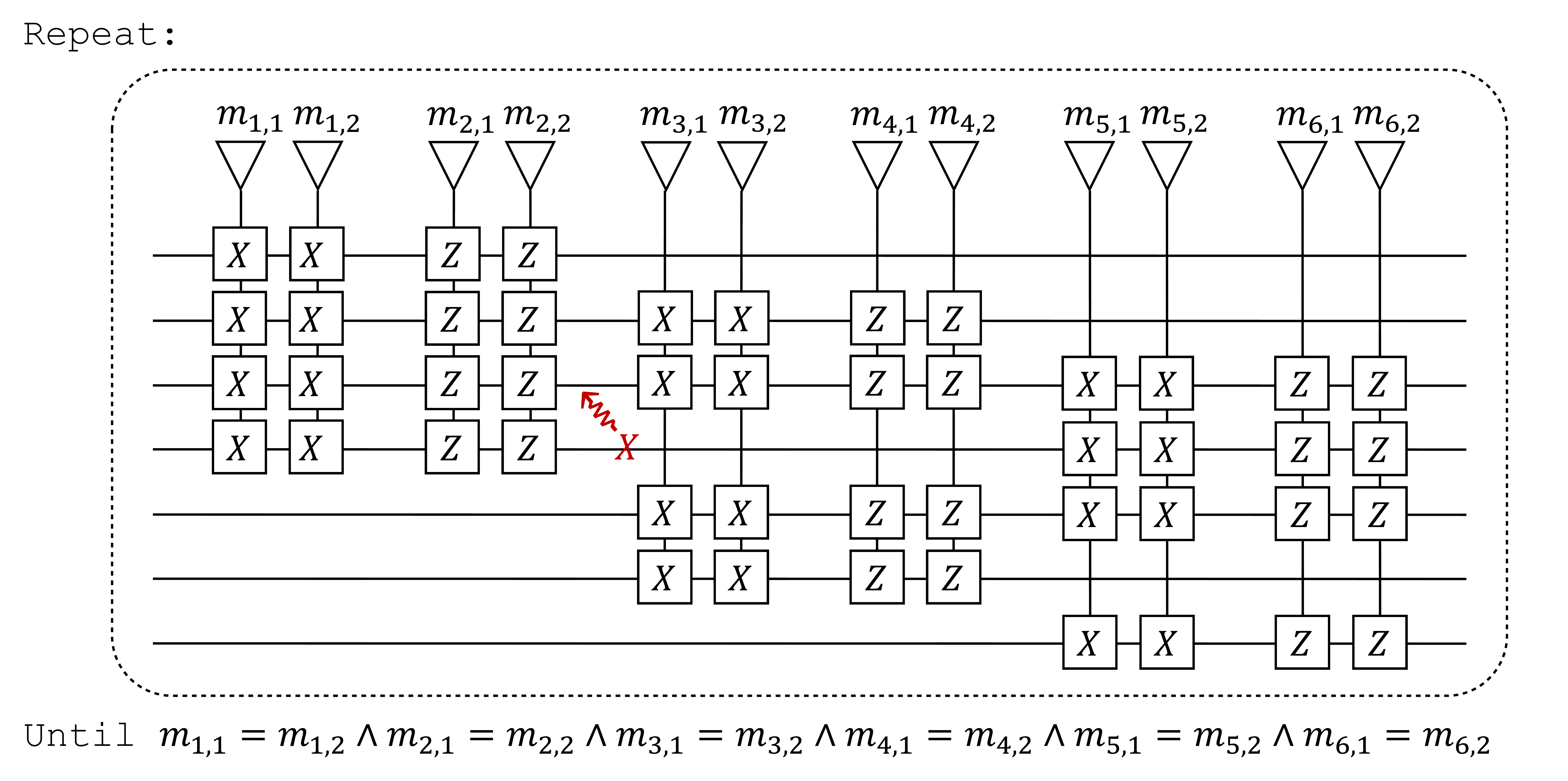}
    \caption{An example of bad syndrome measurement ordering of the color code~\cite{bombin2006topological}, where the Pauli measurements are implemented using cat states (see \cref{fig-multi-Pauli-cat-1181649}). A fault is found by our verification tool to witness its non-fault-tolerance.}
    \label{fig-bad-order-syndrome-1192330}
    \vspace{-3mm}
\end{figure}

In fault-tolerant error correction, syndrome measurements are repeated multiple times to reliably extract the syndrome for further error correction. It is worth noting that we must carefully design the ordering of syndrome measurements in the repetition. Specifically, the fault-tolerant method is to perform the full syndrome measurements for all stabilizers of the QEC and then repeat this process multiple times until we observe \(\lfloor d/2\rfloor + 1\) consecutive, all-agree, full syndrome results. In contrast, if the syndrome measurement for different stabilizers is performed sequentially --- where the repetition is conducted separately for the syndrome measurement of each stabilizer --- then a fault occurring after all repetitions of the syndrome measurement will be unnoticed in that syndrome result. An example of bad syndrome measurement ordering of the \([[7,1,3]]\) color code~\cite{bombin2006topological,rodriguez2024experimental} is shown in \cref{fig-bad-order-syndrome-1192330}, where our verification tool identified a fault that will cause a mismatch between the detected syndrome and the true syndrome, ultimately violating the fault-tolerance condition.

\subsection{Lack of Repetition.}
Consider the implementation of a fault-tolerant measurement gadget. Although we can use cat state measurement to ``transversally'' perform the multi-qubit Pauli measurement without propagating the errors, the measurement outcome itself is susceptible to error, i.e., a single fault can flip the measurement outcome. Therefore, to implement a fault-tolerant measurement, the cat state measurement is performed multiple times, each followed by a fault-tolerant error correction process on the data qubits. To tolerate up to \(t\) faults, it is enough to repeat it \(2t+1\) times and compute the majority of the outcomes. In contrast, we implement the logical-\(Z\) basis measurement of the \([[25,1,5]]\) rotated surface code~\cite{bombin2007optimal,horsman2012surface} with only \(2t=d-1=4\) repetitions. Using our verification tool, we found that two bit-flip faults in the outcomes of the first and second logical-\(Z\) measurements suffice to compromise the final measurement result, demonstrating its non-fault-tolerance.

\section{Verification of Hadamard and Phase Gates}\label{sec-5211304}
The Hadamard and phase gates can be implemented via quantum teleportation using the corresponding magic states $\ket{H}$ and $\ket{S}$ (see Theorem 13.2 in \cite{gottesman2024surviving}). The quantum teleportation procedure involves transversal Pauli gates and fault-tolerant Pauli measurement (which is similar to the $Z$-measurement previously verified in \cref{tab:veri-time-1132028}). Here, we provide the verification time of the fault-tolerant state preparation gadgets of $\ket{H}$ and $\ket{S}$ in \cref{table-5211308}.

\begin{table}[ht]
\vspace{-0pt}
\centering
\caption{Verification time of state preparation gadgets of $\ket{H}$ and $\ket{S}$.}\label{table-5211308}
\vspace{2mm}
\begin{tabular}{|c||c||cc|}
\hline
\multirow{2}{*}{\textbf{QECC}} &
  \multirow{2}{*}{\(\bm{[[n,k,d]]}\)} &
  \multicolumn{2}{c|}{\textbf{Time (s)}} \\ \cline{3-4} 
 &
   &
  \multicolumn{1}{c|}{\textbf{Preparation of $\ket{H}$}} &
  \textbf{Preparation of $\ket{S}$} \\ \hline
\multirow{2}{*}{Color Code~\cite{bombin2006topological,rodriguez2024experimental}} &
  [[7,1,3]] &
  \multicolumn{1}{c|}{2.67} &
  2.98 \\ \cline{2-4} 
 &
  [[17,1,5]] &
  \multicolumn{1}{c|}{46.37} &
   43.98 \\ \hline
\multirow{2}{*}{Rotated Surface Code~\cite{bombin2007optimal,horsman2012surface}} &
  [[9,1,3]] &
  \multicolumn{1}{c|}{3.11} &
  3.13 \\ \cline{2-4} 
 &
  [[25,1,5]] &
  \multicolumn{1}{c|}{205.94} &
  216.28 \\ \hline
\multirow{2}{*}{Toric Code~\cite{Kitaev2003anyons}} &
  [[18,2,3]] &
  \multicolumn{1}{c|}{3.08} &
  3.24 \\ \cline{2-4} 
 &
  [[50,2,5]] &
  \multicolumn{1}{c|}{402.43} &
   465.52 \\ \hline
\end{tabular}
\end{table}

\section{Fault-Tolerant Gadgets of the [[7,1,3]] Color Code\texorpdfstring{~\cite{bombin2006topological,rodriguez2024experimental}}{}}\label{sec:713-cc}

\cref{fig:713-cc} shows the schema of the [[7,1,3]] color code. The generators are 
\[\langle Z_1Z_2Z_3Z_4,X_1X_2X_3X_4,Z_2Z_3Z_5Z_6,X_2X_3X_5X_6,Z_3Z_4Z_5Z_7,X_3X_4X_5X_7\rangle.\] The logical $Z$ operator is $Z_L=Z_1Z_2Z_6$.

\cref{fig:713-gadgets} shows all four fault-tolerant gadgets of the [[7,1,3]] color code studied in \cref{sec-case-study-1221814}. Here, $t=\lfloor \frac{d-1}{2} \rfloor=1$, and all multi-Pauli measurements are performed with the cat state measurement. \cref{fig:713-cc-prep0} only shows the preparation of $\ket{\overline{0}}$. To prepare $\ket{\overline{1}}$, change the last measurement in each block ($Z_1Z_2Z_6=Z_L$) to $-Z_L$.

\begin{figure}[ht]
    \centering
    \includegraphics[width=0.4\linewidth]{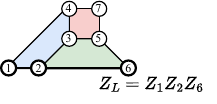}
    \vspace{-3mm}
    \caption{[[7,1,3]] color code schema}
    \label{fig:713-cc}
\end{figure}

\begin{figure}[ht]
    \centering
    \begin{subfigure}[b]{0.9\linewidth}
        \centering
        \includegraphics[width=1.0\linewidth]{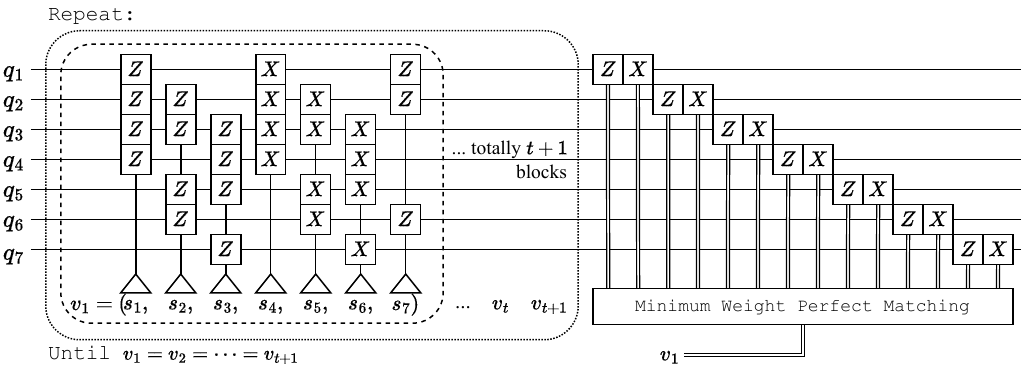}
        \caption{$\ket{\overline{0}}$ preparation gadget}
        \label{fig:713-cc-prep0}
    \end{subfigure}
    \begin{subfigure}[b]{0.25\linewidth}
        \centering
        \includegraphics[width=1.0\linewidth]{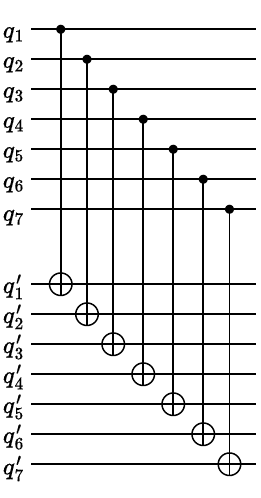}
        \caption{CNOT gadget}
        \label{fig:713-cc-cnot}
    \end{subfigure}
    \hfill
    \begin{subfigure}[b]{0.5\linewidth}
        \centering
        \includegraphics[width=1.0\linewidth]{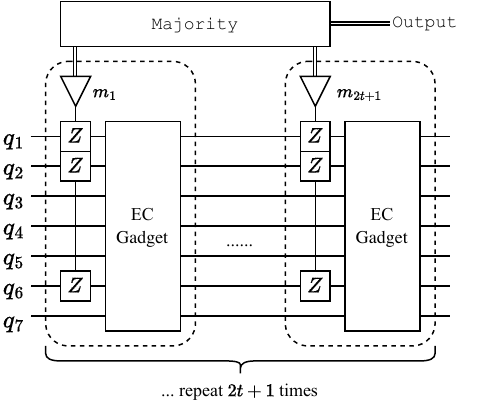}
        \caption{Measurement gadget}
        \label{fig:713-cc-meas}
    \end{subfigure}
    \begin{subfigure}[b]{\linewidth}
        \vspace{10pt}
        \centering
        \includegraphics[width=0.9\linewidth]{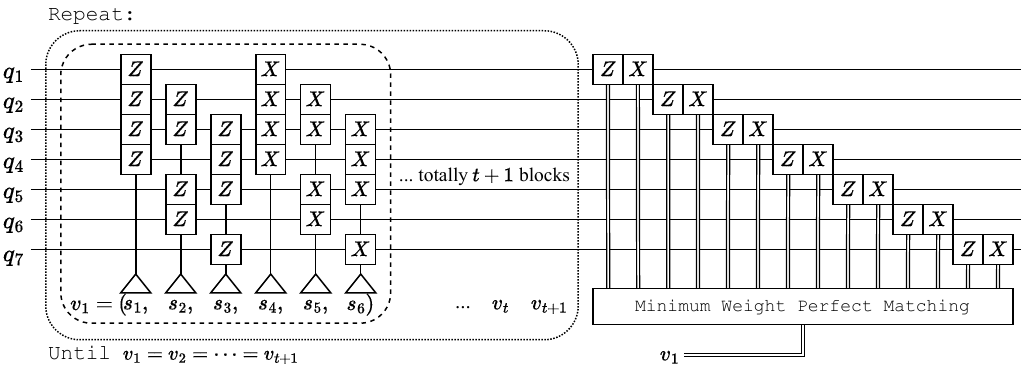}
        \caption{Error correction gadget}
        \label{fig:713-cc-ec}
    \end{subfigure}
    \caption{Fault-Tolerant gadgets of the [[7,1,3]] color code}
    \label{fig:713-gadgets}
    \vspace{-20mm}
\end{figure}



\end{document}